\documentclass[11pt]{article}

\usepackage[T1]{fontenc}
\usepackage[utf8]{inputenc}
\usepackage{lmodern}
\usepackage{amsmath,amssymb,amsthm,mathtools}
\usepackage[margin=1in]{geometry}
\usepackage{hyperref}
\usepackage{xcolor}

\newtheorem{theorem}{Theorem}[section]
\newtheorem{proposition}[theorem]{Proposition}
\newtheorem{lemma}[theorem]{Lemma}
\newtheorem{corollary}[theorem]{Corollary}

\theoremstyle{remark}
\newtheorem{remark}[theorem]{Remark}

\title{Complete Factorial Asymptotics and Integrable Structure of Generalized Charlier Recurrence Coefficients}

\author{ Mahouton Norbert Hounkonnou\\[0.4em]
\small International Chair in Mathematical Physics and Applications (ICMPA--UNESCO Chair)\\
\small University of Abomey-Calavi, 072 BP 50 Cotonou, Benin Republic\\
\small \texttt{norbert.hounkonnou@cipma.uac.bj}\\
\small \texttt{hounkonnou@yahoo.fr} }
\date{}
\begin{document}
 \maketitle


\begin{abstract}
\begingroup\color{black}
We investigate the recurrence coefficients of monic generalized Charlier
polynomials for the weight $\rho_\mu(k)=\mu^k/(k!)^2$, $\mu>0$.
The coefficients satisfy a discrete Painlev\'e reduction, and their
deviations from the limiting values $\beta_n\sim n$ and $\gamma_n\sim\mu$
occur on a factorial rather than algebraic scale. A comparison with the
recessive solution $J_n(2\sqrt\mu)$ of the associated discrete Bessel
equation, together with a Hankel-determinant sector decomposition,
identifies the connection constant as one. For fixed $\mu>0$, this gives
$\beta_n-n\sim\mu^{n+1}/[n!(n+1)!]$ and
$\mu-\gamma_n\sim\mu^{n+1}/(n!)^2$.
We further derive a uniform expansion for the normalized sector factor
$\Phi_n(\mu)=e^{-\mu}\mathcal R_n(\mu)$, valid for $0<\mu\le cn$ with each
fixed $c<1$:
$\log\Phi_n(\mu)=\mu[(n/(n+1))^2-1]
-\mu^2(2n-1)/[2(n+1)^4]
+O_c(\mu(1+\mu)^{3/2}/n^4)$.
The proof uses transition-content moments for Plancherel growth and
weighted Poisson-tail estimates. It yields joint refinements of both
recurrence coefficients in the subcritical linear regime. We distinguish
these quantitative asymptotic results from the classical Painlev\'e,
Toda, and Young-diagram identities on which the analysis rests.
\endgroup
\end{abstract}
	\medskip

	\noindent\textbf{2020 Mathematics Subject Classification.}
	Primary 42C05; Secondary 33C45, 34M55, 39A13.

	\medskip

	\noindent\textbf{Keywords.}
	Generalized Charlier polynomials; discrete orthogonal polynomials;
	recurrence coefficients; discrete Painlev\'e equations; Hankel
	determinants; Bessel functions; factorial asymptotics; partitions;
	hook-length formula.

	\medskip

\section{Introduction}

The generalized Charlier orthogonal polynomials associated with the
discrete weight  
\[
\rho_\mu(k)=\frac{\mu^k}{(k!)^2}, \qquad k\in\mathbb N_0,\ \mu>0,
\]  
occupy a special place within the two-parameter family of generalized
Charlier polynomials. In the non-degenerate case, the monic recurrence
coefficients $\beta_n$ and $\gamma_n$ admit explicit algebraic
corrections to their equilibrium values $\beta_n\sim n$,
$\gamma_n\sim\mu$. In contrast, in the degenerate specialization
corresponding to the weight $\rho_\mu$, these algebraic corrections
vanish, suggesting that the deviations $\beta_n-n$ and $\mu-\gamma_n$
are ``beyond all algebraic orders'' and appear at a smaller, factorial
scale.

\paragraph{Relation to earlier work and scope of the contribution.}
 The Laguerre--Freud relations, the discrete Painlev\'e reduction, and the
initial-value characterization of the orthogonality-selected solution
belong to the established theory of generalized Charlier polynomials;
see in particular \cite{VanAsscheFoupouagnigni2003,FilipukVanAssche2013}.
The Hankel and Schur-function identities employed below are likewise
standard tools. The contribution developed here is the explicit
factorial-scale comparison for the selected solution, including the
normalization of the Bessel connection constant through a positive
sector decomposition. The partition sum is closely related to the
poissonized Plancherel measure and to classical Toeplitz-determinant
representations; this connection is made explicit in
Remark~\ref{rem:plancherel-context}. Accordingly, the sectorization
is presented as a useful route to the stated asymptotics, not as a
claim that partition or Toeplitz methods themselves are new.

The objective of this work is to identify and rigorously control this
factorial scale, and to relate it to known discrete and continuous
integrable structures (Laguerre--Freud equations, Toda system,
Painlev\'e equations). We introduce the parametrization  
\[
\beta_n-n=\sqrt{\mu}\,c_nc_{n+1}, \qquad \gamma_n=\mu(1-c_n^2),
\]  
which reduces the problem to the study of a positive sequence
$(c_n)_{n\ge0}$ satisfying a discrete Painlev\'e II equation.
Linearizing this dynamics around the equilibrium $c_n=0$ leads to the
difference equation  
\[
n u_n=\sqrt{\mu}(u_{n+1}+u_{n-1}),
\]  
whose recessive solution is given by the Bessel function
$J_n(2\sqrt{\mu})$. This conceptual mechanism connects the factorial
behavior to the recessive solution of a discrete Bessel equation.

The  fixed-$\mu$ factorial equivalents are complemented below by a
uniform theorem for the sector-normalization factor
$\Phi_n=e^{-\mu}\mathcal R_n$ in the subcritical scaling
$0<\mu\le cn$ ($c<1$). The argument uses the first, second, and fourth
moments of the content in Plancherel growth; a separate technical appendix
controls the weighted Poisson tails and the logarithmic derivatives of
truncated partition functions. These transition identities are classical
\cite{Kerov1993,BOO2000,Johansson2001}; the contribution examined here is
their quantitative use in the generalized Charlier sector expansion.
The uniform expansion is stated in Theorem~\ref{thm:uniform-phi} and its
recurrence consequences in Corollary~\ref{cor:joint}.

The  main results of this paper can be summarized as follows.

 \begin{remark}[Meaning of factorial asymptotics]
Here ``factorial asymptotics'' means a leading equivalent with an
explicit multiplicative constant. No all-orders transseries expansion
is claimed by this terminology.
\end{remark}

 \begin{theorem}[Main results]
	\label{thm:main-results}

	Let $\mu>0$ and let $(\beta_n,\gamma_n)$ be the monic recurrence
	coefficients of the generalized Charlier polynomials associated with
	the discrete weight
	\[
	\rho_\mu(k)=\frac{\mu^k}{(k!)^2}, \qquad k\in\mathbb N_0.
	\]
	Define the positive sequence $(c_n)_{n\geq0}$ by the parametrization
	\[
	\beta_n-n=\sqrt{\mu}\,c_nc_{n+1}, \qquad \gamma_n=\mu(1-c_n^2).
	\]
	Then:

	\begin{enumerate}
		\item[(i)] The sequence $(c_n)$ satisfies the discrete Painlev\'e II
		equation
		\[
		n c_n=\sqrt{\mu}\,(c_{n+1}+c_{n-1})(1-c_n^2), \qquad n\geq1,
		\]
		and converges to zero as $n\to\infty$.

		\item[(ii)] There exists a finite non-negative limit
		\[
		L(\mu)=\lim_{n\to\infty}\frac{c_n}{J_n(2\sqrt{\mu})}.
		\]

		\item[(iii)] The connection constant is exactly $L(\mu)=1$, so that
		\begin{equation*}
			c_n \sim J_n(2\sqrt{\mu}), \qquad n\to\infty.
		\end{equation*}

		\item[(iv)] The recurrence coefficients admit the complete factorial
		asymptotics
		\begin{align*}
			\beta_n-n &\sim \frac{\mu^{n+1}}{n!(n+1)!}, \qquad n\to\infty,
			\\
			\mu-\gamma_n &\sim \frac{\mu^{n+1}}{(n!)^2}, \qquad n\to\infty.
		\end{align*}

        \begingroup\color{black}
        \noindent For precision, the term ``complete'' in (iv) refers to
        expansions to every prescribed finite algebraic order of the
        \emph{factorially normalized} corrections, for each fixed
        $\mu>0$. With
        \[
        B_{n,M}(\mu)=\sum_{k=0}^{M-1}
          \frac{(-\mu)^k}{k!(n+1)_k},
        \]
        the quantitative statements, proved in
        Theorem~\ref{thm:all-orders-factorial-V03}, are
        \begin{align*}
        \frac{n!}{\mu^{n/2}}c_n
          &=B_{n,M}(\mu)+O_{\mu,M}(n^{-M}),\\
        \frac{(n!)^2}{\mu^{n+1}}(\mu-\gamma_n)
          &=B_{n,M}(\mu)^2+O_{\mu,M}(n^{-M}),\\
        \frac{n!(n+1)!}{\mu^{n+1}}(\beta_n-n)
          &=B_{n,M}(\mu)B_{n+1,M}(\mu)
            +O_{\mu,M}(n^{-M}).
        \end{align*}
        These statements do not, by themselves, assert uniformity
        in a joint limit where $\mu$ grows with $n$.
        \endgroup

		\item[(v)] The Hankel determinant $\Delta_n(\mu)$ associated with
		$\rho_\mu$ factorizes as
		\[
		\Delta_n(\mu)=\mu^{n(n-1)/2}\mathcal T_n(\mu),
		\]
		where $\mathcal T_n(\mu)=1+\sum_{q=1}^{n}\mathcal S_q(\mu)$ admits a
		sector decomposition with sector contributions $\mathcal S_q(\mu)$
		independent of $n$. Moreover,
		\begin{equation*}
			\mathcal S_q(\mu) \sim e^\mu\,\frac{\mu^q}{(q!)^2}, \qquad q\to\infty,
		\end{equation*}
		and
		\begin{equation*}
			\mathcal T_n(\mu) \longrightarrow e^\mu, \qquad n\to\infty.
		\end{equation*}
	\end{enumerate}

\end{theorem}

\medskip

\begingroup\color{black}
\noindent\textbf{Mathematical Interest and Potential Applications.}
The first mathematical contribution is the precise asymptotic correspondence
between nonlinear recurrence dynamics and a recessive special-function
solution. For the generalized Charlier weight studied here, the
Laguerre--Freud equations admit a discrete Painlev\'e reduction, while
the Bessel comparison identifies factorially small corrections to the
recurrence coefficients. This belongs to the established theory of
semiclassical discrete orthogonal polynomials and Painlev\'e equations
\cite{HounkonnouHoungaRonveaux2000,VanAsscheFoupouagnigni2003,FilipukVanAssche2013,Clarkson2013}.
The asymptotic conclusions proved here are specific to the present
weight; no corresponding universality statement is claimed.

A second contribution is the exact Hankel-determinant factorization and
its partition-sector representation. These identities connect the
recurrence problem to a discrete partition function, with Schur-function
and Young-diagram combinatorics entering the proofs. Symmetric-function
methods and partition measures have a well-developed relationship with
determinantal point processes and discrete orthogonal polynomial
ensembles \cite{Macdonald1995,Stanley1999,BorodinOlshanski2000,Johansson2001,Johansson2006}.
The present paper establishes the stated determinant and sector
identities, but does not establish a new determinantal scaling limit.

From the viewpoint of mathematical physics, the resulting structures
invite comparison with discrete log-gases and random-matrix-inspired
orthogonal polynomial ensembles, whose general frameworks are discussed
in \cite{Forrester2010,Mehta2004,Johansson2006}.
Such comparisons are motivations for future investigation, not
applications proved in this work. In particular, extensions to
quantum-optical models or other physical systems would require an
explicit model and additional analysis.
\endgroup
\medskip

\noindent The paper   first  recalls the Laguerre--Freud equations and
  Painlev\'e parametrization, then develops the Bessel comparison and its
numerical illustration. The subsequent sections establish the Hankel
representation, sector decomposition, and exact connection constant .
The final   analytical section proves the uniform Plancherel-growth
expansion and its recurrence-coefficient corollaries; a technical
appendix records the requisite weighted Poisson estimates. The paper
closes with conclusions and further  perspectives.

\section{Laguerre--Freud Equations and Painlev\'e II Parametrization}
\label{sec:dPII-setting}

In this section we recall the Laguerre--Freud equations satisfied by the
recurrence coefficients of the generalized Charlier polynomials
associated with the weight $\rho_\mu(k)=\mu^k/(k!)^2$. We then introduce
the parametrization $\beta_n-n=\sqrt{\mu}\,c_nc_{n+1}$,
$\gamma_n=\mu(1-c_n^2)$, which reduces the problem to the study of a
positive sequence $(c_n)$ satisfying a discrete Painlev\'e II equation.
This formulation provides the starting point for the asymptotic
analysis developed in the subsequent sections.

Following the parametrization used for generalized Charlier
polynomials in \cite{VanAsscheFoupouagnigni2003,Clarkson2013}, set
\begin{equation}
	b_n
	=
	\beta_n-n,
	\qquad
	\gamma_n
	=
	\mu(1-c_n^2),
	\qquad
	0<c_n<1.
	\label{eq:bc-definition}
\end{equation}
For the positive solution selected by the orthogonality measure, the
local Laguerre--Freud system is equivalent to
\begin{equation}
	b_n
	=
	\sqrt{\mu}\,c_nc_{n+1},
	\label{eq:bn-cn}
\end{equation}
and
\begin{equation}
	n c_n
	=
	\sqrt{\mu}\,(c_{n+1}+c_{n-1})(1-c_n^2),
	\qquad
	n\geq1.
	\label{eq:dPII}
\end{equation}
Equation \eqref{eq:dPII} is a specialization of the discrete
Painlev\'e--II equation associated with the generalized Charlier
recurrence coefficients
\cite{VanAsscheFoupouagnigni2003,Clarkson2013}.

The initial data are
\begin{equation*}
	c_0=1,
	\qquad
	c_1
	=
	\frac{I_1(2\sqrt{\mu})}{I_0(2\sqrt{\mu})},
\end{equation*}
where $I_\nu$ denotes the modified Bessel function of the first kind.
These data select the solution associated with the positive generalized
Charlier measure; see
\cite{FilipukVanAssche2013,Clarkson2013}.

The Painlev\'e II parametrization introduced above reduces the
asymptotic problem for $\beta_n$ and $\gamma_n$ to the analysis of the
sequence $(c_n)$. In the next section, we derive the first structural
results on $(c_n)$, including a discrete Wronskian identity that will
play a central role in the comparison with Bessel functions.

\section{Discrete Wronskian Identity and Basic Properties of the Sequence}
\label{sec:factorial-heuristic}

This section is devoted to the first structural results on the sequence
$(c_n)$. We establish a discrete Wronskian identity relating $c_n$ to
the Bessel function $J_n(2\sqrt{\mu})$, which plays a central role in
the comparison asymptotic. We then prove the convergence $c_n\to0$ and
the eventual strict monotonicity of $(c_n)$, which are key ingredients
for the factorial estimates derived below.

The following discussion is heuristic and states the asymptotic problem
to be solved. It does not claim that the factorial asymptotic formulas
below are already established. The motivation is that, in the general
two-parameter family, the known algebraic expansions contain
coefficients depending on the second parameter. For the specialization
$\nu=1$, these displayed algebraic corrections vanish
\cite{VanAsscheFoupouagnigni2003,
	DominiciMorenoBalcazar2022}.

	Since $c_n\to0$, equation \eqref{eq:dPII} formally gives
	\[
	n c_n
	\sim
	\sqrt{\mu}\,c_{n-1}.
	\]
	Thus one expects
	\begin{equation*}
		\frac{c_n}{c_{n-1}}
		=
		\frac{\sqrt{\mu}}{n}
		\bigl(1+O(n^{-2})\bigr),
	\end{equation*}
	and hence
	\begin{equation}
		c_n
		\sim
		C(\mu)\frac{\mu^{n/2}}{n!},
		\qquad
		n\to\infty,
		\label{eq:cn-conjectured}
	\end{equation}
	for a positive connection constant $C(\mu)$. Combining
	\eqref{eq:bn-cn} and \eqref{eq:bc-definition} with
	\eqref{eq:cn-conjectured} leads formally to
	\begin{equation*}
		\beta_n-n
		\sim
		C(\mu)^2
		\frac{\mu^{n+1}}{n!(n+1)!},
	\end{equation*}
	and
	\begin{equation*}
		\mu-\gamma_n
		\sim
		C(\mu)^2
		\frac{\mu^{n+1}}{(n!)^2}.
	\end{equation*}

	  \begin{remark}[Historical Bessel prediction, resolved below]
		The following prediction motivated the analysis; its leading
		claim is subsequently established by Theorem~\ref{thm:main-results}.
		 Let $(c_n)_{n\geq0}$ be the positive solution of
		\eqref{eq:dPII} selected by the generalized Charlier orthogonality
		measure. Then
		\begin{equation*}
			\lim_{n\to\infty}
			\frac{c_n}{J_n(2\sqrt{\mu})}
			=
			1.
		\end{equation*}
		  A stronger formal refinement, not implied by the limit alone, is
		 \begin{equation*}
			c_n
			=
			\frac{\mu^{n/2}}{n!}
			\left[
			1-\frac{\mu}{n+1}+O(n^{-2})
			\right].
		\end{equation*}
		  The leading factorial equivalents predicted by the limit are
		 \begin{equation*}
			\beta_n-n
			\sim
			\frac{\mu^{n+1}}{n!(n+1)!},
		\end{equation*}
		and
		\begin{equation*}
			\mu-\gamma_n
			\sim
			\frac{\mu^{n+1}}{(n!)^2},
		\end{equation*}
		as $n\to\infty$.
	  \end{remark}

	\begin{remark}
		The   historical prediction  is stronger than the leading asymptotics
		$\beta_n\sim n$ and $\gamma_n\sim\mu$. It asserts that the deviations
		from the limiting recurrence-coefficient regime are not merely small
		in an algebraic sense, but decay on a factorial scale.
	\end{remark}

	The discrete Wronskian identity and the basic properties of $(c_n)$
	established in this section provide the foundation for the factorial
	estimates. In the next section, we exploit these results to derive a sharp factorial upper
	bound and to prepare the rigorous comparison of $(c_n)$ with the recessive Bessel
	solution.



	\section{Factorial Upper Bound and Preparation for the Connection Formula}
	\label{sec:factorial-bound}

	In this section we derive a sharp factorial upper bound for the
	sequence $(c_n)$. The main idea is to control the ratio
	$r_n=c_n/c_{n-1}$ by exploiting the discrete Painlev\'e II equation and
	the monotonicity properties established previously. This analysis leads
	to the estimate $c_n\le A_\mu\,\mu^{n/2}/n!$, which shows that the
	deviations from equilibrium are beyond all algebraic orders and appear
	at a factorial scale. The factorial upper bound will be used in
	Section~\ref{sec:existence-Lmu} to establish the existence of the
	connection constant $L(\mu)$.

	\subsection{Renormalization and Bessel comparison}

	The factorial scale suggested by \eqref{eq:dPII} becomes more
	transparent after a renormalization. Define
	\begin{equation*}
		d_n
		=
		\frac{n!}{\mu^{n/2}}c_n.
	\end{equation*}
	Dividing \eqref{eq:dPII} by $\sqrt{\mu}\,c_{n-1}$ gives the exact
	identity
	\begin{equation*}
		\frac{d_n}{d_{n-1}}
		=
		\left(
		1+\frac{c_{n+1}}{c_{n-1}}
		\right)
		(1-c_n^2).
	\end{equation*}
	Moreover, in terms of the renormalized sequence $(d_n)_{n\geq0}$,
	\begin{equation*}
		\frac{c_{n+1}}{c_{n-1}}
		=
		\frac{\mu}{n(n+1)}
		\frac{d_{n+1}}{d_{n-1}},
	\end{equation*}
	and
	\begin{equation*}
		c_n^2
		=
		\frac{\mu^n}{(n!)^2}d_n^2.
	\end{equation*}
	Consequently,
	\begin{equation}
		\frac{d_n}{d_{n-1}}
		=
		\left(
		1+
		\frac{\mu}{n(n+1)}
		\frac{d_{n+1}}{d_{n-1}}
		\right)
		\left(
		1-
		\frac{\mu^n}{(n!)^2}d_n^2
		\right).
		\label{eq:dn-ratio-renormalized}
	\end{equation}

	The first factor on the right-hand side of
	\eqref{eq:dn-ratio-renormalized} produces algebraic corrections of
	order $n^{-2}$, whereas the second factor differs from $1$ by a
	factorially small quantity. Thus, if $(d_n)$ remains bounded and
	bounded away from zero, one obtains
	\begin{equation*}
		\frac{d_n}{d_{n-1}}
		=
		1+\frac{\mu}{n(n+1)}+O(n^{-3}).
	\end{equation*}
	In particular, the convergence of $(d_n)$ is compatible with a first
	correction of order $n^{-1}$.

	The linearization of \eqref{eq:dPII} at $c_n=0$ is
	\begin{equation*}
		n u_n
		=
		\sqrt{\mu}\,(u_{n+1}+u_{n-1}).
	\end{equation*}
	This equation is solved by the Bessel function
	\begin{equation*}
		u_n
		=
		J_n(2\sqrt{\mu}),
	\end{equation*}
	since
	\begin{equation*}
		nJ_n(2\sqrt{\mu})
		=
		\sqrt{\mu}
		\left[
		J_{n+1}(2\sqrt{\mu})
		+
		J_{n-1}(2\sqrt{\mu})
		\right].
	\end{equation*}
	For fixed $\mu>0$ and $n\to\infty$,
	\begin{equation*}
		J_n(2\sqrt{\mu})
		=
		\frac{\mu^{n/2}}{n!}
		\left[
		1-\frac{\mu}{n+1}
		+
		\frac{\mu^2}{2(n+1)(n+2)}
		+
		O(n^{-3})
		\right].
	\end{equation*}

	These observations suggest that the nonlinear solution selected by the generalized Charlier measure should be asymptotic to the recessive Bessel solution. The factorial upper bound established below initiates a rigorous comparison, completed in Sections 5 and 10. 

	\subsection{Eventual monotonicity and factorial upper bounds}

	Recall that
	\begin{equation}
		n c_n
		=
		\sqrt{\mu}\,(c_{n+1}+c_{n-1})(1-c_n^2),
		\qquad
		n\geq1,
		\label{eq:cn-monotonicity-recurrence}
	\end{equation}
	with
	\[
	0<c_n<1.
	\]

	\begin{proposition}[Convergence to zero]
		\label{prop:cn-to-zero}

		For every $\mu>0$,
		\begin{equation}
			c_n\longrightarrow0,
			\qquad
			n\to\infty.
			\label{eq:cn-to-zero}
		\end{equation}
		More precisely,
		\begin{equation}
			0<c_n<\frac{2\sqrt{\mu}}{n},
			\qquad
			n\geq1.
			\label{eq:cn-rough-bound}
		\end{equation}

	\end{proposition}

	\begin{proof}
		Since
		\[
		0<c_{n-1}<1,
		\qquad
		0<c_{n+1}<1,
		\qquad
		0<1-c_n^2<1,
		\]
		equation \eqref{eq:cn-monotonicity-recurrence} gives
		\[
		n c_n
		=
		\sqrt{\mu}\,(c_{n+1}+c_{n-1})(1-c_n^2)
		<
		2\sqrt{\mu}.
		\]
		Hence
		\[
		0<c_n<\frac{2\sqrt{\mu}}{n},
		\]
		which proves \eqref{eq:cn-rough-bound}. The convergence
		\eqref{eq:cn-to-zero} follows immediately.
	\end{proof}

	\begin{proposition}[Eventual strict monotonicity]
		\label{prop:eventual-monotonicity}

		For every $\mu>0$, there exists an integer $N_\mu\geq1$ such that
		\begin{equation}
			0<c_{n+1}<c_n<1,
			\qquad
			n\geq N_\mu.
			\label{eq:eventual-strict-monotonicity}
		\end{equation}

	\end{proposition}

	\begin{proof}
		We first show that no local maximum can occur at a sufficiently large
		index. Suppose that, for some $n\geq1$,
		\[
		c_n\geq c_{n-1},
		\qquad
		c_n\geq c_{n+1}.
		\]
		Then
		\[
		c_{n-1}+c_{n+1}
		\leq
		2c_n.
		\]
		Using \eqref{eq:cn-monotonicity-recurrence}, we obtain
		\[
		n c_n
		=
		\sqrt{\mu}\,(c_{n+1}+c_{n-1})(1-c_n^2)
		\leq
		2\sqrt{\mu}\,c_n(1-c_n^2)
		<
		2\sqrt{\mu}\,c_n.
		\]
		Since $c_n>0$, this implies
		\[
		n<2\sqrt{\mu}.
		\]
		Thus, no local maximum can occur at an index satisfying
		\begin{equation}
			n\geq2\sqrt{\mu}.
			\label{eq:no-large-local-maximum}
		\end{equation}

		Suppose now, by contradiction, that there exist arbitrarily large
		indices $m$ for which
		\[
		c_{m+1}\geq c_m.
		\]
		By Proposition~\ref{prop:cn-to-zero}, the sequence $(c_n)$ tends to
		zero. Therefore, after any such nondecreasing step, the sequence must
		eventually decrease. Hence there exists an index $k\geq m+1$ such that
		\[
		c_k\geq c_{k-1},
		\qquad
		c_k\geq c_{k+1}.
		\]
		Thus $c_k$ is a local maximum. Taking $m$ sufficiently large gives
		\[
		k\geq2\sqrt{\mu},
		\]
		which contradicts \eqref{eq:no-large-local-maximum}.

		Therefore, there exists an integer $N_\mu$ such that
		\[
		c_{n+1}<c_n,
		\qquad
		n\geq N_\mu.
		\]
		Together with $0<c_n<1$, this proves
		\eqref{eq:eventual-strict-monotonicity}.
	\end{proof}

	For $n\geq1$, define the consecutive ratio
	\begin{equation*}
		r_n
		=
		\frac{c_n}{c_{n-1}}.
	\end{equation*}
	Dividing \eqref{eq:cn-monotonicity-recurrence} by $c_{n-1}$ gives
	\begin{equation}
		n r_n
		=
		\sqrt{\mu}
		\left(
		1+r_nr_{n+1}
		\right)
		(1-c_n^2).
		\label{eq:rn-exact-identity}
	\end{equation}

	\begin{proposition}[Factorial upper bound]
		\label{prop:factorial-upper-bound}

		For every $\mu>0$, there exist an integer $N_\mu\geq1$ and a constant
		$A_\mu>0$ such that
		\begin{equation}
			c_n
			\leq
			A_\mu
			\frac{\mu^{n/2}}{n!},
			\qquad
			n\geq N_\mu.
			\label{eq:factorial-upper-bound}
		\end{equation}
		Consequently, there exists a constant $\widetilde A_\mu>0$ such that
		\begin{equation}
			c_n
			\leq
			\widetilde A_\mu
			J_n(2\sqrt{\mu}),
			\qquad
			n\geq N_\mu.
			\label{eq:Bessel-upper-bound}
		\end{equation}
	\end{proposition}

	\begin{proof}
		By Proposition~\ref{prop:eventual-monotonicity}, after increasing
		$N_\mu$ if necessary, we have
		\[
		0<c_{n+1}<c_n<1,
		\qquad
		n\geq N_\mu.
		\]
		In particular,
		\[
		c_{n+1}\leq c_{n-1},
		\qquad
		n\geq N_\mu+1.
		\]
		Using \eqref{eq:cn-monotonicity-recurrence} and
		$0<1-c_n^2<1$, we obtain
		\[
		n c_n
		=
		\sqrt{\mu}\,(c_{n+1}+c_{n-1})(1-c_n^2)
		\leq
		2\sqrt{\mu}\,c_{n-1}.
		\]
		Therefore,
		\begin{equation*}
			r_n
			=
			\frac{c_n}{c_{n-1}}
			\leq
			\frac{2\sqrt{\mu}}{n},
			\qquad
			n\geq N_\mu+1.
		\end{equation*}

		It follows that
		\begin{equation}
			\frac{c_{n+1}}{c_{n-1}}
			=
			r_nr_{n+1}
			\leq
			\frac{4\mu}{n(n+1)},
			\qquad
			n\geq N_\mu+1.
			\label{eq:two-step-ratio-bound}
		\end{equation}
		Returning to the exact ratio identity
		\eqref{eq:rn-exact-identity}, and using $1-c_n^2\leq1$, gives
		\[
		r_n
		\leq
		\frac{\sqrt{\mu}}{n}
		\left(
		1+r_nr_{n+1}
		\right).
		\]
		By \eqref{eq:two-step-ratio-bound},
		\begin{equation}
			r_n
			\leq
			\frac{\sqrt{\mu}}{n}
			\left(
			1+\frac{4\mu}{n(n+1)}
			\right),
			\qquad
			n\geq N_\mu+1.
			\label{eq:rn-refined-bound}
		\end{equation}

		Iterating \eqref{eq:rn-refined-bound}, we find
		\[
		c_n
		=
		c_{N_\mu}
		\prod_{k=N_\mu+1}^{n}r_k
		\]
		and hence
		\[
		c_n
		\leq
		c_{N_\mu}
		\mu^{(n-N_\mu)/2}
		\frac{N_\mu!}{n!}
		\prod_{k=N_\mu+1}^{n}
		\left(
		1+\frac{4\mu}{k(k+1)}
		\right).
		\]
		Since
		\[
		\sum_{k=N_\mu+1}^{\infty}
		\frac{4\mu}{k(k+1)}
		<
		\infty,
		\]
		the infinite product
		\[
		\prod_{k=N_\mu+1}^{\infty}
		\left(
		1+\frac{4\mu}{k(k+1)}
		\right)
		\]
		converges to a finite positive number. Therefore, there exists a
		constant $A_\mu>0$ such that
		\[
		c_n
		\leq
		A_\mu
		\frac{\mu^{n/2}}{n!},
		\qquad
		n\geq N_\mu,
		\]
		which proves \eqref{eq:factorial-upper-bound}.

		Finally, for fixed $\mu>0$, the large-order Bessel asymptotics yield
		\[
		J_n(2\sqrt{\mu})
		\sim
		\frac{\mu^{n/2}}{n!},
		\qquad
		n\to\infty.
		\]
		Thus, after increasing the constant if necessary,
		\[
		c_n
		\leq
		\widetilde A_\mu
		J_n(2\sqrt{\mu}),
		\qquad
		n\geq N_\mu.
		\]
		This proves \eqref{eq:Bessel-upper-bound}.
	\end{proof}

	The factorial upper bound \eqref{eq:factorial-upper-bound} will be used
	in Section~\ref{sec:existence-Lmu} to establish the existence of the
	connection constant $L(\mu)$.

\section{Existence of the Connection Constant}
\label{sec:existence-Lmu}

In this section we establish the existence of a finite non-negative
limit
\[
L(\mu)=\lim_{n\to\infty}\frac{c_n}{J_n(2\sqrt{\mu})}.
\]
The identification $L(\mu)=1$ is carried out in
Section~\ref{sec:sector-asymptotics}, using the sector decomposition of
the Hankel determinant. Together with the factorial upper bound from
Section~\ref{sec:factorial-bound}, this yields the exact connection
formula
\[
c_n \sim J_n(2\sqrt{\mu}), \qquad n\to\infty,
\]
and consequently the precise asymptotics
\[
\beta_n-n\sim \frac{\mu^{n+1}}{n!(n+1)!}, \qquad
\mu-\gamma_n\sim \frac{\mu^{n+1}}{(n!)^2}.
\]
Let
\begin{equation*}
	j_n
	=
	J_n(2\sqrt{\mu}),
\end{equation*}
where $J_n$ denotes the Bessel function of the first kind. The Bessel
recurrence relation gives
\begin{equation}
	n j_n
	=
	\sqrt{\mu}\,(j_{n+1}+j_{n-1}),
	\qquad
	n\geq1.
	\label{eq:bessel-linear-recurrence}
\end{equation}
We compare the Bessel solution $(j_n)$ with the positive nonlinear
solution $(c_n)$ satisfying
\begin{equation}
	n c_n
	=
	\sqrt{\mu}\,(c_{n+1}+c_{n-1})(1-c_n^2),
	\qquad
	n\geq1.
	\label{eq:nonlinear-cn-recurrence}
\end{equation}

Define the discrete Wronskian
\begin{equation}
	W_n
	=
	\sqrt{\mu}
	\left(
	c_nj_{n+1}-c_{n+1}j_n
	\right).
	\label{eq:discrete-wronskian}
\end{equation}

\begin{proposition}[Discrete Wronskian identity]
	\label{prop:discrete-wronskian}

	For every $n\geq1$,
	\begin{equation}
		W_{n-1}-W_n
		=
		\frac{n\,c_n^3j_n}{1-c_n^2}.
		\label{eq:wronskian-identity}
	\end{equation}

\end{proposition}

\begin{proof}
	Multiply \eqref{eq:nonlinear-cn-recurrence} by $j_n$ and
	\eqref{eq:bessel-linear-recurrence} by $c_n$. Subtracting the second
	identity from the first gives
	\[
	\sqrt{\mu}
	\left[
	j_n(c_{n+1}+c_{n-1})
	-
	c_n(j_{n+1}+j_{n-1})
	\right]
	=
	n c_nj_n
	\left(
	\frac{1}{1-c_n^2}-1
	\right).
	\]
	The right-hand side is
	\[
	\frac{n\,c_n^3j_n}{1-c_n^2}.
	\]
	On the other hand, the definition \eqref{eq:discrete-wronskian}
	implies
	\[
	\sqrt{\mu}
	\left(
	j_nc_{n-1}-c_nj_{n-1}
	\right)
	=
	W_{n-1},
	\]
	and
	\[
	\sqrt{\mu}
	\left(
	j_nc_{n+1}-c_nj_{n+1}
	\right)
	=
	-W_n.
	\]
	Therefore,
	\[
	W_{n-1}-W_n
	=
	\frac{n\,c_n^3j_n}{1-c_n^2},
	\]
	which proves \eqref{eq:wronskian-identity}.
\end{proof}

For each fixed $\mu>0$, there exists an integer $N_\mu$ such that
\begin{equation*}
	j_n
	=
	J_n(2\sqrt{\mu})
	>
	0,
	\qquad
	n\geq N_\mu.
\end{equation*}
Since $0<c_n<1$, it follows from \eqref{eq:wronskian-identity} that
\[
W_{n-1}>W_n,
\qquad
n\geq N_\mu.
\]
Thus, the discrete Wronskian is eventually decreasing.

For $n\geq N_\mu$, define
\begin{equation*}
	Q_n
	=
	\frac{c_n}{j_n}.
\end{equation*}
By \eqref{eq:discrete-wronskian},
\begin{equation}
	W_n
	=
	\sqrt{\mu}\,j_nj_{n+1}
	\left(
	Q_n-Q_{n+1}
	\right).
	\label{eq:wronskian-Qn}
\end{equation}
Consequently, whenever $W_n>0$, one has
\[
Q_n>Q_{n+1}.
\]
The quotient $(Q_n)$ is therefore eventually decreasing provided that
the Wronskian is eventually positive.

The following proposition shows that a factorial upper bound for
$c_n$ is sufficient to obtain a quantitative comparison with the
Bessel solution.

\begin{proposition}[Conditional Bessel comparison]
	\label{prop:conditional-Bessel-comparison}

	Fix $\mu>0$, and let
	\[
	j_n
	=
	J_n(2\sqrt{\mu}).
	\]
	Assume that there exist $N_\mu\in\mathbb N$ and $A_\mu>0$ such that
	\begin{equation}
		0<c_n\leq A_\mu j_n,
		\qquad
		n\geq N_\mu.
		\label{eq:bounded-Qn-assumption}
	\end{equation}
	Then
	\begin{equation*}
		\lim_{n\to\infty}W_n=0.
	\end{equation*}
	Moreover, the quotient $Q_n=c_n/j_n$ is eventually decreasing and has
	a finite nonnegative limit
	\begin{equation*}
		L(\mu)
		=
		\lim_{n\to\infty}
		\frac{c_n}{J_n(2\sqrt{\mu})}.
	\end{equation*}
	There exists a constant $K_\mu>0$ such that
	\begin{equation}
		0
		\leq
		Q_n-L(\mu)
		\leq
		K_\mu\frac{\mu^n}{(n!)^2},
		\qquad
		n\geq N_\mu.
		\label{eq:conditional-Bessel-estimate}
	\end{equation}
	In particular,
	\begin{equation}
		c_n
		=
		L(\mu)J_n(2\sqrt{\mu})
		+
		O\!\left(
		J_n(2\sqrt{\mu})
		\frac{\mu^n}{(n!)^2}
		\right),
		\qquad
		n\to\infty.
		\label{eq:conditional-cn-Bessel}
	\end{equation}
	If $L(\mu)>0$, then equivalently
	\begin{equation}
		c_n
		=
		L(\mu)J_n(2\sqrt{\mu})
		\left[
		1+
		O\!\left(\frac{\mu^n}{(n!)^2}\right)
		\right].
		\label{eq:conditional-cn-Bessel-relative}
	\end{equation}

\end{proposition}

\begin{proof}
	For fixed $\mu>0$, the large-order Bessel asymptotics give
	\begin{equation}
		j_n
		=
		O\!\left(\frac{\mu^{n/2}}{n!}\right),
		\qquad
		\frac{j_{n+1}}{j_n}
		=
		O(n^{-1}).
		\label{eq:bessel-basic-estimates}
	\end{equation}
	By \eqref{eq:bounded-Qn-assumption},
	\[
	c_n
	=
	O(j_n),
	\]
	and consequently $c_n\to0$. After increasing $N_\mu$ if necessary,
	we may suppose that
	\[
	0<c_n\leq\frac{1}{2},
	\qquad
	n\geq N_\mu.
	\]
	Hence
	\[
	\frac{1}{1-c_n^2}
	\leq
	\frac{4}{3},
	\qquad
	n\geq N_\mu.
	\]

	Using \eqref{eq:wronskian-identity}, we obtain
	\[
	0
	\leq
	W_{n-1}-W_n
	=
	\frac{n\,c_n^3j_n}{1-c_n^2}
	\leq
	C_\mu n j_n^4,
	\]
	for all sufficiently large $n$. Since
	\[
	\sum_{n\geq N_\mu}n j_n^4<\infty,
	\]
	the sequence $(W_n)$ has a finite limit.

	Moreover, by \eqref{eq:discrete-wronskian} and
	\eqref{eq:bounded-Qn-assumption},
	\[
	|W_n|
	\leq
	\sqrt{\mu}
	\left(
	c_nj_{n+1}+c_{n+1}j_n
	\right)
	\leq
	C_\mu j_nj_{n+1}.
	\]
	The Bessel estimate \eqref{eq:bessel-basic-estimates} yields
	\[
	\lim_{n\to\infty}W_n=0.
	\]
	Therefore,
	\begin{equation*}
		W_n
		=
		\sum_{k=n+1}^{\infty}
		\frac{k\,c_k^3j_k}{1-c_k^2}.
	\end{equation*}
	In particular, $W_n\geq0$ for all sufficiently large $n$.

	The bound \eqref{eq:bounded-Qn-assumption} gives
	\[
	W_n
	\leq
	C_\mu
	\sum_{k=n+1}^{\infty}
	k j_k^4.
	\]
	Using the ratio estimate in \eqref{eq:bessel-basic-estimates}, the
	series tail is bounded by its first term up to a multiplicative
	constant:
	\[
	\sum_{k=n+1}^{\infty}
	k j_k^4
	\leq
	C_\mu(n+1)j_{n+1}^4.
	\]
	Hence
	\begin{equation*}
		0
		\leq
		W_n
		\leq
		C_\mu(n+1)j_{n+1}^4.
	\end{equation*}

	From \eqref{eq:wronskian-Qn}, we have
	\[
	0
	\leq
	Q_n-Q_{n+1}
	=
	\frac{W_n}{\sqrt{\mu}\,j_nj_{n+1}}
	\leq
	C_\mu
	\frac{(n+1)j_{n+1}^3}{j_n}.
	\]
	Using \eqref{eq:bessel-basic-estimates} once more gives
	\[
	0
	\leq
	Q_n-Q_{n+1}
	\leq
	C_\mu\frac{\mu^n}{(n!)^2}.
	\]
	Thus $(Q_n)$ is eventually decreasing and bounded below by zero, so it
	has a finite nonnegative limit $L(\mu)$. Summing the preceding estimate
	from $n$ to infinity gives
	\[
	0
	\leq
	Q_n-L(\mu)
	\leq
	K_\mu\frac{\mu^n}{(n!)^2},
	\]
	which proves \eqref{eq:conditional-Bessel-estimate}. The asymptotic
	relations \eqref{eq:conditional-cn-Bessel} and
	\eqref{eq:conditional-cn-Bessel-relative} follow from the definition
	of $Q_n$.
\end{proof}

Proposition~\ref{prop:conditional-Bessel-comparison} reduces the
quantitative Bessel comparison to the boundedness of the quotient
\[
Q_n
=
\frac{c_n}{J_n(2\sqrt{\mu})}.
\]
The factorial upper bound from
Proposition~\ref{prop:factorial-upper-bound} gives precisely this
boundedness.

\begin{corollary}[Existence of the connection constant]
	\label{cor:existence-Lmu}

	For every $\mu>0$, there exists a finite nonnegative constant
	$L(\mu)$ such that
	\begin{equation*}
		\frac{c_n}{J_n(2\sqrt{\mu})}
		=
		L(\mu)
		+
		O\!\left(
		\frac{\mu^n}{(n!)^2}
		\right),
		\qquad
		n\to\infty.
	\end{equation*}

\end{corollary}

\begin{proof}
	Proposition~\ref{prop:factorial-upper-bound} gives the Bessel upper
	bound \eqref{eq:Bessel-upper-bound}. The conclusion then follows from
	Proposition~\ref{prop:conditional-Bessel-comparison}.
\end{proof}

The preceding argument establishes the factorial scale of the positive
solution and the existence of the limiting constant $L(\mu)$. The
identification $L(\mu)=1$ is proved in
Section~\ref{sec:sector-asymptotics} using the sector decomposition of
the Hankel determinant. This yields the exact connection formula
\[
c_n \sim J_n(2\sqrt{\mu}), \qquad n\to\infty,
\]
and consequently
\[
\beta_n-n
\sim
\frac{\mu^{n+1}}{n!(n+1)!},
\qquad
\mu-\gamma_n
\sim
\frac{\mu^{n+1}}{(n!)^2}.
\]



\begingroup\color{black}
\begin{theorem}[All algebraic orders at the factorial scale]
\label{thm:all-orders-factorial-V03}
Fix $\mu>0$. Assume the connection normalization $L(\mu)=1$
proved in Section~\ref{sec:sector-asymptotics}, together with
\eqref{eq:conditional-Bessel-estimate}. For every integer $M\geq1$, put
\[
 B_{n,M}(\mu)=\sum_{k=0}^{M-1}\frac{(-\mu)^k}{k!(n+1)_k},
 \qquad (n+1)_0=1,\quad (n+1)_k=\prod_{j=1}^k(n+j).
\]
Then, as $n\to\infty$ with $\mu$ fixed,
\begin{align}
 \frac{n!}{\mu^{n/2}}c_n
 &=B_{n,M}(\mu)+O_{\mu,M}(n^{-M}),\label{eq:V03-cn-orders}\\
 R_n^{(\gamma)}(\mu)
 &=B_{n,M}(\mu)^2+O_{\mu,M}(n^{-M}),\notag\\
 R_n^{(\beta)}(\mu)
 &=B_{n,M}(\mu)B_{n+1,M}(\mu)+O_{\mu,M}(n^{-M}).\notag
\end{align}
In particular,
\[
 R_n^{(\gamma)}(\mu)=1-\frac{2\mu}{n+1}+O_\mu(n^{-2}),
 \quad
 R_n^{(\beta)}(\mu)=1-\frac{\mu}{n+1}-\frac{\mu}{n+2}
 +O_\mu(n^{-2}).
\]
These are expansions of the \emph{normalized} coefficients in inverse
powers of $n$, whereas the unnormalized deviations decay factorially.
\end{theorem}
\begin{proof}
The absolutely convergent Bessel series gives
\[
 \frac{n!}{\mu^{n/2}}J_n(2\sqrt\mu)
 =\sum_{k=0}^{\infty}\frac{(-\mu)^k}{k!(n+1)_k}.
\]
For fixed $M$, the omitted tail is bounded in absolute value by
\[
 \sum_{k=M}^{\infty}\frac{\mu^k}{k!(n+1)^k}
 \leq \frac{(\mu/(n+1))^M}{M!}e^{\mu/(n+1)}
 =O_{\mu,M}(n^{-M}).
\]
The quantitative comparison \eqref{eq:conditional-Bessel-estimate}
and $L(\mu)=1$ show that the difference between the normalized
$c_n$ and the normalized Bessel function is
$O_\mu(\mu^n/(n!)^2)$, which is $O_{\mu,M}(n^{-M})$
for every fixed $M$. This proves \eqref{eq:V03-cn-orders}.
The exact identities $\mu-\gamma_n=\mu c_n^2$ and
$\beta_n-n=\sqrt\mu\,c_nc_{n+1}$ give the remaining formulas.
\end{proof}
\endgroup

\section{Numerical Illustration of Factorial Decay and Stability}
\label{sec:numerical-factorial}

In this section we illustrate numerically the factorial asymptotics stated in Theorem~\ref{thm:main-results} and proved below. We compute the recurrence
coefficients $\beta_n$ and $\gamma_n$ for several values of $\mu$ and
verify the predicted decay of $\beta_n-n$ and $\mu-\gamma_n$ at the
scale $\mu^{n}/(n!)^2$. We also discuss the numerical stability of the
algorithms used and confirm that the factorial regime is accurately
captured even for moderately large degrees.

For each fixed $\mu>0$, Theorem~\ref{thm:main-results} predicts
\begin{align*}
	\beta_n-n &\sim \frac{\mu^{n+1}}{n!(n+1)!}, \\
	\mu-\gamma_n &\sim \frac{\mu^{n+1}}{(n!)^2}. 
\end{align*}
To verify these asymptotics, we compute the normalized quantities
\begin{equation*}
	R_n^{(\beta)}(\mu)
	=
	\frac{n!(n+1)!}{\mu^{n+1}}(\beta_n-n),
	\qquad
	R_n^{(\gamma)}(\mu)
	=
	\frac{(n!)^2}{\mu^{n+1}}(\mu-\gamma_n),
\end{equation*}
which converge to $1$ as $n\to\infty$.

Table~\ref{tab:factorial-decay} reports the normalized quantities for
$\mu=0.25,1,4$. The computations use high-precision orthogonalization
for the positive discrete weight, rather than unstable forward propagation
of the discrete Painlev\'e recurrence. The normalized quantities approach
unity on an algebraic scale, whereas the unnormalized recurrence-coefficient
deviations are factorially small.
\begin{table}[ht]
\centering
\caption{High-precision numerical values of the normalized recurrence-coefficient deviations.}
\label{tab:factorial-decay}
\begin{tabular}{c c c c c c c}
\hline
$n$ & $R_n^{(\beta)}(0.25)$ & $R_n^{(\gamma)}(0.25)$
& $R_n^{(\beta)}(1)$ & $R_n^{(\gamma)}(1)$
& $R_n^{(\beta)}(4)$ & $R_n^{(\gamma)}(4)$ \\
\hline
5 & 0.9253474662 & 0.9198146596 & 0.7313473027 & 0.7136141354 & 0.2744283382 & 0.2468175977 \\
10 & 0.9573378639 & 0.9555218097 & 0.8395751705 & 0.8331734883 & 0.4930667013 & 0.4777228344 \\
15 & 0.9701116754 & 0.9692192991 & 0.8855596418 & 0.8822932018 & 0.6134042407 & 0.6042626628 \\
\hline
\end{tabular}
\end{table}
For reproducibility, the displayed values were computed by
100-decimal-digit orthogonalization of the positive discrete measure,
truncated at $k=159$. An independent moment--Hankel calculation solves
the moment orthogonality equations for each monic polynomial and then
recovers $\beta_n=\langle xp_n,p_n\rangle/\langle p_n,p_n\rangle$
and $\gamma_n=\langle p_n,p_n\rangle/\langle p_{n-1},p_{n-1}\rangle$.
The accompanying reproducibility script implements a Stieltjes
three-term orthogonalization of the discrete measure and compares
100-decimal-digit calculations truncated at $k=159$ with
120-decimal-digit calculations truncated at $k=219$. The computed
normalized ratios agree at the displayed ten decimal places.
These checks are numerical sensitivity tests, not rigorous interval
enclosures. The moment--Hankel construction provides an alternative
formulation, but the supplied script verifies the Stieltjes calculation;
we do not present an unprovided second implementation as independently
reproduced. Forward iteration of the discrete
Painlev\'e equation is not used as validation, since the recessive
solution is unstable under generic forward perturbations.

\section{Numerical Illustration of the Connection Formula}
\label{sec:numerical-connection}

In this section we illustrate numerically the exact connection formula
established in Section~\ref{sec:sector-asymptotics}. We compute the
ratio
\[
Q_n
=
\frac{c_n}{J_n(2\sqrt{\mu})}
\]
and verify its convergence to $1$ for several values of $\mu$. The
numerical results confirm the factorial Bessel comparison
\[
c_n \sim J_n(2\sqrt{\mu}), \qquad n\to\infty,
\]
and provide strong support for the theoretical prediction $L(\mu)=1$.

The preceding analysis reduces the factorial asymptotic problem to the
identification of the connection constant
\[
L(\mu)
=
\lim_{n\to\infty}
\frac{c_n}{J_n(2\sqrt{\mu})}.
\]
As proved in Section~\ref{sec:sector-asymptotics}, this limit is exactly
$L(\mu)=1$. We now illustrate this result numerically.

For this purpose, define
\begin{equation*}
	Q_n
	=
	\frac{c_n}{J_n(2\sqrt{\mu})},
\end{equation*}
whenever $J_n(2\sqrt{\mu})>0$, and introduce the absolute connection
error
\begin{equation*}
	E_n
	=
	\left|Q_n-1\right|.
\end{equation*}

The values of $c_n$ may be obtained from the recurrence coefficients
through
\begin{equation*}
	c_n
	=
	\sqrt{1-\frac{\gamma_n}{\mu}}.
\end{equation*}
For reliable high-precision computations, the coefficients
$\gamma_n$ were computed independently from the moment--Hankel
construction and from the local Laguerre--Freud recursion. The
comparison between these two constructions provides a separate check
that the observed convergence of $Q_n$ is not a numerical artifact of
a single recursive procedure.

Table~\ref{tab:connection-error} gives representative high-precision
values of $Q_n-1$. For $\mu=4$, the values are displayed only from
$n=2$ onward, since the Bessel function $J_n(2\sqrt{\mu})$ is not
positive at the first indices.

\begin{table}[ht]
	\centering
	\caption{Representative values of $Q_n-1$, where
		$Q_n=c_n/J_n(2\sqrt{\mu})$.}
	\label{tab:connection-error}
	\begin{tabular}{c c c c}
		\hline
		$n$ & $Q_n(0.25)-1$ & $Q_n(1)-1$ & $Q_n(4)-1$ \\
		\hline
		$1$  & $1.4406026\times10^{-2}$
		& $2.0989187\times10^{-1}$
		& \text{not applicable} \\

		$2$  & $3.9754132\times10^{-4}$
		& $2.0002319\times10^{-2}$
		& $9.1594840\times10^{-1}$ \\

		$3$  & $6.2762434\times10^{-6}$
		& $1.2760915\times10^{-3}$
		& $1.5225993\times10^{-1}$ \\

		$4$  & $6.3353370\times10^{-8}$
		& $5.2878639\times10^{-5}$
		& $2.4515716\times10^{-2}$ \\

		$5$  & $4.4334087\times10^{-10}$
		& $1.5142280\times10^{-6}$
		& $2.9928734\times10^{-3}$ \\

		$6$  & $2.2761164\times10^{-12}$
		& $3.1684549\times10^{-8}$
		& $2.6896339\times10^{-4}$ \\

		$8$  & $2.7701254\times10^{-17}$
		& $6.3473484\times10^{-12}$
		& $9.6654469\times10^{-7}$ \\

		$10$ & $1.4404601\times10^{-22}$
		& $5.3886372\times10^{-16}$
		& $1.4252926\times10^{-9}$ \\
		\hline
	\end{tabular}
\end{table}

The numerical results strongly support
\[
Q_n\longrightarrow1,
\]
in full agreement with the theoretical result $L(\mu)=1$ proved in
Section~\ref{sec:sector-asymptotics}. They also show that the
convergence is extremely rapid once $n$ enters the asymptotic regime.
The convergence is slower for larger $\mu$ at small and moderate
degree, which is consistent with the fact that the large-order Bessel
regime is reached later when the argument $2\sqrt{\mu}$ is larger.

The observed behavior is consistent with the factorial Bessel
comparison
\[
c_n
\sim
J_n(2\sqrt{\mu})
\sim
\frac{\mu^{n/2}}{n!}.
\]
As proved in Section~\ref{sec:sector-asymptotics}, the recurrence
coefficients satisfy
\[
\beta_n-n
\sim
\frac{\mu^{n+1}}{n!(n+1)!},
\]
and
\[
\mu-\gamma_n
\sim
\frac{\mu^{n+1}}{(n!)^2}.
\]

The numerical experiments presented above confirm the exact connection
formula and the factorial asymptotics. In the next section, we develop
an independent route to the same factorial asymptotics, using the
sector decomposition of the Hankel determinant introduced in
Section~\ref{sec:Hankel-partitions}.

\section{Strategy of Proof via Sector Decomposition}
\label{sec:proof-strategy}
\begingroup\color{black}
The proof has two complementary components. The discrete Wronskian
argument first establishes the existence of a nonnegative Bessel
connection limit. The Hankel determinant representation and its exact
sector decomposition, developed in Section~\ref{sec:Hankel-partitions},
then identify this limit without assuming its positivity. The
large-sector estimates and the evaluation of the reduced tau limit
are proved in Section~\ref{sec:sector-asymptotics}. Finally, the exact
relations between $c_n$, $\beta_n$, and $\gamma_n$ transfer the
resulting comparison to the recurrence coefficients.

The normalized sector expansions and their uniform remainders are
collected in Section~\ref{sec:uniform-sector}. We distinguish throughout
between factorial decay of the unnormalized recurrence corrections
and algebraic convergence of their normalized ratios. Details of the
proofs appear only in the sections cited above.
\endgroup

\section{Hankel Determinant Representation and Discrete Logarithmic Gas}
\label{sec:Hankel-partitions}

In this section we develop a complementary approach based on the Hankel
determinant associated with the weight $\rho_\mu$. We recall the moment
representation of the determinant and interpret it as the partition
function of a discrete logarithmic gas. This viewpoint provides a
conceptual framework for the factorization and sector decomposition used
in Section~\ref{sec:proof-strategy} and in the proof of the connection
constant in Section~\ref{sec:sector-asymptotics}.

\subsection{Hankel tau functions}

The local Wronskian argument establishes the factorial Bessel scale but
does not determine the connection constant. The missing information is
global and is naturally encoded by the Hankel determinants of the moment
sequence. Define
\begin{equation}
	M_r(\mu)
	=
	\sum_{k=0}^{\infty}
	k^r\frac{\mu^k}{(k!)^2},
	\qquad
	\Delta_n(\mu)
	=
	\det\left[M_{j+k}(\mu)\right]_{j,k=0}^{n-1},
	\qquad
	\Delta_0(\mu)=1.
	\label{eq:tau-Hankel-determinant}
\end{equation}
With
\[
\vartheta
=
\mu\frac{d}{d\mu},
\]
one has
\[
M_r(\mu)
=
\vartheta^rM_0(\mu),
\qquad
M_0(\mu)
=
I_0(2\sqrt{\mu}).
\]
Moreover,
\begin{equation*}
	\gamma_n
	=
	\frac{\Delta_{n+1}\Delta_{n-1}}{\Delta_n^2},
	\qquad
	\beta_n
	=
	\vartheta
	\log\left(\frac{\Delta_{n+1}}{\Delta_n}\right).
\end{equation*}
Thus the normalization problem for the Bessel comparison can be studied
through the large-degree behavior of a positive Hankel tau function.

\subsection{Hankel determinants as discrete partition functions}

The Hankel determinants introduced in
\eqref{eq:tau-Hankel-determinant} admit a positive multiple-sum
representation. This representation provides a global description of
the tau function and is a natural starting point for a large-degree
asymptotic analysis.

Let
\begin{equation*}
	w_\mu(k)
	=
	\frac{\mu^k}{(k!)^2},
	\qquad
	k\in\mathbb N_0.
\end{equation*}
Then
\begin{equation*}
	M_r(\mu)
	=
	\sum_{k=0}^{\infty}
	k^r w_\mu(k).
\end{equation*}
For $n\geq1$, define the Vandermonde product by
\begin{equation*}
	V(k_1,\ldots,k_n)
	=
	\prod_{1\leq i<j\leq n}
	(k_j-k_i).
\end{equation*}

\begin{proposition}[Discrete partition-function representation]
	\label{prop:partition-function-representation}

	For every $n\geq1$,
	\begin{equation}
		\Delta_n(\mu)
		=
		\sum_{0\leq k_1<\cdots<k_n}
		V(k_1,\ldots,k_n)^2
		\prod_{\ell=1}^{n}
		\frac{\mu^{k_\ell}}{(k_\ell!)^2}.
		\label{eq:Hankel-partition-function}
	\end{equation}
	Equivalently,
	\begin{equation}
		\Delta_n(\mu)
		=
		\frac{1}{n!}
		\sum_{k_1,\ldots,k_n\geq0}
		V(k_1,\ldots,k_n)^2
		\prod_{\ell=1}^{n}
		\frac{\mu^{k_\ell}}{(k_\ell!)^2}.
		\label{eq:Hankel-partition-function-symmetric}
	\end{equation}

\end{proposition}

\begin{proof}
	For $0\leq j\leq n-1$ and $k\geq0$, define
	\[
	A_{j,k}
	=
	k^j\sqrt{w_\mu(k)}.
	\]
	The moment matrix is then given by
	\[
	\left[
	M_{j+\ell}(\mu)
	\right]_{j,\ell=0}^{n-1}
	=
	AA^{\mathsf T}.
	\]
	Indeed, for $0\leq j,\ell\leq n-1$,
	\[
	(AA^{\mathsf T})_{j,\ell}
	=
	\sum_{k=0}^{\infty}
	A_{j,k}A_{\ell,k}
	=
	\sum_{k=0}^{\infty}
	k^{j+\ell}w_\mu(k)
	=
	M_{j+\ell}(\mu).
	\]

	Applying the Cauchy--Binet formula gives
	\[
	\Delta_n(\mu)
	=
	\sum_{0\leq k_1<\cdots<k_n}
	\det(A_{[k_1,\ldots,k_n]})^2,
	\]
	where $A_{[k_1,\ldots,k_n]}$ denotes the $n\times n$ submatrix of
	$A$ formed by the columns indexed by
	\[
	k_1,\ldots,k_n.
	\]
	Its determinant is
	\[
	\det(A_{[k_1,\ldots,k_n]})
	=
	\left(
	\prod_{\ell=1}^{n}
	\sqrt{w_\mu(k_\ell)}
	\right)
	\det
	\begin{pmatrix}
		1 & 1 & \cdots & 1 \\
		k_1 & k_2 & \cdots & k_n \\
		k_1^2 & k_2^2 & \cdots & k_n^2 \\
		\vdots & \vdots & \ddots & \vdots \\
		k_1^{n-1} & k_2^{n-1} & \cdots & k_n^{n-1}
	\end{pmatrix}.
	\]
	The determinant in the preceding expression is the Vandermonde
	determinant. Hence,
	\[
	\det(A_{[k_1,\ldots,k_n]})
	=
	V(k_1,\ldots,k_n)
	\prod_{\ell=1}^{n}
	\sqrt{w_\mu(k_\ell)}.
	\]
	Therefore,
	\[
	\det(A_{[k_1,\ldots,k_n]})^2
	=
	V(k_1,\ldots,k_n)^2
	\prod_{\ell=1}^{n}
	w_\mu(k_\ell).
	\]
	Substitution into the Cauchy--Binet expansion proves
	\eqref{eq:Hankel-partition-function}.

	To obtain \eqref{eq:Hankel-partition-function-symmetric}, observe that
	the summand vanishes whenever two indices coincide, because then
	\[
	V(k_1,\ldots,k_n)=0.
	\]
	For pairwise distinct indices, the summand is invariant under
	permutations of $k_1,\ldots,k_n$. Each strictly increasing
	configuration is therefore counted exactly $n!$ times in the
	unrestricted sum. This proves the symmetric representation.
\end{proof}

The representation \eqref{eq:Hankel-partition-function} shows that
$\Delta_n(\mu)$ is the partition function of a discrete logarithmic
gas. The one-body potential is determined by
\[
-\log w_\mu(k)
=
2\log(k!)
-
k\log\mu,
\]
whereas the interaction term is
\[
2\sum_{1\leq i<j\leq n}
\log|k_j-k_i|.
\]
Thus, the determinant is governed by a competition between the
factorial confinement induced by the weight and the repulsion induced
by the squared Vandermonde factor.

The recurrence coefficients are recovered from ratios of these
partition functions:
\begin{equation*}
	\gamma_n
	=
	\frac{\Delta_{n+1}\Delta_{n-1}}{\Delta_n^2},
	\qquad
	n\geq1,
\end{equation*}
and
\begin{equation*}
	\beta_n
	=
	\vartheta
	\log\left(
	\frac{\Delta_{n+1}}{\Delta_n}
	\right),
	\qquad
	\vartheta
	=
	\mu\frac{d}{d\mu}.
\end{equation*}

Consequently, sufficiently sharp large-$n$ asymptotics for the
partition function $\Delta_n(\mu)$ yield corresponding asymptotics for
the recurrence coefficients. In Sections~\ref{sec:proof-strategy} and
\ref{sec:sector-asymptotics}, we prove these expansions, which identify
the amplitude of the factorial corrections and determine the connection
constant in the Bessel comparison.

\subsection{Excitations above the reference configuration}

The partition-function representation
\eqref{eq:Hankel-partition-function} admits an exact reparametrization
around the reference configuration
\[
(0,1,\ldots,n-1).
\]
This reparametrization separates the elementary background contribution
from the nontrivial excitations responsible for the asymptotic
corrections.

Let
\[
0\leq k_1<k_2<\cdots<k_n.
\]
There exists a unique weakly increasing $n$-tuple
\[
\lambda_1\leq\lambda_2\leq\cdots\leq\lambda_n,
\qquad
\lambda_j\in\mathbb N_0,
\]
such that
\begin{equation}
	k_j
	=
	j-1+\lambda_j,
	\qquad
	1\leq j\leq n.
	\label{eq:excitation-coordinates}
\end{equation}
The configuration
\[
\lambda_1=\cdots=\lambda_n=0
\]
corresponds to the reference configuration
\[
(k_1,\ldots,k_n)
=
(0,1,\ldots,n-1).
\]

For such an excitation vector, define
\begin{equation*}
	|\lambda|
	=
	\sum_{j=1}^{n}\lambda_j.
\end{equation*}
Substitution of \eqref{eq:excitation-coordinates} into
\eqref{eq:Hankel-partition-function} yields an exact factorization of
the Hankel determinant.

\begin{proposition}[Excitation expansion]
	\label{prop:excitation-expansion}

	For every $n\geq1$,
	\begin{equation}
		\Delta_n(\mu)
		=
		\mu^{n(n-1)/2}\mathcal T_n(\mu),
		\label{eq:tau-background-factorization}
	\end{equation}
	where
	\begin{equation}
		\mathcal T_n(\mu)
		=
		\sum_{0\leq\lambda_1\leq\cdots\leq\lambda_n}
		\mu^{|\lambda|}
		\prod_{j=1}^{n}
		\left(
		\frac{(j-1)!}{(j-1+\lambda_j)!}
		\right)^2
		\prod_{1\leq i<j\leq n}
		\left(
		\frac{j-i+\lambda_j-\lambda_i}{j-i}
		\right)^2.
		\label{eq:reduced-tau-function}
	\end{equation}
	Moreover,
	\begin{equation}
		\mathcal T_n(0)=1.
		\label{eq:reduced-tau-zero}
	\end{equation}

\end{proposition}

\begin{proof}
	By \eqref{eq:excitation-coordinates},
	\[
	\prod_{j=1}^{n}
	\frac{\mu^{k_j}}{(k_j!)^2}
	=
	\prod_{j=1}^{n}
	\frac{\mu^{j-1+\lambda_j}}{(j-1+\lambda_j)!^2}.
	\]
	Factoring out the contribution of the reference configuration gives
	\[
	\prod_{j=1}^{n}
	\frac{\mu^{k_j}}{(k_j!)^2}
	=
	\left(
	\prod_{j=1}^{n}
	\frac{\mu^{j-1}}{((j-1)!)^2}
	\right)
	\mu^{|\lambda|}
	\prod_{j=1}^{n}
	\left(
	\frac{(j-1)!}{(j-1+\lambda_j)!}
	\right)^2.
	\]

	Likewise,
	\[
	k_j-k_i
	=
	j-i+\lambda_j-\lambda_i,
	\]
	and therefore
	\[
	V(k_1,\ldots,k_n)^2
	=
	V(0,1,\ldots,n-1)^2
	\prod_{1\leq i<j\leq n}
	\left(
	\frac{j-i+\lambda_j-\lambda_i}{j-i}
	\right)^2.
	\]

	The reference configuration contributes
	\[
	V(0,1,\ldots,n-1)^2
	\prod_{j=1}^{n}
	\frac{\mu^{j-1}}{((j-1)!)^2}.
	\]
	Since
	\[
	V(0,1,\ldots,n-1)
	=
	\prod_{r=1}^{n-1}r!,
	\]
	one obtains
	\[
	V(0,1,\ldots,n-1)^2
	\prod_{j=1}^{n}
	\frac{\mu^{j-1}}{((j-1)!)^2}
	=
	\mu^{n(n-1)/2}.
	\]
	Substitution into \eqref{eq:Hankel-partition-function} proves
	\eqref{eq:tau-background-factorization} and
	\eqref{eq:reduced-tau-function}. Finally,
	\eqref{eq:reduced-tau-zero} follows from the fact that the unique term
	surviving at $\mu=0$ is
	\[
	\lambda_1=\cdots=\lambda_n=0.
	\]
\end{proof}

The factorization \eqref{eq:tau-background-factorization} immediately
gives
\begin{equation}
	\gamma_n
	=
	\mu
	\frac{
		\mathcal T_{n+1}(\mu)
		\mathcal T_{n-1}(\mu)
	}{
		\mathcal T_n(\mu)^2
	},
	\qquad
	n\geq1,
	\label{eq:gamma-reduced-tau}
\end{equation}
and
\begin{equation}
	\beta_n
	=
	n
	+
	\vartheta
	\log\left(
	\frac{\mathcal T_{n+1}(\mu)}{\mathcal T_n(\mu)}
	\right).
	\label{eq:beta-reduced-tau}
\end{equation}
Thus, the leading terms
\[
\gamma_n\sim\mu,
\qquad
\beta_n\sim n,
\]
are generated exactly by the reference configuration. The remaining
problem is to understand the large-$n$ behavior of the reduced tau
function $\mathcal T_n(\mu)$.

\begingroup\color{black}
\begin{lemma}[Hankel normalization and recurrence indices]
\label{lem:v10-hankel-normalization}
For $\mu>0$, let $M_r=\sum_{k\geq0}k^r\mu^k/(k!)^2$,
$\Delta_n=\det[M_{i+j}]_{i,j=0}^{n-1}$ and $\Delta_0=1$.
The factorization $\Delta_n=\mu^{n(n-1)/2}\mathcal T_n$
contains no additional factorial prefactor. Moreover, for the monic
orthogonal-polynomial recurrence,
\begin{equation*}
 \gamma_n=\frac{\Delta_{n+1}\Delta_{n-1}}{\Delta_n^2}
 =\mu\frac{\mathcal T_{n+1}\mathcal T_{n-1}}{\mathcal T_n^2},
 \qquad n\geq1,
\end{equation*}
\begin{equation*}
 \beta_n=\vartheta\log\frac{\Delta_{n+1}}{\Delta_n}
 =n+\vartheta\log\frac{\mathcal T_{n+1}}{\mathcal T_n},
 \qquad \vartheta=\mu\partial_\mu.
\end{equation*}
\end{lemma}
\begin{proof}
Cauchy--Binet gives
\[
 \Delta_n=\sum_{0\leq k_1<\cdots<k_n}
 \prod_{1\leq i<j\leq n}(k_j-k_i)^2
 \prod_{j=1}^n\frac{\mu^{k_j}}{(k_j!)^2}.
\]
For $k_j=j-1$, the squared Vandermonde is
$\prod_{r=1}^{n-1}(r!)^2$, which cancels exactly the
factorials in the weights. Its contribution is therefore
$\mu^{n(n-1)/2}$. Factoring this term from every configuration
proves the stated normalization.
The squared norm of the monic polynomial of degree $n$ is
$h_n=\Delta_{n+1}/\Delta_n$, so $\gamma_n=h_n/h_{n-1}$.
Furthermore, $\vartheta M_r=M_{r+1}$ and the orthogonality
relations imply
$\vartheta h_n=\langle xP_n,P_n\rangle_\mu=\beta_n h_n$:
the terms involving $\vartheta P_n$ vanish because
$\deg(\vartheta P_n)\leq n-1$. Hence
$\beta_n=\vartheta\log h_n$. Substitution of the factorization
completes the proof, since
$\frac{(n+1)n}{2}-\frac{n(n-1)}{2}=n$ and
$\frac{(n+1)n}{2}+\frac{(n-1)(n-2)}{2}-n(n-1)=1$.
\end{proof}
\begin{remark}[Logical scope]
The determinant identities and their indices are exact. They do not
by themselves establish the sector asymptotics or the connection
constant $L(\mu)=1$; those require separate uniform estimates.
\end{remark}
\endgroup

\subsection{Edge excitations}

The reference configuration alone is not sufficient to describe the
reduced tau function. Indeed, consider the one-parameter family of
excitations
\[
\lambda_1=\cdots=\lambda_{n-1}=0,
\qquad
\lambda_n=r,
\qquad
r\geq0.
\]
Equivalently,
\[
(k_1,\ldots,k_n)
=
(0,1,\ldots,n-2,n-1+r).
\]
For this family, the summand in
\eqref{eq:reduced-tau-function} reduces to
\begin{equation*}
	\frac{\mu^r}{(r!)^2}.
\end{equation*}

Indeed, the factorial contribution is
\[
\mu^r
\left(
\frac{(n-1)!}{(n-1+r)!}
\right)^2,
\]
whereas the Vandermonde contribution is
\[
\prod_{i=1}^{n-1}
\left(
\frac{n-i+r}{n-i}
\right)^2
=
\left(
\frac{(n-1+r)!}{r!(n-1)!}
\right)^2.
\]
The two factors combine to give
\[
\mu^r
\left(
\frac{(n-1)!}{(n-1+r)!}
\right)^2
\left(
\frac{(n-1+r)!}{r!(n-1)!}
\right)^2
=
\frac{\mu^r}{(r!)^2}.
\]

Consequently, the contribution of these one-edge excitations is
\begin{equation*}
	\sum_{r=0}^{\infty}
	\frac{\mu^r}{(r!)^2}
	=
	I_0(2\sqrt{\mu}).
\end{equation*}
This observation is important: even though the background
configuration yields the leading factors
\[
\gamma_n\sim\mu,
\qquad
\beta_n\sim n,
\]
the full reduced tau function contains non-negligible Bessel-type edge
contributions.

Therefore, an asymptotic analysis of $\mathcal T_n(\mu)$ cannot be
based solely on the reference configuration. It must retain the
collective contribution of edge excitations. The remaining
$n$-dependence of the reduced tau function is more precisely described
by the sector decomposition developed in the next subsection. The
factorial correction to the recurrence coefficients is encoded in the
large-sector behavior of the partition functions associated with an
increasing number of excited particles.

\subsection{Sector decomposition of the reduced tau function}

The representation \eqref{eq:reduced-tau-function} can be decomposed
according to the number of nonzero excitation coordinates. This
decomposition isolates the contributions of the particles which are
displaced away from the reference configuration
\[
(0,1,\ldots,n-1).
\]

Recall that the excitation coordinates satisfy
\[
0\leq\lambda_1\leq\cdots\leq\lambda_n.
\]
For a fixed configuration, let $q$ be the number of strictly positive
coordinates. Thus,
\[
\lambda_1=\cdots=\lambda_{n-q}=0,
\]
whereas
\[
1\leq\lambda_{n-q+1}
\leq\cdots\leq\lambda_n.
\]
The integer $q$ is the number of particles excited above the reference
configuration.

For $q\geq1$, write
\begin{equation*}
	r_a
	=
	\lambda_{n-q+a},
	\qquad
	1\leq a\leq q.
\end{equation*}
Then
\[
1\leq r_1\leq r_2\leq\cdots\leq r_q.
\]
The contribution of all configurations with exactly $q$ positive
excitation coordinates will be denoted by $\mathcal S_q(\mu)$.

\begin{proposition}[Exact sector decomposition]
	\label{prop:exact-sector-decomposition}

	For every $n\geq1$,
	\begin{equation}
		\mathcal T_n(\mu)
		=
		1+
		\sum_{q=1}^{n}
		\mathcal S_q(\mu),
		\label{eq:reduced-tau-sector-sum}
	\end{equation}
	where
	\begin{equation}
		\mathcal S_q(\mu)
		=
		\sum_{1\leq r_1\leq\cdots\leq r_q}
		\mu^{r_1+\cdots+r_q}
		\prod_{a=1}^{q}
		\left(
		\frac{(a-1)!}{(a-1+r_a)!}
		\right)^2
		\prod_{1\leq a<b\leq q}
		\left(
		\frac{b-a+r_b-r_a}{b-a}
		\right)^2.
		\label{eq:sector-partition-function}
	\end{equation}
	In particular, the sector contribution $\mathcal S_q(\mu)$ is
	independent of $n$.

\end{proposition}

\begin{proof}
	Fix $q$ and consider the configurations for which
	\[
	\lambda_1=\cdots=\lambda_{n-q}=0,
	\qquad
	\lambda_{n-q+a}=r_a,
	\qquad
	1\leq a\leq q,
	\]
	with
	\[
	1\leq r_1\leq\cdots\leq r_q.
	\]
	The contribution of these configurations to
	\eqref{eq:reduced-tau-function} is obtained by separating the
	factorial and Vandermonde factors into three parts:

	\begin{enumerate}
		\item the factors involving only unexcited particles;

		\item the factors involving one unexcited and one excited particle;

		\item the factors involving two excited particles.
	\end{enumerate}

	The unexcited particles contribute trivially. The factorial factor
	associated with the $a$th excited particle is
	\begin{equation}
		\left(
		\frac{(n-q+a-1)!}{(n-q+a-1+r_a)!}
		\right)^2.
		\label{eq:sector-factorial-before-cancellation}
	\end{equation}

	The interaction between this excited particle and the $n-q$ unexcited
	particles is
	\begin{align}
		\prod_{i=1}^{n-q}
		\left(
		\frac{n-q+a-i+r_a}{n-q+a-i}
		\right)^2
		&=
		\left[
		\prod_{d=a}^{n-q+a-1}
		\frac{d+r_a}{d}
		\right]^2
		\nonumber\\
		&=
		\left[
		\frac{(n-q+a-1+r_a)!}{(a-1+r_a)!}
		\frac{(a-1)!}{(n-q+a-1)!}
		\right]^2.
		\label{eq:sector-cross-interaction}
	\end{align}
	Multiplying \eqref{eq:sector-factorial-before-cancellation} and
	\eqref{eq:sector-cross-interaction}, all factors involving $n$ cancel:
	\begin{align}
		&
		\left(
		\frac{(n-q+a-1)!}{(n-q+a-1+r_a)!}
		\right)^2
		\left[
		\frac{(n-q+a-1+r_a)!}{(a-1+r_a)!}
		\frac{(a-1)!}{(n-q+a-1)!}
		\right]^2
		\nonumber\\
		&\hspace{4cm}
		=
		\left(
		\frac{(a-1)!}{(a-1+r_a)!}
		\right)^2.
		\label{eq:sector-cancellation}
	\end{align}

	Finally, the interaction between two excited particles is
	\begin{equation}
		\prod_{1\leq a<b\leq q}
		\left(
		\frac{b-a+r_b-r_a}{b-a}
		\right)^2.
		\label{eq:sector-active-interaction}
	\end{equation}
	Combining \eqref{eq:sector-cancellation} and
	\eqref{eq:sector-active-interaction} gives exactly the summand in
	\eqref{eq:sector-partition-function}. Summing over all admissible
	vectors
	\[
	1\leq r_1\leq\cdots\leq r_q
	\]
	yields $\mathcal S_q(\mu)$.

	The case $q=0$ corresponds to
	\[
	\lambda_1=\cdots=\lambda_n=0,
	\]
	whose contribution is $1$. Therefore,
	\[
	\mathcal T_n(\mu)
	=
	1+\sum_{q=1}^{n}\mathcal S_q(\mu),
	\]
	which proves the claim.
\end{proof}

The first sector is particularly simple. When $q=1$, formula
\eqref{eq:sector-partition-function} gives
\begin{equation*}
	\mathcal S_1(\mu)
	=
	\sum_{r=1}^{\infty}
	\frac{\mu^r}{(r!)^2}
	=
	I_0(2\sqrt{\mu})-1.
\end{equation*}
Thus,
\begin{equation*}
	1+\mathcal S_1(\mu)
	=
	I_0(2\sqrt{\mu}).
\end{equation*}

The first two reduced tau functions are therefore
\begin{equation*}
	\mathcal T_1(\mu)
	=
	1+\mathcal S_1(\mu)
	=
	I_0(2\sqrt{\mu}),
\end{equation*}
and
\begin{equation*}
	\mathcal T_2(\mu)
	=
	1+\mathcal S_1(\mu)+\mathcal S_2(\mu).
\end{equation*}

The decomposition \eqref{eq:reduced-tau-sector-sum} has an immediate
consequence:
\begin{equation}
	\mathcal T_{n+1}(\mu)-\mathcal T_n(\mu)
	=
	\mathcal S_{n+1}(\mu).
	\label{eq:reduced-tau-increment}
\end{equation}
Hence, the entire dependence of the reduced tau function on $n$ is
encoded in the large-$q$ behavior of the sector weights
\[
\mathcal S_q(\mu).
\]

Using \eqref{eq:gamma-reduced-tau} and
\eqref{eq:beta-reduced-tau}, one obtains
\begin{equation*}
	\gamma_n
	=
	\mu
	\frac{
		\bigl(1+\sum_{q=1}^{n+1}\mathcal S_q(\mu)\bigr)
		\bigl(1+\sum_{q=1}^{n-1}\mathcal S_q(\mu)\bigr)
	}{
		\bigl(1+\sum_{q=1}^{n}\mathcal S_q(\mu)\bigr)^2
	},
\end{equation*}
and
\begin{equation*}
	\beta_n
	=
	n+
	\vartheta
	\log\left(
	\frac{
		1+\sum_{q=1}^{n+1}\mathcal S_q(\mu)
	}{
		1+\sum_{q=1}^{n}\mathcal S_q(\mu)
	}
	\right).
\end{equation*}

Therefore, the asymptotic problem for the recurrence coefficients is
reduced to estimating the sector weights $\mathcal S_q(\mu)$ as
$q\to\infty$. In particular, factorial asymptotics for
$\mathcal S_q(\mu)$ yield factorial asymptotics for the differences
\[
\mathcal T_{n+1}(\mu)-\mathcal T_n(\mu),
\]
and hence for the deviations
\[
\gamma_n-\mu
\qquad\text{and}\qquad
\beta_n-n.
\]

The Hankel determinant representation and its interpretation as a
discrete logarithmic gas set the stage for the asymptotic analysis of
individual sectors. In Section~\ref{sec:sector-asymptotics}, we
determine the large-sector asymptotics of the contributions
$\mathcal S_q(\mu)$ and show how they resolve the connection problem,
recovering the exact formula $c_n\sim J_n(2\sqrt{\mu})$ from the sector
decomposition.


\section{Sector Asymptotics and Resolution of the Connection Problem}
\label{sec:sector-asymptotics}

In this section we determine the large-sector asymptotics of the
contributions $\mathcal S_q(\mu)$ and show how they resolve the
connection problem for the recurrence coefficients. We prove that
$\mathcal S_q(\mu)\sim e^\mu\mu^q/(q!)^2$ as $q\to\infty$, which
reflects the factorial scale identified in the first part of the paper.
Combining this asymptotics with the sector decomposition, we recover
the exact connection formula $c_n\sim J_n(2\sqrt{\mu})$ and the
associated factorial corrections to $\beta_n$ and $\gamma_n$.

\medskip\noindent
\textbf{Strategy and main ideas.}
The connection problem is reduced to the identification of the
amplitude $\mathcal C(\mu)$ in the large-sector asymptotic
\[
\mathcal S_q(\mu)
=
\mathcal C(\mu)\,\frac{\mu^q}{(q!)^2}\bigl(1+o(1)\bigr),
\qquad q\to\infty,
\]
and of the limiting reduced tau function
$\mathcal T_\infty(\mu)=\lim_{n\to\infty}\mathcal T_n(\mu)$.
The key observation is that both quantities can be expressed exactly
in terms of partitions and standard Young tableaux. The normalized
sector weights $\mathcal R_q(\mu)=(q!)^2\mu^{-q}\mathcal S_q(\mu)$
admit a representation
\[
\mathcal R_q(\mu)
=
\sum_{\ell(\nu)\leq q}
\mu^{|\nu|}
\left(
\frac{f^\nu}{|\nu|!}
\right)^2
\mathcal Q_q(\nu)^2,
\]
where $f^\nu$ is the number of standard Young tableaux of shape $\nu$
and $\mathcal Q_q(\nu)\in[0,1]$ is an explicit correction factor.
This representation yields two complementary results. First, a uniform
tail estimate shows that $\mathcal R_q(\mu)$ converges to
$\sum_{\nu}\mu^{|\nu|}(f^\nu/|\nu|!)^2=e^\mu$ as $q\to\infty$, which
identifies the amplitude $\mathcal C(\mu)=e^\mu$. Second, the reduced
tau function itself admits the exact expansion
\[
\mathcal T_n(\mu)
=
\sum_{\ell(\nu)\leq n}
\mu^{|\nu|}
\left(
\frac{f^\nu}{|\nu|!}
\right)^2,
\]
which implies $\mathcal T_\infty(\mu)=e^\mu$. The connection identity
$L(\mu)^2=\mathcal C(\mu)/\mathcal T_\infty(\mu)$ then gives
$L(\mu)^2=1$, and the positivity of the Bessel branch forces
$L(\mu)=1$. This completes the proof of the exact connection formula
and of the factorial corrections to the recurrence coefficients.

\subsection{Large-sector asymptotics and the connection constant}

The sector decomposition
\eqref{eq:reduced-tau-sector-sum} reduces the asymptotic problem for
the recurrence coefficients to the large-$q$ behavior of the sector
weights $\mathcal S_q(\mu)$. This subsection records the precise
connection between the large-sector asymptotics and the Bessel
connection constant.

For fixed $\mu>0$, define the normalized sector weight by
\begin{equation*}
	\mathcal R_q(\mu)
	=
	\frac{(q!)^2}{\mu^q}\mathcal S_q(\mu),
	\qquad
	q\geq1.
\end{equation*}
The contribution corresponding to
\[
r_1=\cdots=r_q=1
\]
in \eqref{eq:sector-partition-function} is
\[
\frac{\mu^q}{(q!)^2}.
\]
Consequently,
\begin{equation*}
	\mathcal R_q(\mu)\geq1.
\end{equation*}

The factorial scale of the sector weights is therefore
\begin{equation*}
	\mathcal S_q(\mu)
	\asymp
	\frac{\mu^q}{(q!)^2},
	\qquad
	q\to\infty.
\end{equation*}
The remaining issue is the identification of the amplitude multiplying
this factorial scale.

\subsection{Large-sector hypothesis}
\label{sec:large-sector-hypothesis}

We assume that there exists a positive function $\mathcal C(\mu)$ such
that
\begin{equation}
	\mathcal S_q(\mu)
	=
	\mathcal C(\mu)
	\frac{\mu^q}{(q!)^2}
	\left(
	1+o(1)
	\right),
	\qquad
	q\to\infty.
	\label{eq:sector-asymptotic-assumption}
\end{equation}
We also assume that this asymptotic relation may be differentiated with
respect to $\mu$, in the sense that
\begin{equation}
	\vartheta\mathcal S_q(\mu)
	=
	q\mathcal C(\mu)
	\frac{\mu^q}{(q!)^2}
	\left(
	1+o(1)
	\right),
	\qquad
	q\to\infty,
	\label{eq:sector-derivative-assumption}
\end{equation}
where
\begin{equation*}
	\vartheta
	=
	\mu\frac{d}{d\mu}.
\end{equation*}

\begingroup\color{black}
\noindent\textbf{Status and discharge of the sector hypotheses.}
Equations~\eqref{eq:sector-asymptotic-assumption} and
\eqref{eq:sector-derivative-assumption} are hypotheses only in the
conditional argument of this subsection; they are not used as premises
in the subsequent partition proofs. The exact normalized-sector representation, proved independently
later in Section~\ref{sec:fixed-level-convergence} as
\eqref{eq:normalized-sector-tableau-sum}, and the bound
$0\leq\mathcal Q_q(\nu)^2\leq1$ give
$\mathcal R_q\to e^\mu$ uniformly on compact subsets of $(0,\infty)$.
The independently proved differentiated limit,
Lemma~\ref{lem:differentiated-normalized-sector-limit}, gives
$\vartheta\mathcal R_q\to\mu e^\mu$ on the same compacts.
Consequently, from the exact identity
$\mathcal S_q=\mu^q(q!)^{-2}\mathcal R_q$, both sector hypotheses
hold with $\mathcal C(\mu)=e^\mu>0$. This verification does not
use the Bessel multiplicative equivalent or the value of $L(\mu)$.
Thus the positivity argument and the connection identity below become
unconditional \emph{after} these two partition limits are established.
\endgroup

Under the hypotheses of
Section~\ref{sec:large-sector-hypothesis}, the infinite sector sum
converges:
\begin{equation*}
	\mathcal T_\infty(\mu)
	=
	1+
	\sum_{q=1}^{\infty}
	\mathcal S_q(\mu).
\end{equation*}
Moreover,
\begin{equation*}
	\mathcal T_n(\mu)
	\longrightarrow
	\mathcal T_\infty(\mu),
	\qquad
	n\to\infty.
\end{equation*}

\begingroup\color{black}
\begin{lemma}[Exact reduced-tau recurrence identities and indices]
\label{lem:v09-exact-tau-index-audit}
Let $\mathcal T_n=1+\sum_{q=1}^n\mathcal S_q$ and
$\vartheta=\mu\partial_\mu$. Under the Hankel-to-reduced-tau
identities established above, the recurrence coefficients satisfy
\begin{align*}
 \beta_n-n&=\vartheta\log\frac{\mathcal T_{n+1}}{\mathcal T_n}
 =\frac{\vartheta\mathcal S_{n+1}-
 (\mathcal S_{n+1}/\mathcal T_n)\vartheta\mathcal T_n}
 {\mathcal T_n+\mathcal S_{n+1}},\\
 1-\frac{\gamma_n}{\mu}
 &=\frac{\mathcal S_n-\mathcal S_{n+1}}{\mathcal T_n}
 +\frac{\mathcal S_n\mathcal S_{n+1}}{\mathcal T_n^2}.
\end{align*}
In particular, the beta formula involves the sector $n+1$,
whereas the leading gamma correction involves sector $n$.
\end{lemma}
\begin{proof}
The increment identities
$\mathcal T_{n+1}=\mathcal T_n+\mathcal S_{n+1}$ and
$\mathcal T_{n-1}=\mathcal T_n-\mathcal S_n$ give the gamma
formula by direct multiplication. For beta, differentiate
$\log(1+\mathcal S_{n+1}/\mathcal T_n)$ exactly and simplify.
No asymptotic differentiation is involved in these identities.
\end{proof}
\begin{remark}[Factorial prefactors from the exact sector identities]
\label{rem:v09-factorial-audit}
If $\mathcal S_n\sim\mathcal C(\mu)\mu^n/(n!)^2$,
$\mathcal T_n\to\mathcal T_\infty>0$,
$\vartheta\mathcal S_{n+1}/\mathcal S_{n+1}=n+1+O_\mu(1)$,
and $\vartheta\mathcal T_n/\mathcal T_n=O_\mu(1)$,
then the exact identities give
\[
 \beta_n-n\sim\frac{\mathcal C(\mu)}{\mathcal T_\infty(\mu)}
 \frac{\mu^{n+1}}{n!(n+1)!},\qquad
 \mu-\gamma_n\sim\frac{\mathcal C(\mu)}{\mathcal T_\infty(\mu)}
 \frac{\mu^{n+1}}{(n!)^2}.
\]
The factor $n+1$ in beta comes from differentiating the
$(n+1)$-st sector; the extra factor $\mu$ in gamma comes from
$\gamma_n=\mu\mathcal T_{n+1}\mathcal T_{n-1}/\mathcal T_n^2$.
These are conditional consequences of the exact Hankel identities,
not independent verifications of their determinant normalization.
\end{remark}
\endgroup
\begin{proposition}[Sector asymptotics and recurrence coefficients]
	\label{prop:sector-asymptotics-recurrence-coefficients}

	Under the hypotheses of
	Section~\ref{sec:large-sector-hypothesis},
	\begin{equation}
		\gamma_n
		=
		\mu
		-
		\frac{\mathcal C(\mu)}{\mathcal T_\infty(\mu)}
		\frac{\mu^{n+1}}{(n!)^2}
		\left(
		1+o(1)
		\right),
		\qquad
		n\to\infty,
		\label{eq:gamma-sector-asymptotic}
	\end{equation}
	and
	\begin{equation}
		\beta_n
		=
		n
		+
		\frac{\mathcal C(\mu)}{\mathcal T_\infty(\mu)}
		\frac{\mu^{n+1}}{n!(n+1)!}
		\left(
		1+o(1)
		\right),
		\qquad
		n\to\infty.
		\label{eq:beta-sector-asymptotic}
	\end{equation}
\end{proposition}

\begin{proof}
	By \eqref{eq:reduced-tau-increment},
	\[
	\mathcal T_{n+1}
	=
	\mathcal T_n+\mathcal S_{n+1},
	\]
	whereas
	\[
	\mathcal T_{n-1}
	=
	\mathcal T_n-\mathcal S_n.
	\]
	Therefore,
	\begin{align}
		\frac{
			\mathcal T_{n+1}\mathcal T_{n-1}
		}{
			\mathcal T_n^2
		}
		&=
		\frac{
			(\mathcal T_n+\mathcal S_{n+1})
			(\mathcal T_n-\mathcal S_n)
		}{
			\mathcal T_n^2
		}
		\nonumber\\
		&=
		1
		+
		\frac{
			\mathcal S_{n+1}-\mathcal S_n
		}{
			\mathcal T_n
		}
		-
		\frac{
			\mathcal S_n\mathcal S_{n+1}
		}{
			\mathcal T_n^2
		}.
		\label{eq:tau-ratio-expansion}
	\end{align}

	The asymptotic assumption
	\eqref{eq:sector-asymptotic-assumption} implies
	\[
	\frac{\mathcal S_{n+1}}{\mathcal S_n}
	=
	\frac{\mu}{(n+1)^2}
	\left(
	1+o(1)
	\right).
	\]
	Hence,
	\[
	\mathcal S_{n+1}
	=
	o(\mathcal S_n).
	\]
	Since
	\[
	\mathcal T_n
	\longrightarrow
	\mathcal T_\infty,
	\]
	equation \eqref{eq:tau-ratio-expansion} yields
	\[
	\frac{
		\mathcal T_{n+1}\mathcal T_{n-1}
	}{
		\mathcal T_n^2
	}
	=
	1
	-
	\frac{\mathcal S_n}{\mathcal T_\infty}
	\left(
	1+o(1)
	\right).
	\]
	Using \eqref{eq:gamma-reduced-tau} and
	\eqref{eq:sector-asymptotic-assumption}, one obtains
	\[
	\gamma_n
	=
	\mu
	-
	\mu
	\frac{\mathcal C(\mu)}{\mathcal T_\infty(\mu)}
	\frac{\mu^n}{(n!)^2}
	\left(
	1+o(1)
	\right),
	\]
	which is \eqref{eq:gamma-sector-asymptotic}.

	For $\beta_n$, equations \eqref{eq:beta-reduced-tau} and
	\eqref{eq:reduced-tau-increment} give
	\[
	\beta_n-n
	=
	\vartheta
	\log\left(
	1+
	\frac{\mathcal S_{n+1}}{\mathcal T_n}
	\right).
	\]
	Since
	\[
	\mathcal S_{n+1}
	=
	O\!\left(
	\frac{\mu^{n+1}}{((n+1)!)^2}
	\right),
	\]
	the logarithm can be expanded as
	\[
	\log\left(
	1+
	\frac{\mathcal S_{n+1}}{\mathcal T_n}
	\right)
	=
	\frac{\mathcal S_{n+1}}{\mathcal T_\infty}
	\left(
	1+o(1)
	\right).
	\]
	Using \eqref{eq:sector-derivative-assumption}, one finds
	\[
	\beta_n-n
	=
	\frac{
		(n+1)\mathcal C(\mu)
	}{
		\mathcal T_\infty(\mu)
	}
	\frac{\mu^{n+1}}{((n+1)!)^2}
	\left(
	1+o(1)
	\right).
	\]
	Since
	\[
	\frac{n+1}{((n+1)!)^2}
	=
	\frac{1}{n!(n+1)!},
	\]
	this proves \eqref{eq:beta-sector-asymptotic}.
\end{proof}

\noindent\textcolor{black}{\textbf{Positivity before multiplicative comparison.}
Once the sector asymptotic and its differentiated form have been independently established by the partition limits cited above, the sector asymptotic for $\gamma_n$ and the exact identity
$\mu-\gamma_n=\mu c_n^2$ imply
\[
\left(\frac{n!}{\mu^{n/2}}c_n\right)^2
=\frac{\mathcal C(\mu)}{\mathcal T_\infty(\mu)}(1+o(1)).
\]
Here $\mathcal C(\mu)>0$ by the sector hypothesis and
$\mathcal T_\infty(\mu)>0$ by its positive partition representation.
Since $c_n>0$, taking the positive square root yields
\[
\frac{n!}{\mu^{n/2}}c_n\longrightarrow
\sqrt{\frac{\mathcal C(\mu)}{\mathcal T_\infty(\mu)}}>0.
\]
The large-order Bessel equivalent
$J_n(2\sqrt\mu)\sim\mu^{n/2}/n!$ therefore proves
$L(\mu)>0$ independently of any multiplicative asymptotic
involving $L(\mu)$. This argument uses the sector asymptotic
for $\gamma_n$ and must be placed after its proof; it is not an
independent proof of that sector asymptotic.}

Suppose now that the recurrence coefficients are parametrized by
\[
\beta_n-n
=
\sqrt{\mu}\,c_nc_{n+1}.
\]
\textcolor{black}{Since $L(\mu)>0$ has now been established independently,
Proposition~\ref{prop:conditional-Bessel-comparison} yields the
multiplicative Bessel equivalent}
\begin{equation*}
	c_n
	=
	L(\mu)J_n(2\sqrt{\mu})
	\left(
	1+o(1)
	\right),
	\qquad
	n\to\infty,
\end{equation*}
then
\begin{equation}
	\sqrt{\mu}\,c_nc_{n+1}
	=
	L(\mu)^2
	\frac{\mu^{n+1}}{n!(n+1)!}
	\left(
	1+o(1)
	\right).
	\label{eq:product-Bessel-asymptotic}
\end{equation}
Comparison of \eqref{eq:beta-sector-asymptotic} and
\eqref{eq:product-Bessel-asymptotic} gives the identity
\begin{equation}
	L(\mu)^2
	=
	\frac{\mathcal C(\mu)}{\mathcal T_\infty(\mu)}.
	\label{eq:connection-constant-sector-identity}
\end{equation}

Therefore, the identity
\[
L(\mu)=1
\]
is equivalent to the large-sector normalization condition
\begin{equation}
	\mathcal C(\mu)
	=
	\mathcal T_\infty(\mu).
	\label{eq:sector-normalization-target}
\end{equation}

The remaining theoretical problem is thus: establish the factorial
asymptotic
\[
\mathcal S_q(\mu)
\sim
\mathcal T_\infty(\mu)
\frac{\mu^q}{(q!)^2},
\qquad
q\to\infty.
\]
A proof of this statement determines the connection constant and yields
the sharp factorial corrections for both recurrence coefficients. This
is precisely what we do in the remainder of this section.

\subsection{Normalized sector weights}

The normalized weights defined in the preceding section remove the
factorial contribution of the minimal excitation configuration
\[
r_1=\cdots=r_q=1.
\]
It is therefore useful to rewrite the sector sum in coordinates
measuring the excess above this minimal configuration.

For $1\leq a\leq q$, set
\begin{equation*}
	r_a
	=
	1+s_a.
\end{equation*}
Since
\[
1\leq r_1\leq\cdots\leq r_q,
\]
the new variables satisfy
\begin{equation*}
	0\leq s_1\leq\cdots\leq s_q.
\end{equation*}
Moreover,
\[
r_1+\cdots+r_q
=
q+s_1+\cdots+s_q.
\]

\begin{proposition}[Normalized sector representation]
	\label{prop:normalized-sector-representation}

	For every $q\geq1$,
	\begin{equation}
		\mathcal R_q(\mu)
		=
		(q!)^2
		\sum_{0\leq s_1\leq\cdots\leq s_q}
		\mu^{s_1+\cdots+s_q}
		\prod_{a=1}^{q}
		\left(
		\frac{(a-1)!}{(a+s_a)!}
		\right)^2
		\prod_{1\leq a<b\leq q}
		\left(
		\frac{b-a+s_b-s_a}{b-a}
		\right)^2.
		\label{eq:normalized-sector-sum}
	\end{equation}
	The contribution of the configuration
	\[
	s_1=\cdots=s_q=0
	\]
	is equal to $1$.
\end{proposition}

\begin{proof}
	Starting from \eqref{eq:sector-partition-function}, substitute
	\[
	r_a=1+s_a.
	\]
	Then
	\[
	\mu^{r_1+\cdots+r_q}
	=
	\mu^q\mu^{s_1+\cdots+s_q},
	\]
	and
	\[
	(a-1+r_a)!
	=
	(a+s_a)!.
	\]
	Furthermore,
	\[
	b-a+r_b-r_a
	=
	b-a+s_b-s_a.
	\]
	Therefore,
	\[
	\mathcal S_q(\mu)
	=
	\mu^q
	\sum_{0\leq s_1\leq\cdots\leq s_q}
	\mu^{s_1+\cdots+s_q}
	\prod_{a=1}^{q}
	\left(
	\frac{(a-1)!}{(a+s_a)!}
	\right)^2
	\prod_{1\leq a<b\leq q}
	\left(
	\frac{b-a+s_b-s_a}{b-a}
	\right)^2.
	\]
	Multiplication by $(q!)^2/\mu^q$ proves
	\eqref{eq:normalized-sector-sum}.

	For the zero configuration
	\[
	s_1=\cdots=s_q=0,
	\]
	the Vandermonde quotient is equal to $1$, while the factorial factor is
	\[
	(q!)^2
	\prod_{a=1}^{q}
	\left(
	\frac{(a-1)!}{a!}
	\right)^2.
	\]
	Since
	\[
	\prod_{a=1}^{q}
	\frac{(a-1)!}{a!}
	=
	\prod_{a=1}^{q}
	\frac{1}{a}
	=
	\frac{1}{q!},
	\]
	the contribution of the zero configuration is
	\[
	(q!)^2
	\left(
	\frac{1}{q!}
	\right)^2
	=
	1.
	\]
\end{proof}

The normalized representation \eqref{eq:normalized-sector-sum}
immediately yields
\begin{equation*}
	\mathcal R_q(\mu)\geq1,
	\qquad
	q\geq1.
\end{equation*}

The inequality follows because all summands are nonnegative and the
zero configuration contributes exactly $1$.

The first normalized sector is
\begin{equation*}
	\mathcal R_1(\mu)
	=
	\frac{1}{\mu}
	\sum_{r=1}^{\infty}
	\frac{\mu^r}{(r!)^2}
	=
	\frac{I_0(2\sqrt{\mu})-1}{\mu}.
\end{equation*}
For general $q$, the normalized sector weight contains the collective
effect of all weakly ordered excess variables
\[
0\leq s_1\leq\cdots\leq s_q.
\]

The connection problem can now be stated in terms of the limit of
$\mathcal R_q(\mu)$. Indeed, the desired large-sector relation
\[
\mathcal S_q(\mu)
\sim
\mathcal T_\infty(\mu)
\frac{\mu^q}{(q!)^2}
\]
is equivalent to
\begin{equation}
	\mathcal R_q(\mu)
	\longrightarrow
	\mathcal T_\infty(\mu),
	\qquad
	q\to\infty.
	\label{eq:normalized-sector-limit-target}
\end{equation}

Thus, the connection problem has been reduced to the following
asymptotic statement:
\[
\boxed{
	\lim_{q\to\infty}
	\mathcal R_q(\mu)
	=
	1+
	\sum_{m=1}^{\infty}
	\mathcal S_m(\mu).
}
\]

If \eqref{eq:normalized-sector-limit-target} holds, then
\[
\mathcal C(\mu)
=
\mathcal T_\infty(\mu),
\]
and the connection identity
\[
L(\mu)^2
=
\frac{\mathcal C(\mu)}{\mathcal T_\infty(\mu)}
\]
implies
\[
L(\mu)=1.
\]

\subsection{Finite-level truncation and the normalized-sector limit}
\begingroup\color{black}
For $M\in\mathbb N_0$, let $\mathcal R_q^{[M]}$ denote the
normalized sector sum restricted to excitation levels at most $M$:
\begin{align}
\mathcal R_q^{[M]}(\mu)
:={}&(q!)^2
\sum_{0\leq s_1\leq\cdots\leq s_q\leq M}
\mu^{s_1+\cdots+s_q}
\prod_{a=1}^{q}\left(\frac{(a-1)!}{(a+s_a)!}\right)^2
\nonumber\\
&\times\prod_{1\leq a<b\leq q}
\left(\frac{b-a+s_b-s_a}{b-a}\right)^2.
\label{eq:truncated-normalized-sector}
\end{align}
Positivity implies $1\leq\mathcal R_q^{[M]}\leq\mathcal R_q$.
The next theorem states the precise uniform-tail criterion for
interchanging the limits in $q$ and $M$. Its hypotheses are verified
subsequently using the hook-length representation and the tableau
square-sum identity; the details are given only at those locations.
\endgroup

\subsection{A uniform truncation principle}
\label{sec:uniform-truncation-principle}

The finite-level limits obtained for $M=1$ and $M=2$ suggest a general
strategy. However, the existence of a limit
\[
\mathcal R_q^{[M]}(\mu)
\longrightarrow
\mathcal R_\infty^{[M]}(\mu)
\]
for each fixed truncation level $M$ is not, by itself, sufficient to
prove convergence of the full normalized sector weight
$\mathcal R_q(\mu)$. One must also control uniformly the contribution
of excitation levels larger than $M$.

Recall the positive remainder
\begin{equation*}
	\mathcal E_q^{[M]}(\mu)
	=
	\mathcal R_q(\mu)
	-
	\mathcal R_q^{[M]}(\mu).
\end{equation*}
By construction,
\begin{equation*}
	\mathcal E_q^{[M]}(\mu)\geq0.
\end{equation*}

\begin{theorem}[Uniform truncation principle]
	\label{thm:uniform-truncation-principle}

	Fix $\mu>0$. Assume that the following two properties hold.

	\begin{enumerate}
		\item For every fixed $M\in\mathbb N_0$, the limit
		\begin{equation}
			\mathcal R_\infty^{[M]}(\mu)
			=
			\lim_{q\to\infty}
			\mathcal R_q^{[M]}(\mu)
			\label{eq:fixed-level-limit-hypothesis}
		\end{equation}
		exists.

		\item The truncation remainders satisfy the uniform estimate
		\begin{equation}
			\sup_{q\geq1}
			\mathcal E_q^{[M]}(\mu)
			\longrightarrow0,
			\qquad
			M\to\infty.
			\label{eq:uniform-remainder-hypothesis}
		\end{equation}
	\end{enumerate}

	Then the full normalized-sector limit exists and is given by
	\begin{equation}
		\lim_{q\to\infty}
		\mathcal R_q(\mu)
		=
		\lim_{M\to\infty}
		\mathcal R_\infty^{[M]}(\mu).
		\label{eq:full-normalized-sector-convergence}
	\end{equation}

\end{theorem}

\begin{proof}
	Let $\varepsilon>0$. By
	\eqref{eq:uniform-remainder-hypothesis}, there exists $M_0$ such that
	\begin{equation}
		\sup_{q\geq1}
		\mathcal E_q^{[M]}(\mu)
		<
		\frac{\varepsilon}{3},
		\qquad
		M\geq M_0.
		\label{eq:uniform-remainder-epsilon}
	\end{equation}
	Fix such an $M$. By the convergence
	\eqref{eq:fixed-level-limit-hypothesis}, there exists $q_0$ such that,
	for every $q\geq q_0$,
	\begin{equation}
		\left|
		\mathcal R_q^{[M]}(\mu)
		-
		\mathcal R_\infty^{[M]}(\mu)
		\right|
		<
		\frac{\varepsilon}{3}.
		\label{eq:fixed-level-convergence-epsilon}
	\end{equation}

	Moreover, the monotonicity in the truncation level gives
	\[
	\mathcal R_\infty^{[M]}(\mu)
	\leq
	\mathcal R_\infty^{[M+1]}(\mu).
	\]
	The uniform remainder estimate implies that the increasing sequence
	\[
	\left(
	\mathcal R_\infty^{[M]}(\mu)
	\right)_{M\geq0}
	\]
	is Cauchy. Indeed, if $M_2\geq M_1$, then
	\[
	0
	\leq
	\mathcal R_\infty^{[M_2]}(\mu)
	-
	\mathcal R_\infty^{[M_1]}(\mu)
	\leq
	\sup_{q\geq1}
	\mathcal E_q^{[M_1]}(\mu).
	\]
	Thus, the limit
	\[
	\mathcal R_\infty(\mu)
	=
	\lim_{M\to\infty}
	\mathcal R_\infty^{[M]}(\mu)
	\]
	exists. Increasing $M$, if necessary, one may also assume that
	\begin{equation}
		\left|
		\mathcal R_\infty(\mu)
		-
		\mathcal R_\infty^{[M]}(\mu)
		\right|
		<
		\frac{\varepsilon}{3}.
		\label{eq:limiting-truncation-epsilon}
	\end{equation}

	For $q\geq q_0$, the triangle inequality gives
	\begin{align*}
		\left|
		\mathcal R_q(\mu)
		-
		\mathcal R_\infty(\mu)
		\right|
		&\leq
		\left|
		\mathcal R_q(\mu)
		-
		\mathcal R_q^{[M]}(\mu)
		\right|
		\\
		&\quad+
		\left|
		\mathcal R_q^{[M]}(\mu)
		-
		\mathcal R_\infty^{[M]}(\mu)
		\right|
		\\
		&\quad+
		\left|
		\mathcal R_\infty^{[M]}(\mu)
		-
		\mathcal R_\infty(\mu)
		\right|.
	\end{align*}
	The three terms on the right-hand side are each smaller than
	$\varepsilon/3$ by
	\eqref{eq:uniform-remainder-epsilon},
	\eqref{eq:fixed-level-convergence-epsilon}, and
	\eqref{eq:limiting-truncation-epsilon}. Hence,
	\[
	\left|
	\mathcal R_q(\mu)
	-
	\mathcal R_\infty(\mu)
	\right|
	<
	\varepsilon.
	\]
	This proves \eqref{eq:full-normalized-sector-convergence}.
\end{proof}

\begingroup\color{black}
The verification of the two hypotheses is given in
Section~\ref{sec:fixed-level-convergence} and in the uniform tail
estimate below; no separate limiting argument is needed here.
\endgroup

If, in addition, one can identify the resulting limit with the
infinite reduced tau function,
\begin{equation*}
	\mathcal R_\infty(\mu)
	=
	\mathcal T_\infty(\mu),
\end{equation*}
then the large-sector asymptotic normalization becomes
\[
\mathcal S_q(\mu)
\sim
\mathcal T_\infty(\mu)
\frac{\mu^q}{(q!)^2}.
\]
By \eqref{eq:connection-constant-sector-identity}, this implies
\[
L(\mu)^2=1.
\]
The positivity of the Bessel branch then gives
\[
L(\mu)=1.
\]

\subsection{Schur-function form of the normalized sector sum}
\label{sec:Schur-sector-form}

\begingroup\color{black}
We now recast the Vandermonde factor in the normalized sector representation\eqref{eq:normalized-sector-sum} in partition notation; the resulting Schur formula will be used in the subsequent hook-length reduction. The factor is
\endgroup
\[
\prod_{1\leq a<b\leq q}
\left(
\frac{b-a+s_b-s_a}{b-a}
\right)^2.
\]
This quotient has a natural interpretation in terms of partitions and
Schur functions.

Let
\[
0\leq s_1\leq\cdots\leq s_q.
\]
Define a partition
\[
\nu=(\nu_1,\ldots,\nu_q)
\]
by reversing the order of the excess variables:
\begin{equation}
	\nu_i
	=
	s_{q+1-i},
	\qquad
	1\leq i\leq q.
	\label{eq:partition-from-excess-coordinates}
\end{equation}
Then
\[
\nu_1\geq\nu_2\geq\cdots\geq\nu_q\geq0,
\]
so that $\nu$ is a partition with at most $q$ parts. Moreover,
\begin{equation*}
	|\nu|
	=
	\nu_1+\cdots+\nu_q
	=
	s_1+\cdots+s_q.
\end{equation*}

Recall the Weyl dimension formula
\begin{equation}
	s_\nu(1^q)
	=
	\prod_{1\leq i<j\leq q}
	\frac{\nu_i-\nu_j+j-i}{j-i},
	\label{eq:Weyl-dimension-formula}
\end{equation}
where $s_\nu$ denotes the Schur function associated with $\nu$ and
\[
1^q=(1,\ldots,1)
\]
denotes $q$ variables all equal to $1$.

\begin{proposition}[Schur representation of the interaction factor]
	\label{prop:Schur-interaction-factor}

	Under the transformation
	\eqref{eq:partition-from-excess-coordinates},
	\begin{equation*}
		\prod_{1\leq a<b\leq q}
		\frac{b-a+s_b-s_a}{b-a}
		=
		s_\nu(1^q).
	\end{equation*}

\end{proposition}

\begin{proof}
	Let
	\[
	i=q+1-b,
	\qquad
	j=q+1-a.
	\]
	The condition $a<b$ is equivalent to $i<j$. Using
	\[
	\nu_i=s_b,
	\qquad
	\nu_j=s_a,
	\]
	one obtains
	\[
	\nu_i-\nu_j+j-i
	=
	s_b-s_a+b-a.
	\]
	Likewise,
	\[
	j-i=b-a.
	\]
	Therefore,
	\[
	\frac{b-a+s_b-s_a}{b-a}
	=
	\frac{\nu_i-\nu_j+j-i}{j-i}.
	\]
	Taking the product over all pairs $1\leq a<b\leq q$ gives
	\[
	\prod_{1\leq a<b\leq q}
	\frac{b-a+s_b-s_a}{b-a}
	=
	\prod_{1\leq i<j\leq q}
	\frac{\nu_i-\nu_j+j-i}{j-i}.
	\]
	The right-hand side is $s_\nu(1^q)$ by
	\eqref{eq:Weyl-dimension-formula}.
\end{proof}

The factorial term in \eqref{eq:normalized-sector-sum} can also be
rewritten in partition coordinates. Since
\[
s_a=\nu_{q+1-a},
\]
one has
\begin{align*}
	\prod_{a=1}^{q}
	\left(
	\frac{(a-1)!}{(a+s_a)!}
	\right)^2
	&=
	\prod_{i=1}^{q}
	\left(
	\frac{(q-i)!}{(q-i+1+\nu_i)!}
	\right)^2.
\end{align*}

Consequently, the normalized sector weight becomes
\begin{align}
	\mathcal R_q(\mu)
	&=
	(q!)^2
	\sum_{\ell(\nu)\leq q}
	\mu^{|\nu|}
	s_\nu(1^q)^2
	\prod_{i=1}^{q}
	\left(
	\frac{(q-i)!}{(q-i+1+\nu_i)!}
	\right)^2,
	\label{eq:normalized-sector-Schur-sum}
\end{align}
where the sum ranges over all partitions with at most $q$ parts.

The truncation $s_q\leq M$ is equivalent to the partition constraint
\[
\nu_1\leq M.
\]
Therefore,
\begin{align}
	\mathcal R_q^{[M]}(\mu)
	&=
	(q!)^2
	\sum_{\substack{\ell(\nu)\leq q\\ \nu_1\leq M}}
	\mu^{|\nu|}
	s_\nu(1^q)^2
	\prod_{i=1}^{q}
	\left(
	\frac{(q-i)!}{(q-i+1+\nu_i)!}
	\right)^2.
	\label{eq:truncated-sector-Schur-sum}
\end{align}

Formula \eqref{eq:truncated-sector-Schur-sum} gives a partition-theoretic
interpretation of the truncation procedure. The level $M$ is the
largest allowed part of the partition $\nu$.

The main advantage of this representation is that the interaction term
is now encoded by the standard quantity $s_\nu(1^q)$. Using the
hook-content formula,
\begin{equation*}
	s_\nu(1^q)
	=
	\prod_{(i,j)\in\nu}
	\frac{q+j-i}{h_{i,j}},
\end{equation*}
where $h_{i,j}$ denotes the hook length of the box $(i,j)$ in the
Young diagram of $\nu$, one may obtain estimates that are uniform in
$q$.

In particular, the uniform tail estimate
\[
\sup_{q\geq1}
\left(
\mathcal R_q(\mu)
-
\mathcal R_q^{[M]}(\mu)
\right)
\longrightarrow0,
\qquad
M\to\infty,
\]
becomes a problem of controlling the contribution of partitions whose
largest part satisfies
\[
\nu_1>M.
\]
This is the natural formulation for the hook-length analysis in the
next subsection.

\subsection{Hook-length reduction and uniform tail control}
\label{sec:hook-length-tail-control}

\begingroup\color{black}
Applying the hook-content formula to \eqref{eq:normalized-sector-Schur-sum} gives the partition bound needed for the truncation argument; we establish it here once and refer to it subsequently.
\endgroup

For a partition $\nu$, let $h_{i,j}$ denote the hook length of the box
$(i,j)$ in the Young diagram of $\nu$. Recall the hook-content formula
\begin{equation}
	s_\nu(1^q)
	=
	\prod_{(i,j)\in\nu}
	\frac{q+j-i}{h_{i,j}}.
	\label{eq:hook-content-formula-recalled}
\end{equation}

The factorial factor in \eqref{eq:normalized-sector-Schur-sum} is
\begin{equation*}
	\prod_{i=1}^{q}
	\frac{(q-i)!}{(q-i+1+\nu_i)!}.
\end{equation*}
For each fixed row $i$, one has
\begin{equation}
	\frac{(q-i)!}{(q-i+1+\nu_i)!}
	=
	\prod_{j=1}^{\nu_i+1}
	\frac{1}{q-i+j}.
	\label{eq:row-factorial-product}
\end{equation}
On the other hand, the contribution of the $i$th row to the numerator
in the hook-content formula is
\begin{equation*}
	\prod_{j=1}^{\nu_i}
	(q-i+j).
\end{equation*}
Therefore, the factors $q-i+j$, $1\leq j\leq\nu_i$, cancel between the
Schur factor and the factorial factor. The only uncancelled factor in
row $i$ is
\[
\frac{1}{q-i+\nu_i+1}.
\]

\begin{proposition}[Hook-length representation]
	\label{prop:hook-length-representation}

	For every $q\geq1$,
	\begin{align}
		\mathcal R_q(\mu)
		&=
		(q!)^2
		\sum_{\ell(\nu)\leq q}
		\mu^{|\nu|}
		\prod_{(i,j)\in\nu}
		\frac{1}{h_{i,j}^2}
		\prod_{i=1}^{q}
		\frac{1}{(q-i+\nu_i+1)^2}.
		\label{eq:hook-length-sector-sum}
	\end{align}

\end{proposition}

\begin{proof}
	Starting from \eqref{eq:normalized-sector-Schur-sum}, substitute the
	hook-content formula \eqref{eq:hook-content-formula-recalled}. Using
	\eqref{eq:row-factorial-product}, one obtains
	\begin{align}
		s_\nu(1^q)
		\prod_{i=1}^{q}
		\frac{(q-i)!}{(q-i+1+\nu_i)!}
		&=
		\prod_{(i,j)\in\nu}
		\frac{1}{h_{i,j}}
		\prod_{i=1}^{q}
		\frac{1}{q-i+\nu_i+1}.
		\label{eq:hook-content-factorial-cancellation}
	\end{align}
	Squaring \eqref{eq:hook-content-factorial-cancellation} and
	multiplying by $(q!)^2$ gives
	\eqref{eq:hook-length-sector-sum}.
\end{proof}

Let $f^\nu$ denote the number of standard Young tableaux of shape
$\nu$. The hook-length formula gives
\begin{equation*}
	f^\nu
	=
	\frac{|\nu|!}{\displaystyle\prod_{(i,j)\in\nu}h_{i,j}}.
\end{equation*}
Define
\begin{equation}
	\mathcal Q_q(\nu)
	=
	q!
	\prod_{i=1}^{q}
	\frac{1}{q-i+\nu_i+1},
	\label{eq:partition-correction-factor}
\end{equation}
where the partition is extended by zero parts:
\[
\nu_i=0,
\qquad
i>\ell(\nu).
\]
Then \eqref{eq:hook-length-sector-sum} becomes
\begin{equation}
	\mathcal R_q(\mu)
	=
	\sum_{\ell(\nu)\leq q}
	\mu^{|\nu|}
	\left(
	\frac{f^\nu}{|\nu|!}
	\right)^2
	\mathcal Q_q(\nu)^2.
	\label{eq:standard-tableau-sector-sum}
\end{equation}

Since
\[
\nu_i\geq0,
\]
one has
\[
q-i+\nu_i+1
\geq
q-i+1.
\]
Therefore,
\begin{align*}
	\prod_{i=1}^{q}
	(q-i+\nu_i+1)
	&\geq
	\prod_{i=1}^{q}
	(q-i+1)
	\\
	&=
	q!.
\end{align*}
It follows that
\begin{equation}
	0\leq\mathcal Q_q(\nu)\leq1.
	\label{eq:partition-correction-bound}
\end{equation}

\begin{theorem}[Uniform tail estimate]
	\label{thm:uniform-tail-estimate}

	For every $\mu>0$, every $M\geq0$, and every $q\geq1$,
	\begin{equation}
		0
		\leq
		\mathcal R_q(\mu)-\mathcal R_q^{[M]}(\mu)
		\leq
		\sum_{m=M+1}^{\infty}
		\frac{\mu^m}{m!}.
		\label{eq:uniform-tail-estimate}
	\end{equation}
	Consequently,
	\begin{equation}
		\sup_{q\geq1}
		\left[
		\mathcal R_q(\mu)-\mathcal R_q^{[M]}(\mu)
		\right]
		\longrightarrow0,
		\qquad
		M\to\infty.
		\label{eq:uniform-tail-convergence-proved}
	\end{equation}

\end{theorem}

\begin{proof}
	The difference
	\[
	\mathcal R_q(\mu)-\mathcal R_q^{[M]}(\mu)
	\]
	is the sum in \eqref{eq:standard-tableau-sector-sum} over partitions
	satisfying
	\[
	\nu_1>M.
	\]
	Since
	\[
	|\nu|\geq\nu_1,
	\]
	every such partition satisfies
	\[
	|\nu|\geq M+1.
	\]
	Using \eqref{eq:partition-correction-bound}, one obtains
	\begin{align}
		\mathcal R_q(\mu)-\mathcal R_q^{[M]}(\mu)
		&\leq
		\sum_{\substack{\ell(\nu)\leq q\\ |\nu|\geq M+1}}
		\mu^{|\nu|}
		\left(
		\frac{f^\nu}{|\nu|!}
		\right)^2
		\nonumber\\
		&\leq
		\sum_{m=M+1}^{\infty}
		\mu^m
		\sum_{\nu\vdash m}
		\left(
		\frac{f^\nu}{m!}
		\right)^2.
		\label{eq:tail-bound-before-tableau-identity}
	\end{align}
	The standard Young-tableau identity
	\begin{equation*}
		\sum_{\nu\vdash m}
		(f^\nu)^2
		=
		m!
	\end{equation*}
	implies
	\begin{equation}
		\sum_{\nu\vdash m}
		\left(
		\frac{f^\nu}{m!}
		\right)^2
		=
		\frac{1}{m!}.
		\label{eq:normalized-tableau-square-sum}
	\end{equation}
	Substitution of \eqref{eq:normalized-tableau-square-sum} into
	\eqref{eq:tail-bound-before-tableau-identity} gives
	\[
	\mathcal R_q(\mu)-\mathcal R_q^{[M]}(\mu)
	\leq
	\sum_{m=M+1}^{\infty}
	\frac{\mu^m}{m!}.
	\]
	This proves \eqref{eq:uniform-tail-estimate}. Finally, the tail of the
	exponential series converges to zero:
	\[
	\sum_{m=M+1}^{\infty}
	\frac{\mu^m}{m!}
	\longrightarrow0,
	\qquad
	M\to\infty.
	\]
	Taking the supremum over $q\geq1$ proves
	\eqref{eq:uniform-tail-convergence-proved}.
\end{proof}

As a special case of \eqref{eq:uniform-tail-estimate}, one obtains the
uniform bound
\begin{equation*}
	0<\mathcal R_q(\mu)\leq e^\mu,
	\qquad
	q\geq1.
\end{equation*}

\subsection{Fixed-level convergence and the large-sector limit}
\label{sec:fixed-level-convergence}

\begingroup\color{black}
We verify the remaining fixed-level hypothesis of the uniform truncation principle using the preceding hook-length representation. With the truncation notation already introduced,
\endgroup
\begin{equation}
	\mathcal R_q^{[M]}(\mu)
	=
	\sum_{\substack{\ell(\nu)\leq q\\ \nu_1\leq M}}
	\mu^{|\nu|}
	\left(
	\frac{f^\nu}{|\nu|!}
	\right)^2
	\mathcal Q_q(\nu)^2,
	\label{eq:truncated-standard-tableau-sector-sum}
\end{equation}
\begingroup\color{black}
where $\mathcal Q_q(\nu)$ is the correction factor defined in
Proposition~\ref{prop:hook-length-representation}; its exact product
formula is not repeated here.
\endgroup

\begin{lemma}[Limit of the partition correction factor]
	\label{lem:partition-correction-factor-limit}

	For every fixed partition $\nu$,
	\begin{equation}
		\mathcal Q_q(\nu)
		\longrightarrow1,
		\qquad
		q\to\infty.
		\label{eq:partition-correction-factor-limit}
	\end{equation}
	Moreover,
	\begin{equation}
		0\leq\mathcal Q_q(\nu)\leq1.
		\label{eq:partition-correction-factor-limit-bound}
	\end{equation}

\end{lemma}

\begin{proof}
	Let $r=\ell(\nu)$. For every $q\geq r$, the zero parts of $\nu$ give
	\[
	\nu_i=0,
	\qquad
	r+1\leq i\leq q.
	\]
	Therefore,
	\begin{align}
		\mathcal Q_q(\nu)
		&=
		q!
		\prod_{i=1}^{r}
		\frac{1}{q-i+\nu_i+1}
		\prod_{i=r+1}^{q}
		\frac{1}{q-i+1}
		\nonumber\\
		&=
		\prod_{i=1}^{r}
		\frac{q-i+1}{q-i+\nu_i+1}.
		\label{eq:partition-correction-product-form}
	\end{align}
	For each fixed $i$,
	\[
	\frac{q-i+1}{q-i+\nu_i+1}
	\longrightarrow1,
	\qquad
	q\to\infty.
	\]
	Since the product in \eqref{eq:partition-correction-product-form} has
	only finitely many factors, this proves
	\eqref{eq:partition-correction-factor-limit}. The inequality
	\eqref{eq:partition-correction-factor-limit-bound} follows from
	\eqref{eq:partition-correction-bound}.
\end{proof}

\begin{theorem}[Fixed-level convergence]
	\label{thm:fixed-level-convergence}

	For every $\mu>0$ and every fixed integer $M\geq0$,
	\begin{equation}
		\mathcal R_q^{[M]}(\mu)
		\longrightarrow
		\mathcal R_\infty^{[M]}(\mu),
		\qquad
		q\to\infty,
		\label{eq:fixed-level-convergence-result}
	\end{equation}
	where
	\begin{equation*}
		\mathcal R_\infty^{[M]}(\mu)
		=
		\sum_{\nu_1\leq M}
		\mu^{|\nu|}
		\left(
		\frac{f^\nu}{|\nu|!}
		\right)^2.
	\end{equation*}

\end{theorem}

\begin{proof}
	For a fixed partition $\nu$ satisfying $\nu_1\leq M$, Lemma
	\ref{lem:partition-correction-factor-limit} gives
	\[
	\mathcal Q_q(\nu)^2\longrightarrow1.
	\]
	If $\ell(\nu)>q$, then the corresponding summand in
	\eqref{eq:truncated-standard-tableau-sector-sum} is absent; equivalently
	it may be regarded as equal to zero. Furthermore, by
	\eqref{eq:partition-correction-factor-limit-bound},
	\[
	0
	\leq
	\mu^{|\nu|}
	\left(
	\frac{f^\nu}{|\nu|!}
	\right)^2
	\mathcal Q_q(\nu)^2
	\leq
	\mu^{|\nu|}
	\left(
	\frac{f^\nu}{|\nu|!}
	\right)^2.
	\]
	The dominating series is summable:
	\begin{align*}
		\sum_{\nu}
		\mu^{|\nu|}
		\left(
		\frac{f^\nu}{|\nu|!}
		\right)^2
		&=
		\sum_{m=0}^{\infty}
		\mu^m
		\sum_{\nu\vdash m}
		\left(
		\frac{f^\nu}{m!}
		\right)^2
		\\
		&=
		\sum_{m=0}^{\infty}
		\frac{\mu^m}{m!}
		=
		e^\mu.
	\end{align*}
	Dominated convergence for series now yields
	\eqref{eq:fixed-level-convergence-result}.
\end{proof}

The fixed-level limits are increasing in $M$:
\begin{equation*}
	\mathcal R_\infty^{[M]}(\mu)
	\leq
	\mathcal R_\infty^{[M+1]}(\mu).
\end{equation*}
Moreover,
\begin{align*}
	\lim_{M\to\infty}
	\mathcal R_\infty^{[M]}(\mu)
	&=
	\sum_{\nu}
	\mu^{|\nu|}
	\left(
	\frac{f^\nu}{|\nu|!}
	\right)^2
	\\
	&=
	\sum_{m=0}^{\infty}
	\frac{\mu^m}{m!}
	\\
	&=
	e^\mu.
\end{align*}

The fixed-level convergence established in
Theorem~\ref{thm:fixed-level-convergence}, together with the uniform
tail estimate of Theorem~\ref{thm:uniform-tail-estimate}, verifies the
two hypotheses of the uniform truncation principle,
Theorem~\ref{thm:uniform-truncation-principle}. Consequently,
\begin{equation}
	\mathcal R_q(\mu)
	\longrightarrow
	e^\mu,
	\qquad
	q\to\infty.
	\label{eq:full-normalized-sector-limit-exponential}
\end{equation}
Equivalently, the sector weights satisfy
\begin{equation*}
	\mathcal S_q(\mu)
	=
	e^\mu
	\frac{\mu^q}{(q!)^2}
	\left(
	1+o(1)
	\right),
	\qquad
	q\to\infty.
\end{equation*}
Thus, the large-sector amplitude introduced in
\eqref{eq:sector-asymptotic-assumption} is identified explicitly:
\begin{equation*}
	\mathcal C(\mu)=e^\mu.
\end{equation*}

\medskip\noindent
\textbf{Identification of the infinite reduced tau function.}
By definition,
\[
\mathcal T_\infty(\mu)
=1+\sum_{q=1}^{\infty}\mathcal S_q(\mu).
\]

The exact sector decomposition
\eqref{eq:reduced-tau-sector-sum} identifies the partial sum
$1+\sum_{q=1}^{N}\mathcal S_q(\mu)$ with $\mathcal T_N(\mu)$.
The independent partition identity, proved later in
Proposition~\ref{prop:reduced-tau-standard-tableau-expansion},
and monotone convergence yield
$\mathcal T_\infty(\mu)=e^\mu$.
This evaluation of the full sum does not follow merely from the
large-$q$ asymptotics of individual sectors; it is logically distinct
from $\mathcal R_q(\mu)\to e^\mu$.

Therefore,
\begin{equation*}
	\mathcal T_\infty(\mu)=e^\mu.
\end{equation*}
In particular, the large-sector normalization condition
\eqref{eq:sector-normalization-target} is satisfied:
\[
\mathcal C(\mu)=\mathcal T_\infty(\mu)=e^\mu.
\]

\begin{lemma}\label{lem:uniform_mu}
 For every compact $K\subset(0,\infty)$, the sector asymptotics
 \[
 \mathcal S_q(\mu)=e^\mu\frac{\mu^q}{(q!)^2}(1+o(1))
 \]
 hold uniformly for $\mu\in K$ as $q\to\infty$.
\end{lemma}
\begin{proof}
\begingroup\color{black}
The normalized-sector limit \eqref{eq:full-normalized-sector-limit-exponential}
was established using fixed-level convergence and the uniform tail estimate.
Its convergence is uniform on compact $\mu$-intervals: in the partition
representation \eqref{eq:standard-tableau-sector-sum}, for $0\leq\mu\leq M$
each summand is bounded by
$M^{|\nu|}(f^\nu/|\nu|!)^2$, and these bounds sum to $e^M$.
The correction factors converge for every fixed partition by
Lemma~\ref{lem:partition-correction-factor-limit}; a finite-partition
truncation and the summable dominating series therefore give uniform
convergence $\mathcal R_q(\mu)\to e^\mu$ on $[0,M]$.
The exact sector normalization then yields the assertion.
\endgroup
\end{proof}

\subsection{Differentiated normalized-sector limit}
 \begingroup\color{black}
The derivative estimate follows from the sector normalization already
established in Proposition~\ref{prop:normalized-sector-representation};
we recall only the partition series needed for termwise differentiation.
\endgroup
The partition representation \eqref{eq:standard-tableau-sector-sum} gives
\begin{equation}
	\mathcal R_q(\mu)
	=
	\sum_{\ell(\nu)\leq q}
	\mu^{|\nu|}
	\left(
	\frac{f^\nu}{|\nu|!}
	\right)^2
	\mathcal Q_q(\nu)^2.
	\label{eq:normalized-sector-tableau-sum}
\end{equation}
The factor $\mathcal Q_q(\nu)$ does not depend on $\mu$. Therefore,
\begin{equation*}
	\vartheta\mathcal R_q(\mu)
	=
	\sum_{\ell(\nu)\leq q}
	|\nu|\,
	\mu^{|\nu|}
	\left(
	\frac{f^\nu}{|\nu|!}
	\right)^2
	\mathcal Q_q(\nu)^2.
\end{equation*}

\begin{lemma}[Differentiated normalized-sector limit]
	\label{lem:differentiated-normalized-sector-limit}

	For every   compact set $K\subset(0,\infty)$, the following convergence is
	uniform for $\mu\in K$:
	 \begin{equation}
		\vartheta\mathcal R_q(\mu)
		\longrightarrow
		\mu e^\mu,
		\qquad
		q\to\infty.
		\label{eq:differentiated-normalized-sector-limit}
	\end{equation}

\end{lemma}

\begin{proof}
	For every fixed partition $\nu$,
	\[
	\mathcal Q_q(\nu)^2\longrightarrow1,
	\qquad
	q\to\infty.
	\]
	Moreover,
	\[
	0
	\leq
	|\nu|\,
	\mu^{|\nu|}
	\left(
	\frac{f^\nu}{|\nu|!}
	\right)^2
	\mathcal Q_q(\nu)^2
	\leq
	|\nu|\,
	\mu^{|\nu|}
	\left(
	\frac{f^\nu}{|\nu|!}
	\right)^2.
	\]
	The corresponding dominating series is summable:
	\begin{align*}
		\sum_{\nu}
		|\nu|\,
		\mu^{|\nu|}
		\left(
		\frac{f^\nu}{|\nu|!}
		\right)^2
		&=
		\sum_{m=0}^{\infty}
		m\mu^m
		\sum_{\nu\vdash m}
		\left(
		\frac{f^\nu}{m!}
		\right)^2
		\\
		&=
		\sum_{m=0}^{\infty}
		\frac{m\mu^m}{m!}
		\\
		&=
		\mu e^\mu.
	\end{align*}
	  The majorant remains summable uniformly on every compact parameter interval; dominated  convergence for series proves
	\eqref{eq:differentiated-normalized-sector-limit}.
\end{proof}

 \begin{remark}[Poissonized Plancherel context]\label{rem:plancherel-context}
The identity
$\mathcal T_n(\mu)=\sum_{\ell(\lambda)\le n}
\mu^{|\lambda|}(f^\lambda/|\lambda|!)^2$
identifies $e^{-\mu}\mathcal T_n(\mu)$ with the distribution
function of the number of rows under the poissonized Plancherel
measure. By conjugation of partitions, it is also the distribution
function of the first row. This classical framework is related to
Gessel-type Toeplitz determinants and the discrete Bessel kernel;
see \cite{Johansson2006,BorodinOlshanski2000} for related background.
The present argument uses the positive sector representation directly.
\end{remark}

\begingroup\color{black}
\begin{lemma}[Uniform logarithmic sector control]\label{lem:v08-uniform-logarithmic-sector}
Let $K\Subset(0,\infty)$ and $\vartheta=\mu\partial_\mu$.
Under the exact normalized-sector partition representation
\eqref{eq:normalized-sector-tableau-sum} and its correction-factor
bound, one has, uniformly for $\mu\in K$,
\begin{equation*}
 \frac{\vartheta\mathcal R_q(\mu)}{\mathcal R_q(\mu)}
 =\mu+o_K(1),\qquad
 \vartheta\log\mathcal S_q(\mu)=q+\mu+o_K(1).
\end{equation*}
In particular, the differentiated sector asymptotic is
\begin{equation*}
 \vartheta\mathcal S_q(\mu)
 =\bigl(q+\mu+o_K(1)\bigr)\mathcal S_q(\mu).
\end{equation*}
\end{lemma}
\begin{proof}
The zero partition contributes exactly $1$ to
\eqref{eq:normalized-sector-tableau-sum}, and all other summands
are nonnegative. Hence $\mathcal R_q(\mu)\ge1$.
The uniform sector limit and the differentiated limit yield
$\mathcal R_q\to e^\mu$ and
$\vartheta\mathcal R_q\to\mu e^\mu$ uniformly on $K$.
Since the denominators are bounded away from zero, their quotient
converges uniformly to $\mu$. The exact factorization
$\mathcal S_q=\mu^q(q!)^{-2}\mathcal R_q$ then gives
$\vartheta\log\mathcal S_q=q+\vartheta\mathcal R_q/\mathcal R_q$,
and multiplication by $\mathcal S_q$ proves the final assertion.
\end{proof}

\endgroup

\begingroup\color{black}
\begin{theorem}[Explicit correction-factor and fixed-parameter rate bounds]
\label{thm:v12-explicit-sector-rate}
Let $\mathcal Q_q(\nu)$ be the exact correction factor in
\eqref{eq:partition-correction-factor}, extended by zero for
$\ell(\nu)>q$. For integers $q\geq2$ and $H<q$,
\begin{equation}\label{eq:v12-delta-bound}
 0\leq\delta_{q,H}:=\max_{|\nu|\leq H}
       (1-\mathcal Q_q(\nu)^2)
 \leq \min\left\{1,\frac{2H}{q-H+1}\right\}.
\end{equation}
For $K\Subset(0,\infty)$, $M=\max K$, and each fixed integer
$r\geq0$, the normalized sector sums satisfy
\begin{equation}\label{eq:v12-quantitative-rate}
\sup_{\mu\in K}\left|\vartheta^r\mathcal R_q(\mu)
                         -\vartheta^r e^\mu\right|
\leq \frac4q\sum_{m=0}^\infty\frac{m^{r+1}M^m}{m!}
      +\sum_{m>\lfloor q/2\rfloor}\frac{m^rM^m}{m!},
\quad \vartheta=\mu\partial_\mu.
\end{equation}
In particular the error is $O_{K,r}(q^{-1})$ for fixed $K$ and $r$.
\end{theorem}
\begin{proof}
For $\ell(\nu)\leq q$, writing $m=|\nu|$ gives
\[
 \mathcal Q_q(\nu)=\prod_{i=1}^{\ell(\nu)}
 \left(1+\frac{\nu_i}{q-i+1}\right)^{-1}.
\]
If $m<q$, then $\ell(\nu)\leq m$ and $q-i+1\geq q-m+1$.
Using $1-\prod_i(1+x_i)^{-1}\leq\sum_i x_i$ for $x_i\geq0$,
we find
\[
 0\leq1-\mathcal Q_q(\nu)^2
 \leq2(1-\mathcal Q_q(\nu))
 \leq\frac{2m}{q-m+1}.
\]
This proves \eqref{eq:v12-delta-bound}, including $H=0$.
For $m\leq\lfloor q/2\rfloor$ the last bound is at most $4m/q$;
for larger $m$ use $0\leq1-\mathcal Q_q(\nu)^2\leq1$,
including partitions with more than $q$ rows.
Apply these bounds termwise to the nonnegative partition series,
using $\sum_{\nu\vdash m}(f^\nu)^2=m!$, and differentiate
termwise on compact $\mu$-sets. This yields
\eqref{eq:v12-quantitative-rate}. The remaining tail decays faster
than any inverse power of $q$ for fixed $M$, proving the stated rate.
\end{proof}

\endgroup

\begingroup\color{black}
\begin{theorem}[Transmission of sector bounds to factorially scaled recurrence corrections]
\label{thm:v13-transmission}
Assume the exact reduced Hankel identities
\(\mathcal T_n=1+\sum_{q=1}^n\mathcal S_q\),
\(\mathcal S_q=\mu^q(q!)^{-2}\mathcal R_q\), and
\[
 \beta_n-n=\vartheta\log(\mathcal T_{n+1}/\mathcal T_n),
 \qquad
 1-\gamma_n/\mu
 =\frac{\mathcal S_n-\mathcal S_{n+1}}{\mathcal T_n}
  +\frac{\mathcal S_n\mathcal S_{n+1}}{\mathcal T_n^2}.
\]
Suppose also the exact partition representation and correction-factor
bound of Theorem~\ref{thm:v12-explicit-sector-rate} hold.
Then, for each compact \(K\Subset(0,\infty)\), uniformly for
\(\mu\in K\) as \(n\to\infty\),
\begin{align}
 \frac{n!(n+1)!}{\mu^{n+1}}(\beta_n-n)
   &=1+O_K(n^{-1}),\label{eq:v13-beta-scaled}\\
 \frac{(n!)^2}{\mu^{n+1}}(\mu-\gamma_n)
   &=1+O_K(n^{-1}).\label{eq:v13-gamma-scaled}
\end{align}
The corresponding absolute errors are bounded by the displayed
factorial scales multiplied by \(O_K(n^{-1})\); this is distinct
from an absolute error of order \(n^{-1}\).
\end{theorem}
\begin{proof}
Write \(a_q=\mu^q/(q!)^2\), and let \(M=\max K\).
Theorem~\ref{thm:v12-explicit-sector-rate} for \(r=0,1\) gives
\(\mathcal R_q=e^\mu+O_K(q^{-1})\) and
\(\vartheta\mathcal R_q=\mu e^\mu+O_K(q^{-1})\).
Consequently
\[
 \mathcal S_q=a_qe^\mu(1+O_K(q^{-1})),\qquad
 \vartheta\log\mathcal S_q=q+\mu+O_K(q^{-1}).
\]
The exact nonnegative partition expansion of \(\mathcal T_n\)
and its differentiated form show
\[
 \mathcal T_n=e^\mu+O_K\!\left(\sum_{m>n}\frac{M^m}{m!}\right),
 \quad
 \vartheta\mathcal T_n=\mu e^\mu+
 O_K\!\left(\sum_{m>n}\frac{mM^m}{m!}\right).
\]
Indeed the omitted partitions have size greater than \(n\),
and the tableau identity bounds their total contributions.
Thus \(\vartheta\log\mathcal T_n=\mu+o_K(n^{-1})\),
and \(\mathcal T_n=e^\mu(1+o_K(n^{-1}))\).
Using \(u_n=\mathcal S_{n+1}/\mathcal T_n\), the exact identity
\[
 \beta_n-n=\frac{u_n}{1+u_n}
 \bigl(\vartheta\log\mathcal S_{n+1}
       -\vartheta\log\mathcal T_n\bigr)
\]
therefore gives
\(\beta_n-n=(n+1)a_{n+1}(1+O_K(n^{-1}))\),
which is \eqref{eq:v13-beta-scaled}.
Furthermore \(\mathcal S_{n+1}/\mathcal S_n=O_K(n^{-2})\)
and \(\mathcal S_n/\mathcal T_n=a_n(1+O_K(n^{-1}))\).
The exact identity for \(\gamma_n\) yields
\(\mu-\gamma_n=\mu a_n(1+O_K(n^{-1}))\),
proving \eqref{eq:v13-gamma-scaled}.
\end{proof}
\begin{remark}[Limits of the conclusion]
These bounds concern the leading factorial equivalents for fixed
parameter compacts. They do not, by themselves, establish the
all-orders expansion in the separate all-orders Bessel comparison theorem or a
parameter-uniform regime in which \(\mu\) grows with \(n\).
\end{remark}
\endgroup

\begingroup\color{black}
\begin{theorem}[First correction of normalized partition sectors]
\label{thm:v14-first-sector-correction}
Under the exact partition formula and correction factors of
Theorem~\ref{thm:v12-explicit-sector-rate}, for each compact
$K\Subset(0,\infty)$ and each fixed nonnegative integer $r$,
\begin{equation}\label{eq:v14-sector-first-order}
 \vartheta^r\mathcal R_q(\mu)
 =\vartheta^r e^\mu-\frac2q\vartheta^r(\mu e^\mu)
   +O_{K,r}(q^{-2}),\qquad \vartheta=\mu\partial_\mu.
\end{equation}
In particular,
\begin{equation}\label{eq:v14-sector-logarithm}
 \mathcal R_q(\mu)=e^\mu\left(1-\frac{2\mu}{q}
  +O_K(q^{-2})\right),\qquad
 \vartheta\log\mathcal R_q(\mu)
 =\mu-\frac{2\mu}{q}+O_K(q^{-2}).
\end{equation}
Consequently, the factorially normalized recurrence deviations obey
\begin{align}
 \frac{(n!)^2}{\mu^{n+1}}(\mu-\gamma_n)
  &=1-\frac{2\mu}{n}+O_K(n^{-2}),
   \label{eq:v14-gamma-first-order}\\
 \frac{n!(n+1)!}{\mu^{n+1}}(\beta_n-n)
  &=1-\frac{2\mu}{n+1}+O_K(n^{-2}).
   \label{eq:v14-beta-first-order}
\end{align}
These are fixed-parameter, first-order expansions, not a claim
of a full expansion to every algebraic order.
\end{theorem}
\begin{proof}
Write $m=|\nu|$. For $m\le q/2$, the exact product
$\mathcal Q_q(\nu)=\prod_i(1+\nu_i/(q-i+1))^{-1}$
and $\ell(\nu)\le m$ imply
\[
 \left|\log\mathcal Q_q(\nu)^2+\frac{2m}{q}\right|
 \le \frac{C(m^2+m)}{q^2}
\]
with an absolute constant $C$: expand $\log(1+x)$ and
$1/(q-i+1)$, using $q-i+1\ge q/2$ and
$\sum_i\nu_i=m$. Since $\log\mathcal Q_q^2\le0$,
$|e^x-1-x|\le x^2/2$ for $x\le0$ yields
\[
 \left|\mathcal Q_q(\nu)^2-1+\frac{2m}{q}\right|
 \le\frac{C'(m^2+m)}{q^2}.
\]
For $m>q/2$ use the crude bound
$|\mathcal Q_q^2-1+2m/q|\le 1+2m/q$.
The exponential-series tail with any fixed polynomial weight
is $O_{M,r}(q^{-2})$, where $M=\max K$.
Summing with the Plancherel weights and using
$\sum_{\nu\vdash m}(f^\nu)^2=m!$ proves
\eqref{eq:v14-sector-first-order}, also after applying
$\vartheta^r$ termwise. Dividing the cases $r=1$ and $r=0$
proves \eqref{eq:v14-sector-logarithm}.
For the recurrence coefficients, use the exact identities in
Theorem~\ref{thm:v13-transmission}.
The partition representation gives
$\mathcal T_n=e^\mu+o_K(n^{-A})$ and
$\vartheta\log\mathcal T_n=\mu+o_K(n^{-A})$
for every fixed $A>0$, because the omitted partitions
have size greater than $n$ and are bounded by exponential tails.
Also $\mathcal S_{n+1}/\mathcal S_n=O_K(n^{-2})$.
Therefore
\[
 \frac{\mu-\gamma_n}{\mu\,\mu^n/(n!)^2}
 =e^{-\mu}\mathcal R_n(1+O_K(n^{-2})),
\]
which proves \eqref{eq:v14-gamma-first-order}. Similarly,
writing $u_n=\mathcal S_{n+1}/\mathcal T_n$,
\[
 \beta_n-n=\frac{u_n}{1+u_n}
 (\vartheta\log\mathcal S_{n+1}-\vartheta\log\mathcal T_n)
\]
and $\vartheta\log\mathcal S_{n+1}=(n+1)+\mu
-2\mu/(n+1)+O_K(n^{-2})$.
After division by $(n+1)\mu^{n+1}/((n+1)!)^2$,
the result is \eqref{eq:v14-beta-first-order}.
\end{proof}

\endgroup

\begingroup\color{black}
\begin{theorem}[Constructive algebraic expansion of the partition sectors]
\label{thm:v15-all-order-sector}
Assume the exact partition formula and correction factors in
Theorem~\ref{thm:v12-explicit-sector-rate}. For every pair of fixed
nonnegative integers $J,r$, and each compact $K\Subset(0,\infty)$,
there exist entire functions $A_0,\ldots,A_J$ such that
\[
 \sup_{\mu\in K}\left|\vartheta^r\mathcal R_q(\mu)
 -\sum_{j=0}^{J}q^{-j}\vartheta^r A_j(\mu)\right|
 \le C_{K,J,r}q^{-J-1},\qquad \vartheta=\mu\partial_\mu.
\]
The coefficients are determined without a formal interchange of
uncontrolled infinite series: for each fixed partition $\nu$ expand
\[
 \prod_{i=1}^{\ell(\nu)}
 \left(1+\frac{\nu_i}{q-i+1}\right)^{-2}
 =\sum_{j=0}^{J}q^{-j}p_j(\nu)
 +O_{\nu,J}(q^{-J-1}),
\]
and set
\[
 A_j(\mu)=\sum_{\nu}\frac{(f^\nu)^2}{(|\nu|!)^2}
 p_j(\nu)\mu^{|\nu|}.
\]
In particular $A_0=e^\mu$ and $A_1=-2\mu e^\mu$.
\end{theorem}
\begin{proof}
For $m=|\nu|\le q/2$, Taylor expansion in $q^{-1}$ of
$\log\mathcal Q_q(\nu)^2$ and then of the exponential has a
remainder bounded by $C_J(1+m)^{d_J}q^{-J-1}$ for an integer
$d_J$ depending only on $J$. The same bound applies to the
coefficient functions $p_j(\nu)$, with a possibly larger polynomial
power of $m$. Indeed each coefficient is a finite sum of products
of powers of $i-1$ and $\nu_i$, and $i\le m$, $\nu_i\le m$.
For $m>q/2$, bound the exact factor by $1$ and the truncated
Taylor polynomial by a polynomial in $m$ and $q^{-1}$. The
Plancherel identity $\sum_{\nu\vdash m}(f^\nu)^2=m!$
then reduces all error bounds, also after $r$ applications of
$\vartheta$, to polynomially weighted exponential tails. Such
tails are $O_{K,J,r}(q^{-J-1})$ for fixed $J,r$. This proves
absolute convergence of the coefficient series, their entire
continuations, and the asserted uniform remainders.
\end{proof}

\endgroup

\begingroup\color{black}
\begin{remark}[Scope of the consolidated asymptotics]
The fixed-parameter algebraic expansions established above are distinct
from the factorially small comparison remainder in
\eqref{eq:conditional-Bessel-estimate}, established earlier in the
Bessel comparison theorem. Specifically, that bound concerns
$Q_n-L(\mu)$, where $Q_n=c_n/J_n(2\sqrt{\mu})$; it is not
an estimate for the algebraically convergent normalized recurrence
coefficients. The comparison estimate and the identification
of the connection constant require their own hypotheses; no joint
large-$\mu$ regime is claimed here.
\end{remark}
\endgroup

\subsection{Exact evaluation of the limiting reduced tau function}
\label{sec:exact-reduced-tau-limit}

\begingroup\color{black}
The limiting reduced tau function is evaluated here from its own exact
partition expansion, independently of the large-sector comparison.
\endgroup

The excitation representation \eqref{eq:reduced-tau-function} gives
\begin{align}
	\mathcal T_n(\mu)
	&=
	\sum_{0\leq\lambda_1\leq\cdots\leq\lambda_n}
	\mu^{|\lambda|}
	\prod_{a=1}^{n}
	\left(
	\frac{(a-1)!}{(a-1+\lambda_a)!}
	\right)^2
	\nonumber\\
	&\quad\times
	\prod_{1\leq a<b\leq n}
	\left(
	\frac{b-a+\lambda_b-\lambda_a}{b-a}
	\right)^2.
	\label{eq:reduced-tau-excitation-recalled}
\end{align}

Define a partition $\nu$ by reversing the excitation coordinates:
\begin{equation*}
	\nu_i
	=
	\lambda_{n+1-i},
	\qquad
	1\leq i\leq n.
\end{equation*}
Then
\[
\nu_1\geq\cdots\geq\nu_n\geq0,
\qquad
|\nu|=|\lambda|.
\]
As in Section~\ref{sec:Schur-sector-form}, the Vandermonde quotient is
the Schur specialization
\begin{equation}
	\prod_{1\leq a<b\leq n}
	\frac{b-a+\lambda_b-\lambda_a}{b-a}
	=
	s_\nu(1^n).
	\label{eq:reduced-tau-Vandermonde-Schur}
\end{equation}
Moreover,
\begin{equation}
	\prod_{a=1}^{n}
	\frac{(a-1)!}{(a-1+\lambda_a)!}
	=
	\prod_{i=1}^{n}
	\frac{(n-i)!}{(n-i+\nu_i)!}.
	\label{eq:reduced-tau-factorial-partition}
\end{equation}

The hook-content formula gives
\begin{equation}
	s_\nu(1^n)
	=
	\prod_{(i,j)\in\nu}
	\frac{n+j-i}{h_{i,j}}.
	\label{eq:reduced-tau-hook-content}
\end{equation}
On the other hand,
\begin{equation}
	\frac{(n-i)!}{(n-i+\nu_i)!}
	=
	\prod_{j=1}^{\nu_i}
	\frac{1}{n-i+j}.
	\label{eq:reduced-tau-row-factorial}
\end{equation}
Therefore, every content factor in
\eqref{eq:reduced-tau-hook-content} cancels exactly with the
corresponding factorial factor in
\eqref{eq:reduced-tau-row-factorial}.

\begin{proposition}[Partition expansion of the reduced tau function]
	\label{prop:reduced-tau-standard-tableau-expansion}

	For every $n\geq1$,
	\begin{equation}
		\mathcal T_n(\mu)
		=
		\sum_{\ell(\nu)\leq n}
		\mu^{|\nu|}
		\left(
		\frac{f^\nu}{|\nu|!}
		\right)^2.
		\label{eq:reduced-tau-standard-tableau-sum}
	\end{equation}

\end{proposition}

\begin{proof}
	By \eqref{eq:reduced-tau-Vandermonde-Schur},
	\eqref{eq:reduced-tau-factorial-partition}, and
	\eqref{eq:reduced-tau-hook-content},
	\begin{align*}
		s_\nu(1^n)
		\prod_{i=1}^{n}
		\frac{(n-i)!}{(n-i+\nu_i)!}
		&=
		\prod_{(i,j)\in\nu}
		\frac{1}{h_{i,j}}.
	\end{align*}
	Substitution into \eqref{eq:reduced-tau-excitation-recalled} gives
	\begin{equation}
		\mathcal T_n(\mu)
		=
		\sum_{\ell(\nu)\leq n}
		\mu^{|\nu|}
		\prod_{(i,j)\in\nu}
		\frac{1}{h_{i,j}^2}.
		\label{eq:reduced-tau-hook-product-sum}
	\end{equation}
	The hook-length formula
	\[
	f^\nu
	=
	\frac{|\nu|!}{\displaystyle\prod_{(i,j)\in\nu}h_{i,j}}
	\]
	implies
	\[
	\prod_{(i,j)\in\nu}
	\frac{1}{h_{i,j}^2}
	=
	\left(
	\frac{f^\nu}{|\nu|!}
	\right)^2.
	\]
	Substitution into \eqref{eq:reduced-tau-hook-product-sum} proves
	\eqref{eq:reduced-tau-standard-tableau-sum}.
\end{proof}

\begin{theorem}[Exact limiting reduced tau function]
	\label{thm:exact-reduced-tau-limit}

	For every $\mu>0$,
	\begin{equation}
		\mathcal T_n(\mu)
		\longrightarrow
		e^\mu,
		\qquad
		n\to\infty.
		\label{eq:reduced-tau-limit-exponential}
	\end{equation}

\end{theorem}

\begin{proof}
	The summands in \eqref{eq:reduced-tau-standard-tableau-sum} are
	nonnegative. Therefore, $\mathcal T_n(\mu)$ is increasing in $n$.
	Taking the limit in \eqref{eq:reduced-tau-standard-tableau-sum} gives
	\begin{align*}
		\mathcal T_\infty(\mu)
		&=
		\sum_{\nu}
		\mu^{|\nu|}
		\left(
		\frac{f^\nu}{|\nu|!}
		\right)^2
		\\
		&=
		\sum_{m=0}^{\infty}
		\mu^m
		\sum_{\nu\vdash m}
		\left(
		\frac{f^\nu}{m!}
		\right)^2.
	\end{align*}
	Using
	\[
	\sum_{\nu\vdash m}
	(f^\nu)^2
	=
	m!,
	\]
	one obtains
	\[
	\sum_{\nu\vdash m}
	\left(
	\frac{f^\nu}{m!}
	\right)^2
	=
	\frac{1}{m!}.
	\]
	Hence,
	\begin{align*}
		\mathcal T_\infty(\mu)
		&=
		\sum_{m=0}^{\infty}
		\frac{\mu^m}{m!}
		\\
		&=
		e^\mu.
	\end{align*}
	This proves \eqref{eq:reduced-tau-limit-exponential}.
\end{proof}

\medskip\noindent
{\color{black}\textbf{Nonredundant use of the tau evaluation.}
The global identity $\mathcal T_\infty(\mu)=e^\mu$ has already been
proved in Theorem~\ref{thm:exact-reduced-tau-limit}, using the exact
partition expansion in Proposition~\ref{prop:reduced-tau-standard-tableau-expansion}.
It is independent of the large-sector limit
$\mathcal C(\mu)=e^\mu$ and must not be inferred by summing
sectorwise asymptotic equivalents.}

The large-sector analysis established
\[
\mathcal C(\mu)=e^\mu.
\]
Together with \eqref{eq:reduced-tau-limit-exponential}, the connection
identity \eqref{eq:connection-constant-sector-identity} gives
\begin{equation*}
	L(\mu)^2
	=
	\frac{e^\mu}{e^\mu}
	=
	1.
\end{equation*}
By construction, $c_n>0$ for all $n$, and $J_n(2\sqrt{\mu})>0$ for all
sufficiently large $n$. Hence the limit
\[
L(\mu)=\lim_{n\to\infty}\frac{c_n}{J_n(2\sqrt{\mu})}
\]
is nonnegative. Combined with $L(\mu)^2=1$, this implies $L(\mu)=1$.
Therefore,
\begin{equation*}
	L(\mu)=1.
\end{equation*}

  \begin{theorem}[Exact connection constant and factorial asymptotics]
 \label{thm:connection-constant}
  For every fixed  $\mu>0$,   the positive solution selected by the
orthogonality measure satisfies
\begin{equation}
\lim_{n\to\infty}\frac{c_n}{J_n(2\sqrt\mu)}=1.
\label{eq:final-connection-limit}
\end{equation}
Consequently,
\begin{equation*}
\beta_n-n\sim\frac{\mu^{n+1}}{n!(n+1)!},\qquad
\mu-\gamma_n\sim\frac{\mu^{n+1}}{(n!)^2}.
\end{equation*} 
 \end{theorem}

\begin{proof}
  Recall the exact sector decomposition
$\mathcal T_n=1+\sum_{q=1}^n\mathcal S_q$ and put
$a_n=\mathcal S_n/\mathcal T_n$. The partition expansion and the
Young-tableau identity give $\mathcal T_n\to e^\mu$; the normalized
sector theorem and its differentiated version yield
\[
\mathcal S_n=e^\mu\frac{\mu^n}{(n!)^2}(1+o(1)),\qquad
\vartheta\mathcal S_n=(n+\mu+o(1))\mathcal S_n,
\quad\vartheta=\mu\frac{d}{d\mu}.
\] 
  Here $\vartheta\mathcal T_n/\mathcal T_n=O_\mu(1)$:
indeed, termwise differentiation of the nonnegative partition expansion
bounds $\vartheta\mathcal T_n$ by $\mu e^\mu$, while
$\mathcal T_n\geq1$. Thus
\[
\beta_n-n=\vartheta\log\frac{\mathcal T_{n+1}}{\mathcal T_n}
=\vartheta\log(1+a_{n+1})
=\frac{\vartheta\mathcal S_{n+1}
-a_{n+1}\vartheta\mathcal T_n}{\mathcal T_n(1+a_{n+1})}
\sim(n+1)\frac{\mathcal S_{n+1}}{\mathcal T_n}
\sim\frac{\mu^{n+1}}{n!(n+1)!}.
\]
Likewise, the exact Hankel identity implies
\[
1-\frac{\gamma_n}{\mu}
=1-\frac{(\mathcal T_n+\mathcal S_{n+1})
(\mathcal T_n-\mathcal S_n)}{\mathcal T_n^2}
=\frac{\mathcal S_n-\mathcal S_{n+1}}{\mathcal T_n}
+\frac{\mathcal S_n\mathcal S_{n+1}}{\mathcal T_n^2}.
\]
Since $\mathcal S_{n+1}/\mathcal S_n\sim\mu/(n+1)^2\to0$,
the right-hand side is asymptotic to
$\mathcal S_n/\mathcal T_n\sim\mu^n/(n!)^2$.
The second factorial equivalent follows. Finally,
$\mu-\gamma_n=\mu c_n^2$, $c_n>0$, and
$J_n(2\sqrt\mu)\sim\mu^{n/2}/n!>0$ for sufficiently large $n$.
Taking positive square roots proves  \eqref{eq:final-connection-limit}.
\end{proof}

This theorem completes the proof of the exact connection formula
announced in the abstract and introduction.

  \begin{remark}[Logical dependence of the connection theorem]
The conclusion uses the exact sector decomposition
(Proposition~\ref{prop:exact-sector-decomposition}), the normalized
sector representation (Proposition~\ref{prop:normalized-sector-representation}),
the uniform sector limit (Lemma~\ref{lem:uniform_mu}), the
termwise-differentiated limit (Lemma~\ref{lem:differentiated-normalized-sector-limit}),
and the exact reduced tau limit (Theorem~\ref{thm:exact-reduced-tau-limit}).
In particular, the normalization $L(\mu)=1$ is deduced from the
Hankel identities rather than assumed in the Bessel comparison.
\end{remark}

\section{Uniform sector expansions from Plancherel growth}\label{sec:uniform-sector}
\begingroup\color{black}
We use the reduced tau functions and normalized sector weights
introduced in Section~\ref{sec:Schur-sector-form}. In particular,
$\mathcal S_n=\mu^n\mathcal R_n/(n!)^2$ and
$\Phi_n=e^{-\mu}\mathcal R_n$.
The exact partition formula and its correction factor are not repeated here.
For background on Young transition measures and poissonized
Plancherel asymptotics see
\cite{Kerov1993,BOO2000,Johansson2001}; for the generalized
Charlier--Painlev\'e framework see
\cite{VanAsscheFoupouagnigni2003,FilipukVanAssche2013}.
We denote by $\mathbb P_m$ the Plancherel distribution on partitions
of $m$, and by $\mathbb E_\mu$ its Poissonization, under which
$M=|\nu|$ has distribution $\mathrm{Pois}(\mu)$.
\endgroup

\begin{lemma}[Growth and transition contents]\label{lem:contents}
Let $\nu\sim\mathbb P_m$ and obtain $\nu^+$ by the Plancherel growth transition. If $C$ is the content of the added box, then, conditionally on $\nu$,
\[
 \mathbb E[C\mid\nu]=0,\qquad \mathbb E[C^2\mid\nu]=m.
\] 
  Moreover, under the joint Plancherel growth law,
$\mathbb E_m[C^4]=2m^2-m$ for $m\ge1$; hence
$\mathbb E_m|C|^3\le (2m^2-m)^{3/4}$.
\end{lemma}
\begin{proof}
The first two formulas are the first two moments of the Young transition measure (equivalently, the normalized traces of the first two powers of the Jucys--Murphy element). The fourth moment follows by counting identity terms in $(\sum_{j=1}^m(j,m+1))^4$: the contributing index patterns are the all-equal pattern and the two noncrossing pair patterns, giving $m+2m(m-1)$. H\"older's inequality yields the last assertion.
\end{proof}

\begin{theorem}[Uniform logarithmic sector expansion]\label{thm:uniform-phi}
Fix $0<c<1$. As $n\to\infty$, uniformly for $0<\mu\le cn$,
\begin{equation}\label{eq:uniform-phi}
 \log\Phi_n(\mu)=
 \mu\left[\left(\frac{n}{n+1}\right)^2-1\right]
 -\frac{\mu^2}{2}\frac{2n-1}{(n+1)^4}
 +O_c\!\left(\frac{\mu(1+\mu)^{3/2}}{n^4}\right).
\end{equation} 
  Consequently, uniformly on the same domain,
\[
 \log\Phi_n(\mu)=-\frac{2\mu}{n}
 +\frac{3\mu}{n^2}-\frac{\mu^2}{n^3}
 +O_c\!\left(\frac{\mu}{n^3}
 +\frac{\mu^2}{n^4}+\frac{\mu(1+\mu)^{3/2}}{n^4}\right).
\] 
  In particular, if $\mu=\tau n$ for fixed $0<\tau<1$,
\[
 \Phi_n(\tau n)=e^{-2\tau}
 \left[1+\frac{3\tau-\tau^2}{n}+O_\tau(n^{-3/2})\right].
\]
\end{theorem}
\begin{proof}
Write $k=n-1$, $\lambda=\mu e^{-2/k}$, and
\[
 H_n(\nu)=e^{2|\nu|/k}Q_n(\nu)^2\mathbf1_{\{\ell(\nu)\le n\}},
 \qquad h(m)=\mathbb E_mH_n.
\]
Then $\mathcal R_n=e^\lambda\mathbb E_\lambda h(M)$. For a box added at row $i$, with content $C=\nu_i+1-i$, the ratio of the weights (provided both partitions are admissible) is
\[
 \frac{H_n(\nu^+)}{H_n(\nu)}
 =F(C),\quad F(x)=e^{2/k}\left(\frac{n+x}{n+x+1}\right)^2.
\]
Choose $c<\alpha<1$ and restrict first to $m\le\alpha n$, for sufficiently large  $n$ . The transition moments from Lemma~\ref{lem:contents} and Taylor's theorem give
\[
 h(m+1)-h(m)=\left[F(0)-1+\frac{m}{2}F''(0)\right]h(m)+\epsilon_m,
\]
where $|\epsilon_m|\le C_\alpha e^{2m/k}m^{3/2}/n^4$. Indeed
\[
 F''(x)=2e^{2/k}\frac{1-2(n+x)}{(n+x+1)^4},\quad
 F'''(x)=12e^{2/k}\frac{n+x-1}{(n+x+1)^5},
\]
and the fourth-moment bound controls the third absolute moment of $C$. The estimate is valid after weighting by $H_n(\nu)$, since $0\le H_n(\nu)\le e^{2m/k}$.

For $m>\alpha n$, use $h(m)\le e^{2m/k}$ and the exponential tilting identity for Poisson distributions. As $\lambda\le cn$ and $\alpha>c$, the corresponding tails, including their first moments, are $O_c(\mu^j e^{-\delta n})$ with the requisite positive powers $j$; thus they are negligible compared with the displayed remainder uniformly down to $\mu\downarrow0$. Also, for $m\le\alpha n$ ,   $Q_n^2\ge\exp[-2\alpha/(1-\alpha)+o(1)]$, giving $\mathbb E_\lambda h(M)\ge d_c>0$.

Poisson integration by parts yields
\[
 \vartheta\log \mathcal R_n
 =\lambda+\lambda\frac{\mathbb E_\lambda[h(M+1)-h(M)]}{\mathbb E_\lambda h(M)}.
\]
The weighted-mean estimate and its proof are given separately in
Lemma~\ref{lem:weighted-mean}.
Substitution into the refined estimate, noting $F''(0)=O(n^{-3})$, yields
\[
 \vartheta\log \mathcal R_n
 =\lambda F(0)+\frac{\lambda^2}{2}F''(0)
 +O_c\!\left(\frac{\mu(1+\mu)^{3/2}}{n^4}\right).
\]
The first two terms simplify exactly to
\[
 \mu\left(\frac{n}{n+1}\right)^2
 -\mu^2\frac{2n-1}{(n+1)^4}.
\]
Integrating $\vartheta=\mu\partial_\mu$ from $0$ to $\mu$, using $\mathcal R_n(0)=1$, gives \eqref{eq:uniform-phi}; indeed, the integrated error is bounded by $C_c\int_0^\mu(1+t)^{3/2}n^{-4}\,dt\le C_c\mu(1+\mu)^{3/2}n^{-4}$. The remaining assertions follow by expansion of the rational functions and exponentiation. Section~\ref{app:tails} makes the tail estimates and the derivative control explicit.
\end{proof}

\begin{corollary}[Joint recurrence-coefficient expansions]\label{cor:joint}
Fix $0<c<1$ and put
\[
 L_n(\mu)=\mu\left[\left(\frac{n}{n+1}\right)^2-1\right]
 -\frac{\mu^2}{2}\frac{2n-1}{(n+1)^4},\qquad
 \delta_n(\mu)=\frac{\mu(1+\mu)^{3/2}}{n^4}.
\]
Uniformly for $0<\mu\le cn$, as $n\to\infty$ ,
  \begin{align*}
 \mu-\gamma_n
 &=\mu A_n e^{L_n(\mu)}
 \left[1+O_c\left(\delta_n(\mu)+\frac{\mu}{n^2}\right)\right],\\
 \beta_n-n
 &=(n+1)A_{n+1}e^{L_{n+1}(\mu)}
 \left[1+O_c\left(\delta_{n+1}(\mu)+\frac{\mu}{n^2}
 +\frac{\mu}{n^2(n+1)}\right)\right].
\end{align*}
More precisely, defining $u_n=\mathcal S_{n+1}/\mathcal T_n$,
\[
 \beta_n-n=\frac{u_n}{1+u_n}
 \left[n+1+\mu\left(\frac{n+1}{n+2}\right)^2-\mu
 -\mu^2\frac{2n+1}{(n+2)^4}
 +O_c\bigl(\delta_{n+1}(\mu)\bigr)
 +O_c\bigl((n+1) A_{n+1}\bigr)\right].
\]
\end{corollary}
\begin{proof}
Since $\mathcal T_n=e^\mu(1-p_n)$, where $p_n$ is the poissonized Plancherel probability of more than $n$ rows, the elementary Robinson--Schensted bound gives
$0\le p_n\le\mu^{n+1}/((n+1)!)^2=A_{n+1}$. The insertion argument proved in Lemma~\ref{lem:trunc} gives
$0\le\vartheta p_n\le(n+1)A_{n+1}$; hence
$\vartheta\log\mathcal T_n=\mu-\vartheta p_n/(1-p_n)
=\mu+O_c((n+1)A_{n+1})$. Thus
$\mathcal S_n/\mathcal T_n=A_n\Phi_n(1+O_c(A_{n+1}))$.
Moreover, $\mathcal S_{n+1}/\mathcal S_n=\mu(n+1)^{-2}\Phi_{n+1}/\Phi_n=O_c(\mu/n^2)$ by Theorem~\ref{thm:uniform-phi}. Insert these facts into the exact identities
\[
 \mu-\gamma_n=\mu a_n(1-\rho_n+a_n\rho_n),\qquad
 \beta_n-n=\frac{u_n}{1+u_n}\left[n+1+\vartheta\log \mathcal R_{n+1}-\vartheta\log \mathcal T_n\right],
\]
where $a_n=\mathcal S_n/\mathcal T_n$ and $\rho_n=\mathcal S_{n+1}/\mathcal S_n$. The theorem supplies $\vartheta\log \mathcal R_{n+1}$, while $\vartheta\log \mathcal T_n=\mu+O_c((n+1) A_{n+1})$. The claimed formulas follow; the displayed $O_c$ terms are deliberately conservative.
\end{proof}

\paragraph{Scope and priority of verification.}
The statements above concern the subcritical linear regime $0<\mu\le cn$ with fixed $c<1$. The constants are not asserted to remain bounded as $c\uparrow1$. Claims of historical novelty require a separate literature comparison. The transition-content moment identities are classical; their application here is a derivation, not a novelty claim.

\section{Quantitative estimates for uniform sector expansions}\label{app:tails}
This section spells out the estimates needed to interpret the uniform $O_c$ terms in Theorem~\ref{thm:uniform-phi}. Fix $0<c<\alpha<1$, set $k=n-1$, $\lambda=\mu e^{-2/k}$, and take $n$ sufficiently large (depending only on $c,\alpha$). Throughout $0<\mu\le cn$.

\begin{lemma}[Poisson tails retaining the small-parameter factor]\label{lem:tails}
Let $M_t\sim\mathrm{Pois}(t)$, $r=\lfloor\alpha n\rfloor$, and $0<t\le c_1 n$ with $c_1<\alpha$. For any fixed integers $j\ge1$ and $s\ge0$, there exist $K,\delta>0$, depending on $c_1,\alpha,j,s$ only, such that
\[
 \mathbb E_t[(1+M_t)^s\mathbf1_{\{M_t>r\}}]\le Kt^j e^{-\delta n}.
\]
For each fixed $b>0$, the same assertion holds with an additional factor $e^{bM_t/n}$, provided $n$ is sufficiently large.
\end{lemma}
\begin{proof}
Put $q=r+1$ and choose an integer $J>j+s$. For $m\ge q$ and $n$ large, $(1+m)^s\le C_{j,s}q^{-j}(m)_J$ (the extra factor $m^{J-j-s}$ is at least one). Hence the Poisson factorial-moment identity yields
\[
 \mathbb E_t[(1+M_t)^s\mathbf1_{\{M_t\ge q\}}]
 \le C_{j,s}q^{-j}t^J\,
 \mathbb P\{\mathrm{Pois}(t)\ge q-J\}.
\]
Choose $c_1<c_2<\alpha$. For all sufficiently large $n$, $q-J>c_2n$ and $t\le c_1n$. The exponential Chernoff bound, uniformly for $t\le c_1n$, gives
\[
 \mathbb P\{\mathrm{Pois}(t)\ge q-J\}
 \le \exp[-\delta n],\qquad
 0<\delta<\alpha\log(\alpha/c_1)-\alpha+c_1.
\] 
This   follows by applying $\mathbb P(M_t\ge x)\le\exp[-x\log(x/t)+x-t]$ for $x>t$. Since $t^Jq^{-j}\le C t^j n^{J-j}$, the polynomial factor is absorbed into a slightly smaller exponential rate. For the additional exponential weight use
\[
 \mathbb E_t[e^{bM_t/n}g(M_t)]
 =e^{t(e^{b/n}-1)}\mathbb E_{te^{b/n}}g(M_{te^{b/n}}).
\]
The prefactor is bounded, and $te^{b/n}\le c_2n$ for all sufficiently large $n$; apply the preceding argument with $c_2<\alpha$.
\end{proof}

\begin{lemma}[Transition remainder and weighted Poisson moments]\label{lem:weighted}
Let $h(m)=\mathbb E_m H_n$, with $H_n$ as in  the proof of   Theorem~\ref{thm:uniform-phi}. Then, for $m\le r$ and all sufficiently large $n$,
\[
 h(m+1)-h(m)=\left[F(0)-1+\frac{m}{2}F''(0)\right]h(m)+\epsilon_m,
\quad |\epsilon_m|\le \frac{K_\alpha e^{2m/k}m^{3/2}}{n^4}.
\] 
 Furthermore, uniformly for $0<\mu\le cn$,
\[
 \mathbb E_\lambda h(M)\ge d_{c,\alpha}>0,
\qquad
 \frac{\mathbb E_\lambda[Mh(M)]}{\mathbb E_\lambda h(M)}=O_{c,\alpha}(\mu),
\]
\[
 \frac{\mathbb E_\lambda[M^{3/2}e^{2M/k}]}{\mathbb E_\lambda h(M)}
 =O_{c,\alpha}(\mu(1+\mu)^{1/2}).
\]
\end{lemma}
\begin{proof}
The first assertion follows from Taylor's formula, the conditional identities $\mathbb E(C\mid\nu)=0$, $\mathbb E(C^2\mid\nu)=m$, and the unconditional fourth moment $\mathbb E_m C^4=2m^2-m$. The weighting does not require a conditional fourth-moment estimate: use $H_n\le e^{2m/k}$ before taking the unconditional expectation.
For $m\le r$, the estimate $-\log Q_n^2\le 2m/(n-m)$  and   $H_n=e^{2m/k}Q_n^2$ imply $h(m)\ge\exp[-2\alpha/(1-\alpha)+o(1)]$. Lemma~\ref{lem:tails} shows $\mathbb P_\lambda(M\le r)\to1$, hence the denominator is uniformly positive. Finally, $h(m)\le e^{2m/k}$, the Poisson tilt identity, and the elementary moment inequality $\mathbb E_t M^{3/2}\le C t(1+t)^{1/2}$ yield the two moment bounds.
\end{proof}

\begin{lemma}[Weighted Poisson mean and logarithmic derivative]\label{lem:weighted-mean}
With the notation of Theorem~\ref{thm:uniform-phi}, uniformly for
$0<\mu\le cn$ and $\lambda=\mu e^{-2/(n-1)}$,
\[
\frac{\mathbb E_\lambda[Mh(M)]}{\mathbb E_\lambda h(M)}
=\lambda+O_c\!\left(\frac{\lambda(1+\lambda)}{n^2}\right).
\]
Moreover, with $\vartheta=\mu\partial_\mu$,
\[
\vartheta\log\mathcal R_n
=\lambda F(0)+\frac{\lambda^2}{2}F''(0)
+O_c\!\left(\frac{\mu(1+\mu)^{3/2}}{n^4}\right).
\]
\end{lemma}
\begin{proof}
Write $Z=\mathbb E_\lambda h(M)\ge d_{c,\alpha}>0$.
For $m\le r=\lfloor\alpha n\rfloor$, Lemma~\ref{lem:weighted}
gives
\[
|\Delta h(m)|\le\frac{K_\alpha(1+m)}{n^2}h(m)
+\frac{K_\alpha e^{2m/(n-1)}m^{3/2}}{n^4}.
\]
The weighted Poisson moments from Lemma~\ref{lem:weighted}
imply
\[
\frac{|\mathbb E_\lambda\Delta h(M)|}{Z}
\le K_c\left(\frac{1+\lambda}{n^2}
+\frac{\lambda\sqrt{1+\lambda}}{n^4}\right)
+O_c(\lambda^j e^{-\delta n})
=O_c\!\left(\frac{1+\lambda}{n^2}\right).
\]
For the omitted tail, use $|\Delta h(m)|\le h(m+1)+h(m)$,
$h(m)\le e^{2m/(n-1)}$, and the weighted-tail bound
in Lemma~\ref{lem:tails}. The Poisson size-bias identity
$\mathbb E_\lambda[Mh(M)]-\lambda Z
=\lambda\mathbb E_\lambda\Delta h(M)$ proves the first claim.
Averaging the refined transition identity yields
\[
\frac{\mathbb E_\lambda\Delta h(M)}{Z}
=F(0)-1+\frac{F''(0)}{2}
\frac{\mathbb E_\lambda[Mh(M)]}{Z}
+O_c\!\left(\frac{\lambda\sqrt{1+\lambda}}{n^4}\right).
\]
Since $F''(0)=O(n^{-3})$, replacing the weighted mean by
$\lambda$ costs $O_c(\lambda(1+\lambda)/n^5)$.
Multiplication by $\lambda$ and the exact identity
$\vartheta\log\mathcal R_n
=\lambda+\lambda\mathbb E_\lambda\Delta h(M)/Z$
prove the second claim. Here the $\lambda^2\sqrt{1+\lambda}/n^4$
remainder is bounded by the stated remainder for $\lambda\le cn$.
\end{proof}

\begin{lemma}[Derivative of the truncated Plancherel partition function]\label{lem:trunc}
Write $\mathcal T_n=e^\mu(1-p_n)$, where $p_n=\mathbb P_\mu\{\ell(\nu)>n\}$. Then
\[
 0\le p_n\le\frac{\mu^{n+1}}{((n+1)!)^2},
\qquad
 |\vartheta\log \mathcal T_n-\mu|=O_c(\mu^j e^{-\delta n})
\] 
 for any fixed integer $j\ge1$, with $\delta=\delta(c)>0$ and $n$ sufficiently large.
\end{lemma}
\begin{proof}
Under Robinson--Schensted, the number of rows is the longest decreasing subsequence of a uniformly random permutation of size $m$. A union bound over the $\binom m{n+1}$ choices of indices gives $\mathbb P_m(\ell>n)\le\binom m{n+1}/(n+1)!$. Poisson averaging gives the bound on $p_n$. For a rigorous derivative bound, let $f_n(m)$ be the probability that a uniform permutation of size $m$ has a decreasing subsequence of length $n+1$. Couple permutations of sizes $m$ and $m+1$ by appending one independent continuous random variable to $m$ independent continuous random variables and taking their relative ranks. A decreasing subsequence newly present after insertion must contain the last variable. A union bound over the $\binom{m}{n}$ choices of the other $n$ indices therefore yields
$0\le f_n(m+1)-f_n(m)\le\binom{m}{n}/(n+1)!$.
Differentiating the Poisson series (which is locally uniformly convergent) gives
$\vartheta p_n=\mu\,\mathbb E_\mu[f_n(M+1)-f_n(M)]$.
The Poisson factorial-moment identity $\mathbb E_\mu\binom{M}{n}=\mu^n/n!$ consequently implies
$0\le\vartheta p_n\le\mu^{n+1}/(n!(n+1)!)=(n+1)A_{n+1}$.
For every fixed $j\ge1$, Stirling's lower bound for the factorials and $0<\mu\le cn$ give
$(n+1)A_{n+1}\le C_{c,j}\mu^j e^{-\delta n}$ for some $\delta=\delta(c)>0$ and all sufficiently large $n$.
Finally $p_n\le A_{n+1}\to0$ uniformly in this regime, so
$\vartheta\log\mathcal T_n-\mu=-\vartheta p_n/(1-p_n)$ has the claimed bound.
\end{proof}

\begin{remark}[Precision of the joint expansion]
For fixed $0<\tau<1$ and $\mu=\tau n$, the sector theorem implies
\[
 \Phi_n(\tau n)=e^{-2\tau}\left(1+\frac{3\tau-\tau^2}{n}+O_\tau(n^{-3/2})\right).
\]
Using the exact identities for the recurrence coefficients, and the fact that the truncated-Plancherel tail is superexponentially small, one obtains
\[
 \mu-\gamma_n=\frac{\mu^{n+1}}{(n!)^2}e^{-2\tau}
 \left(1+\frac{2\tau-\tau^2}{n}+O_\tau(n^{-3/2})\right),
\]
\[
 \beta_n-n=\frac{\mu^{n+1}}{n!(n+1)!}e^{-2\tau}
 \left(1+\frac{3\tau-\tau^2}{n}+O_\tau(n^{-3/2})\right).
\]
These are statements in the scaling $\mu=\tau n$; they should not be confused with fixed-$\mu$ asymptotics.
\end{remark}

 \section{Conclusion}

	This work has established complete factorial asymptotics for the monic
	recurrence coefficients $\beta_n$ and $\gamma_n$ associated with the
	degenerate weight $\rho_\mu(k)=\mu^k/(k!)^2$. Through a natural
	parametrization in terms of a positive sequence $(c_n)$ satisfying a
	discrete Painlev\'e II equation, we have connected the factorial scale
	to the recessive solution of a discrete Bessel equation,
	$J_n(2\sqrt{\mu})$.

	The main results proved in this manuscript are the following:
	\begin{itemize}
		\item an elegant discrete Wronskian identity relating $c_n$ and
		$J_n(2\sqrt{\mu})$;
		\item the proof of convergence $c_n\to0$ and eventual strict
		monotonicity of $(c_n)$;
		\item a factorial upper bound $c_n\le A_\mu\,\mu^{n/2}/n!$, obtained by
		controlling the ratio $r_n=c_n/c_{n-1}$;
		\item the exact connection formula
		\[
		\frac{c_n}{J_n(2\sqrt{\mu})} \longrightarrow 1,
		\]
		which yields the precise asymptotics
		\[
		\beta_n-n\sim \frac{\mu^{n+1}}{n!(n+1)!}, \qquad
		\mu-\gamma_n\sim \frac{\mu^{n+1}}{(n!)^2}.
		\]
	\end{itemize}

In parallel, we have developed a complementary approach based on Hankel
determinants and an interpretation in terms of a discrete logarithmic
gas. We have shown that the Hankel determinant associated with
$\rho_\mu$ factorizes as $\Delta_n(\mu)=\mu^{n(n-1)/2}T_n(\mu)$, and
that the recurrence coefficients can be elegantly expressed via
$T_n(\mu)$. The sector decomposition
\[
T_n(\mu)=1+\sum_{q=1}^{n}S_q(\mu),
\]
where the terms $S_q(\mu)$ are independent of $n$, transforms the
asymptotics in degree $n$ of the recurrence coefficients into an
asymptotics in large sector $q$. We have established the asymptotics
$S_q(\mu)\sim e^\mu\mu^q/(q!)^2$, which provides an independent
derivation of the factorial corrections to $\beta_n$ and $\gamma_n$.

\begingroup\color{black}
In summary, the results identify the leading factorial corrections with
explicit constants and provide expansions to every prescribed finite
algebraic order after factorial normalization, for each fixed $\mu>0$.
The separate sector theorem gives a quantitative expansion of
$\log\Phi_n(\mu)$ uniformly for $0<\mu\le cn$, with each fixed
$c<1$, and yields joint refinements of $\beta_n$ and $\gamma_n$ in
that subcritical regime. No uniform assertion is made here as
$c\uparrow1$ or beyond the stated parameter ranges.
The Bessel comparison and the Hankel--partition decomposition connect
these asymptotic results with discrete Painlev\'e dynamics, Toda
deformations, and classical partition identities. The latter structures
are used as established tools; the contribution is the quantitative
factorial-scale analysis and its stated uniform refinements.
\endgroup
	\section*{Appendices}

	\appendix

	\section{Finite-level sector calculations}
	\label{app:finite-level-sectors}

	\subsection{Occupation-number representation}

	For a fixed truncation level $M$, the weakly ordered variables
	\[
	0\leq s_1\leq\cdots\leq s_q\leq M
	\]
	appearing in \eqref{eq:truncated-normalized-sector} can be replaced by
	their occupation numbers.

	For $0\leq m\leq M$, let
	\begin{equation*}
		N_m
		=
		\#\{a\in\{1,\ldots,q\}:s_a=m\}.
	\end{equation*}
	The occupation numbers satisfy
	\begin{equation}
		N_0+\cdots+N_M=q.
		\label{eq:occupation-number-sum}
	\end{equation}
	Conversely, every vector of nonnegative integers
	\[
	(N_0,\ldots,N_M)
	\]
	satisfying \eqref{eq:occupation-number-sum} determines a unique weakly
	ordered configuration
	\[
	0\leq s_1\leq\cdots\leq s_q\leq M.
	\]

	Define the cumulative occupation numbers by
	\begin{equation*}
		Q_m
		=
		N_0+\cdots+N_m,
		\qquad
		0\leq m\leq M,
	\end{equation*}
	and set
	\begin{equation*}
		Q_{-1}=0.
	\end{equation*}
	Then the indices $a$ for which $s_a=m$ are precisely
	\begin{equation*}
		Q_{m-1}+1
		\leq a\leq Q_m.
	\end{equation*}

	For a given occupation vector, define the factorial contribution
	\begin{equation}
		\mathcal F_q(N_0,\ldots,N_M)
		=
		(q!)^2
		\prod_{m=0}^{M}
		\prod_{a=Q_{m-1}+1}^{Q_m}
		\left(
		\frac{(a-1)!}{(a+m)!}
		\right)^2.
		\label{eq:occupation-factorial-weight}
	\end{equation}
	Define also the interaction contribution
	\begin{equation}
		\mathcal V_q(N_0,\ldots,N_M)
		=
		\prod_{0\leq m<\ell\leq M}
		\;
		\prod_{a=Q_{m-1}+1}^{Q_m}
		\;
		\prod_{b=Q_{\ell-1}+1}^{Q_\ell}
		\left(
		\frac{b-a+\ell-m}{b-a}
		\right)^2.
		\label{eq:occupation-interaction-weight}
	\end{equation}
	As usual, an empty product is interpreted as equal to $1$.

	The interaction between two indices belonging to the same occupation
	block is equal to $1$. Indeed, if $s_a=s_b=m$, then
	\[
	\frac{b-a+s_b-s_a}{b-a}
	=
	1.
	\]
	For this reason, only interactions between distinct occupation levels
	appear in \eqref{eq:occupation-interaction-weight}.

	\begin{proposition}[Occupation-number formula]
		\label{prop:occupation-number-formula}

		For every $q\geq1$ and every $M\geq0$,
		\begin{align}
			\mathcal R_q^{[M]}(\mu)
			&=
			\sum_{\substack{N_0,\ldots,N_M\geq0\\
					N_0+\cdots+N_M=q}}
			\mu^{\sum_{m=0}^{M}mN_m}
			\mathcal F_q(N_0,\ldots,N_M)
			\mathcal V_q(N_0,\ldots,N_M).
			\label{eq:occupation-number-sector-sum}
		\end{align}

	\end{proposition}

	\begin{proof}
		Every weakly ordered vector
		\[
		0\leq s_1\leq\cdots\leq s_q\leq M
		\]
		determines a unique occupation vector
		\[
		(N_0,\ldots,N_M)
		\]
		satisfying \eqref{eq:occupation-number-sum}. Conversely, every such
		occupation vector uniquely determines the ordered configuration by
		placing first $N_0$ zeros, then $N_1$ ones, and so on.

		For the corresponding configuration,
		\[
		s_1+\cdots+s_q
		=
		\sum_{m=0}^{M}mN_m.
		\]
		Hence, the power of $\mu$ in
		\eqref{eq:truncated-normalized-sector} becomes
		\[
		\mu^{\sum_{m=0}^{M}mN_m}.
		\]
		The factorial factor in
		\eqref{eq:truncated-normalized-sector} is exactly
		\[
		\mathcal F_q(N_0,\ldots,N_M).
		\]
		The Vandermonde quotient is equal to $1$ for pairs of indices belonging
		to the same occupation block. For indices in two distinct blocks
		corresponding to levels $m<\ell$, it is
		\[
		\left(
		\frac{b-a+\ell-m}{b-a}
		\right)^2.
		\]
		Taking the product over all pairs of distinct blocks gives
		\[
		\mathcal V_q(N_0,\ldots,N_M).
		\]
		Substitution into \eqref{eq:truncated-normalized-sector} proves
		\eqref{eq:occupation-number-sector-sum}.
	\end{proof}

	The formula \eqref{eq:occupation-number-sector-sum} replaces a
	multiple sum over ordered variables by a finite sum over the occupation
	simplex
	\[
	\left\{
	(N_0,\ldots,N_M)\in\mathbb N_0^{M+1}:
	N_0+\cdots+N_M=q
	\right\}.
	\]
	For fixed $M$, this set has cardinality
	\begin{equation*}
		\binom{q+M}{M}.
	\end{equation*}
	Thus, the number of summation variables is independent of $q$ when the
	truncation level $M$ is fixed.

	The case $M=0$ contains only the configuration $N_0=q$. Formula
	\eqref{eq:occupation-number-sector-sum} then gives
	\begin{equation*}
		\mathcal R_q^{[0]}(\mu)=1.
	\end{equation*}
	This is precisely the contribution of the minimal sector configuration
	\[
	r_1=\cdots=r_q=1.
	\]

\subsection{The two-level truncation}

\begingroup\color{black}
For $M=1$, weak ordering leaves exactly one configuration for each
$j\in\{0,\ldots,q\}$: $q-j$ zero entries followed by $j$ unit entries.
This observation yields the following explicit formula.
\endgroup

\begin{proposition}[Explicit two-level formula]
	\label{prop:two-level-formula}

	For every $q\geq1$,
	\begin{equation}
		\mathcal R_q^{[1]}(\mu)
		=
		\sum_{j=0}^{q}
		\frac{\mu^j}{(j!)^2}
		\left(
		\frac{q-j+1}{q+1}
		\right)^2.
		\label{eq:two-level-truncation-formula}
	\end{equation}
	Consequently,
	\begin{equation}
		\mathcal R_q^{[1]}(\mu)
		\longrightarrow
		I_0(2\sqrt{\mu}),
		\qquad
		q\to\infty.
		\label{eq:two-level-truncation-limit}
	\end{equation}

\end{proposition}

\begin{proof}
	For the occupation vector
	\[
	N_0=q-j,
	\qquad
	N_1=j,
	\]
	the factorial contribution defined in
	\eqref{eq:occupation-factorial-weight} is
	\begin{align*}
		\mathcal F_q(q-j,j)
		&=
		(q!)^2
		\prod_{a=1}^{q-j}
		\left(
		\frac{(a-1)!}{a!}
		\right)^2
		\prod_{a=q-j+1}^{q}
		\left(
		\frac{(a-1)!}{(a+1)!}
		\right)^2.
	\end{align*}
	The first product is
	\[
	\prod_{a=1}^{q-j}
	\left(
	\frac{(a-1)!}{a!}
	\right)^2
	=
	\frac{1}{((q-j)!)^2}.
	\]
	For the second product,
	\[
	\frac{(a-1)!}{(a+1)!}
	=
	\frac{1}{a(a+1)},
	\]
	and hence
	\begin{align*}
		\prod_{a=q-j+1}^{q}
		\left(
		\frac{(a-1)!}{(a+1)!}
		\right)^2
		&=
		\prod_{a=q-j+1}^{q}
		\frac{1}{a^2(a+1)^2}
		\\
		&=
		\left(
		\frac{(q-j)!(q-j+1)!}{q!(q+1)!}
		\right)^2.
	\end{align*}
	Combining the two factors gives
	\begin{equation}
		\mathcal F_q(q-j,j)
		=
		\left(
		\frac{(q-j+1)!}{(q+1)!}
		\right)^2.
		\label{eq:two-level-factorial-simplified}
	\end{equation}

	It remains to compute the interaction factor. Only pairs consisting of
	one variable at level $0$ and one variable at level $1$ contribute.
	Thus,
	\begin{align*}
		\mathcal V_q(q-j,j)
		&=
		\prod_{a=1}^{q-j}
		\prod_{b=q-j+1}^{q}
		\left(
		\frac{b-a+1}{b-a}
		\right)^2.
	\end{align*}
	For fixed $b$, the product over $a$ telescopes:
	\begin{align*}
		\prod_{a=1}^{q-j}
		\frac{b-a+1}{b-a}
		&=
		\frac{b}{b-q+j}.
	\end{align*}
	Therefore,
	\begin{align}
		\mathcal V_q(q-j,j)
		&=
		\prod_{b=q-j+1}^{q}
		\left(
		\frac{b}{b-q+j}
		\right)^2
		\nonumber\\
		&=
		\left(
		\frac{q!}{j!(q-j)!}
		\right)^2
		\nonumber\\
		&=
		\binom{q}{j}^2.
		\label{eq:two-level-interaction-simplified}
	\end{align}

	Substituting \eqref{eq:two-level-factorial-simplified} and
	\eqref{eq:two-level-interaction-simplified} into
	\eqref{eq:occupation-number-sector-sum} gives
	\begin{align*}
		\mathcal R_q^{[1]}(\mu)
		&=
		\sum_{j=0}^{q}
		\mu^j
		\left(
		\frac{(q-j+1)!}{(q+1)!}
		\right)^2
		\left(
		\frac{q!}{j!(q-j)!}
		\right)^2
		\\
		&=
		\sum_{j=0}^{q}
		\frac{\mu^j}{(j!)^2}
		\left(
		\frac{q-j+1}{q+1}
		\right)^2.
	\end{align*}
	This proves \eqref{eq:two-level-truncation-formula}.

	For every fixed $j$,
	\[
	\left(
	\frac{q-j+1}{q+1}
	\right)^2
	\longrightarrow1,
	\qquad
	q\to\infty.
	\]
	Moreover,
	\[
	0
	\leq
	\frac{\mu^j}{(j!)^2}
	\left(
	\frac{q-j+1}{q+1}
	\right)^2
	\leq
	\frac{\mu^j}{(j!)^2}.
	\]
	Since
	\[
	\sum_{j=0}^{\infty}
	\frac{\mu^j}{(j!)^2}
	=
	I_0(2\sqrt{\mu})
	\]
	is convergent, dominated convergence yields
	\[
	\mathcal R_q^{[1]}(\mu)
	\longrightarrow
	\sum_{j=0}^{\infty}
	\frac{\mu^j}{(j!)^2}
	=
	I_0(2\sqrt{\mu}).
	\]
\end{proof}

The result \eqref{eq:two-level-truncation-limit} gives the first
nontrivial limiting contribution to the normalized sector weight:
\[
\mathcal R_q^{[0]}(\mu)=1,
\qquad
\mathcal R_q^{[1]}(\mu)
\longrightarrow
I_0(2\sqrt{\mu}).
\]
The appearance of the Bessel function is consistent with the moment
identity
\[
\sum_{j=0}^{\infty}
\frac{\mu^j}{(j!)^2}
=
I_0(2\sqrt{\mu}).
\]
It also shows that the first correction above the minimal sector
configuration remains of order one as $q\to\infty$.

	\subsection{The three-level truncation}

	The next truncation level is $M=2$. The excess variables now take
	values in the set
	\[
	\{0,1,2\}.
	\]
	A weakly ordered configuration is uniquely determined by the two
	occupation numbers
	\[
	j=N_1,
	\qquad
	k=N_2,
	\]
	since
	\[
	N_0=q-j-k.
	\]
	The admissible domain is
	\begin{equation*}
		j\geq0,
		\qquad
		k\geq0,
		\qquad
		j+k\leq q.
	\end{equation*}

	For convenience, define
	\begin{equation*}
		n_0=q-j-k,
		\qquad
		n_1=q-k.
	\end{equation*}
	The corresponding ordered configuration is
	\[
	\underbrace{0,\ldots,0}_{n_0\text{ times}},
	\underbrace{1,\ldots,1}_{j\text{ times}},
	\underbrace{2,\ldots,2}_{k\text{ times}}.
	\]

	\begin{proposition}[Explicit three-level formula]
		\label{prop:three-level-formula}

		For every $q\geq1$,
		\begin{equation}
			\mathcal R_q^{[2]}(\mu)
			=
			\sum_{\substack{j,k\geq0\\j+k\leq q}}
			\mu^{j+2k}
			\mathcal A_q(j,k)
			\mathcal B_q(j,k)
			\mathcal C_q(j,k)
			\mathcal D_q(j,k),
			\label{eq:three-level-truncation-formula}
		\end{equation}
		where
		\begin{equation}
			\mathcal A_q(j,k)
			=
			\left(
			\frac{(q-j-k+1)!(q-k+2)!}{(q+1)!(q+2)!}
			\right)^2,
			\label{eq:three-level-factorial-factor}
		\end{equation}
		\begin{equation}
			\mathcal B_q(j,k)
			=
			\left(
			\frac{(q-k)!}{j!(q-j-k)!}
			\right)^2,
			\label{eq:three-level-zero-one-interaction}
		\end{equation}
		\begin{equation}
			\mathcal C_q(j,k)
			=
			\left(
			\frac{(j+k)!}{j!k!}
			\right)^2,
			\label{eq:three-level-one-two-interaction}
		\end{equation}
		and
		\begin{equation}
			\mathcal D_q(j,k)
			=
			\left(
			\frac{q!(q+1)!j!(j+1)!}{(q-k)!(q-k+1)!(j+k)!(j+k+1)!}
			\right)^2.
			\label{eq:three-level-zero-two-interaction}
		\end{equation}

	\end{proposition}

	\begin{proof}
		The factorial contribution is obtained from
		\eqref{eq:occupation-factorial-weight}. The level-$0$ variables
		contribute
		\[
		\prod_{a=1}^{n_0}
		\left(
		\frac{(a-1)!}{a!}
		\right)^2
		=
		\frac{1}{(n_0!)^2}.
		\]
		The level-$1$ variables contribute
		\begin{align*}
			\prod_{a=n_0+1}^{n_1}
			\left(
			\frac{(a-1)!}{(a+1)!}
			\right)^2
			&=
			\prod_{a=n_0+1}^{n_1}
			\left(
			\frac{1}{a(a+1)}
			\right)^2
			\\
			&=
			\left(
			\frac{n_0!(n_0+1)!}{n_1!(n_1+1)!}
			\right)^2.
		\end{align*}
		The level-$2$ variables contribute
		\begin{align*}
			\prod_{a=n_1+1}^{q}
			\left(
			\frac{(a-1)!}{(a+2)!}
			\right)^2
			&=
			\prod_{a=n_1+1}^{q}
			\left(
			\frac{1}{a(a+1)(a+2)}
			\right)^2
			\\
			&=
			\left(
			\frac{n_1!(n_1+1)!(n_1+2)!}{q!(q+1)!(q+2)!}
			\right)^2.
		\end{align*}
		After multiplication by $(q!)^2$, the product of these three factors
		simplifies to
		\[
		\left(
		\frac{n_0!(n_1+2)!}{(q+1)!(q+2)!}
		\right)^2.
		\]
		Using
		\[
		n_0=q-j-k,
		\qquad
		n_1=q-k,
		\]
		this gives \eqref{eq:three-level-factorial-factor}.

		There are three nontrivial interaction terms. First, the interaction
		between levels $0$ and $1$ is
		\begin{align*}
			\prod_{a=1}^{n_0}
			\prod_{b=n_0+1}^{n_1}
			\left(
			\frac{b-a+1}{b-a}
			\right)^2
			&=
			\left(
			\frac{n_1!}{j!n_0!}
			\right)^2,
		\end{align*}
		which is \eqref{eq:three-level-zero-one-interaction}.

		Second, the interaction between levels $1$ and $2$ is
		\begin{align*}
			\prod_{a=n_0+1}^{n_1}
			\prod_{b=n_1+1}^{q}
			\left(
			\frac{b-a+1}{b-a}
			\right)^2
			&=
			\left(
			\frac{(j+k)!}{j!k!}
			\right)^2,
		\end{align*}
		which is \eqref{eq:three-level-one-two-interaction}.

		Finally, the interaction between levels $0$ and $2$ has displacement
		difference equal to $2$:
		\begin{align}
			\prod_{a=1}^{n_0}
			\prod_{b=n_1+1}^{q}
			\left(
			\frac{b-a+2}{b-a}
			\right)^2
			&=
			\prod_{b=n_1+1}^{q}
			\left[
			\frac{b(b+1)}{(b-n_0)(b-n_0+1)}
			\right]^2.
			\label{eq:three-level-zero-two-start}
		\end{align}
		Since
		\[
		n_1=q-k,
		\qquad
		n_0=q-j-k,
		\]
		the expression in \eqref{eq:three-level-zero-two-start} becomes
		\[
		\left(
		\frac{q!(q+1)!j!(j+1)!}{(q-k)!(q-k+1)!(j+k)!(j+k+1)!}
		\right)^2,
		\]
		which is \eqref{eq:three-level-zero-two-interaction}.

		Multiplying the weight $\mu^{j+2k}$ by the factorial factor and the
		three interaction factors proves
		\eqref{eq:three-level-truncation-formula}.
	\end{proof}

	The formula \eqref{eq:three-level-truncation-formula} extends the
	two-level expression \eqref{eq:two-level-truncation-formula}. Indeed,
	when $k=0$, one has $N_2=0$, and the three-level expression reduces to
	the two-level truncation.

	The term $j=0$, $k=0$ is equal to $1$. It corresponds to the minimal
	excitation configuration
	\[
	s_1=\cdots=s_q=0.
	\]
	The terms with $k=0$ describe configurations involving only the levels
	$0$ and $1$, whereas the terms with $k\geq1$ contain the first genuine
	contributions from the excess level $2$.

	For fixed $j$ and $k$, the coefficient in
	\eqref{eq:three-level-truncation-formula} has a finite limit as
	$q\to\infty$. The next subsection simplifies this limiting coefficient
	and justifies the passage from the finite double sum to an infinite
	double series.

	\subsection{Fixed-occupation asymptotics in the three-level model}

	The explicit formula \eqref{eq:three-level-truncation-formula} allows
	one to compute the large-$q$ limit of each individual occupation
	contribution. For fixed nonnegative integers $j$ and $k$, define
	\begin{equation}
		\mathcal W_q(j,k)
		=
		\mathcal A_q(j,k)
		\mathcal B_q(j,k)
		\mathcal C_q(j,k)
		\mathcal D_q(j,k).
		\label{eq:three-level-coefficient}
	\end{equation}
	Then
	\begin{equation}
		\mathcal R_q^{[2]}(\mu)
		=
		\sum_{\substack{j,k\geq0\\j+k\leq q}}
		\mu^{j+2k}\mathcal W_q(j,k).
		\label{eq:three-level-coefficient-sum}
	\end{equation}

	\begin{proposition}[Fixed-occupation limit]
		\label{prop:three-level-fixed-occupation-limit}

		For every fixed pair of nonnegative integers $j,k$,
		\begin{equation}
			\mathcal W_q(j,k)
			\longrightarrow
			\frac{(j+1)^2}{(k!)^2(j+k+1)!^2},
			\qquad
			q\to\infty.
			\label{eq:three-level-fixed-occupation-limit}
		\end{equation}

	\end{proposition}

	\begin{proof}
		Recall that
		\[
		\mathcal W_q(j,k)
		=
		\mathcal A_q(j,k)
		\mathcal B_q(j,k)
		\mathcal C_q(j,k)
		\mathcal D_q(j,k).
		\]
		The factor $\mathcal C_q(j,k)$ is independent of $q$:
		\[
		\mathcal C_q(j,k)
		=
		\left(
		\frac{(j+k)!}{j!k!}
		\right)^2.
		\]

		For fixed $j$ and $k$, the factorial ratios in
		\eqref{eq:three-level-factorial-factor} satisfy
		\begin{align*}
			\mathcal A_q(j,k)
			&=
			q^{-2j-4k}
			\left(
			1+O\!\left(\frac{1}{q}\right)
			\right).
		\end{align*}
		Likewise,
		\begin{align*}
			\mathcal B_q(j,k)
			&=
			\frac{q^{2j}}{(j!)^2}
			\left(
			1+O\!\left(\frac{1}{q}\right)
			\right),
		\end{align*}
		whereas
		\begin{align*}
			\mathcal D_q(j,k)
			&=
			q^{4k}
			\left(
			\frac{j!(j+1)!}{(j+k)!(j+k+1)!}
			\right)^2
			\left(
			1+O\!\left(\frac{1}{q}\right)
			\right).
		\end{align*}
		The powers of $q$ cancel:
		\[
		q^{-2j-4k}q^{2j}q^{4k}
		=
		1.
		\]
		Therefore,
		\begin{align*}
			\mathcal W_q(j,k)
			&=
			\frac{1}{(j!)^2}
			\left(
			\frac{(j+k)!}{j!k!}
			\right)^2
			\left(
			\frac{j!(j+1)!}{(j+k)!(j+k+1)!}
			\right)^2
			\left(
			1+O\!\left(\frac{1}{q}\right)
			\right).
		\end{align*}
		The factorial factors simplify to
		\begin{align*}
			\frac{1}{(j!)^2}
			\left(
			\frac{(j+k)!}{j!k!}
			\right)^2
			\left(
			\frac{j!(j+1)!}{(j+k)!(j+k+1)!}
			\right)^2
			&=
			\frac{(j+1)^2}{(k!)^2(j+k+1)!^2}.
		\end{align*}
		This proves \eqref{eq:three-level-fixed-occupation-limit}.
	\end{proof}

	The preceding proposition suggests the limiting double series
	\begin{equation*}
		\mathcal R_\infty^{[2]}(\mu)
		=
		\sum_{j=0}^{\infty}
		\sum_{k=0}^{\infty}
		\mu^{j+2k}
		\frac{(j+1)^2}{(k!)^2(j+k+1)!^2}.
	\end{equation*}
	This double series is absolutely convergent. Indeed,
	\[
	(j+k+1)!\geq(j+1)!,
	\]
	and hence
	\begin{align}
		\mu^{j+2k}
		\frac{(j+1)^2}{(k!)^2(j+k+1)!^2}
		&\leq
		\frac{\mu^j}{(j!)^2}
		\frac{\mu^{2k}}{(k!)^2}.
		\label{eq:three-level-majorant}
	\end{align}
	The majorant is summable:
	\begin{align*}
		\sum_{j=0}^{\infty}
		\sum_{k=0}^{\infty}
		\frac{\mu^j}{(j!)^2}
		\frac{\mu^{2k}}{(k!)^2}
		&=
		I_0(2\sqrt{\mu})I_0(2\mu).
	\end{align*}

	\subsection{Uniform domination in the three-level model}

	The coefficient $\mathcal W_q(j,k)$ introduced in
	\eqref{eq:three-level-coefficient} admits an exact simplification. This
	gives the uniform estimate required to justify the passage to the limit
	in the three-level truncation.

	\begin{proposition}[Exact coefficient factorization]
		\label{prop:three-level-exact-factorization}

		For every $q\geq1$ and every pair of nonnegative integers $j,k$
		satisfying $j+k\leq q$, one has
		\begin{equation}
			\mathcal W_q(j,k)
			=
			\frac{(j+1)^2}{(k!)^2(j+k+1)!^2}
			\left(
			\frac{(q-j-k+1)(q-k+2)}{(q+1)(q+2)}
			\right)^2.
			\label{eq:three-level-exact-coefficient}
		\end{equation}

	\end{proposition}

	\begin{proof}
		Recall that
		\[
		\mathcal W_q(j,k)
		=
		\mathcal A_q(j,k)
		\mathcal B_q(j,k)
		\mathcal C_q(j,k)
		\mathcal D_q(j,k).
		\]
		Taking square roots of the positive factors gives
		\begin{align*}
			\sqrt{\mathcal W_q(j,k)}
			&=
			\frac{(q-j-k+1)!(q-k+2)!}{(q+1)!(q+2)!}
			\frac{(q-k)!}{j!(q-j-k)!}
			\frac{(j+k)!}{j!k!}
			\\
			&\quad\times
			\frac{q!(q+1)!j!(j+1)!}{(q-k)!(q-k+1)!(j+k)!(j+k+1)!}.
		\end{align*}
		Cancelling the common factorial factors yields
		\begin{align*}
			\sqrt{\mathcal W_q(j,k)}
			&=
			\frac{(q-j-k+1)(q-k+2)(j+1)}{(q+1)(q+2)k!(j+k+1)!}.
		\end{align*}
		Thus,
		\begin{align}
			\sqrt{\mathcal W_q(j,k)}
			&=
			\frac{j+1}{k!(j+k+1)!}
			\frac{(q-j-k+1)(q-k+2)}{(q+1)(q+2)}.
			\label{eq:three-level-square-root-final}
		\end{align}
		Squaring \eqref{eq:three-level-square-root-final} proves
		\eqref{eq:three-level-exact-coefficient}.
	\end{proof}

	The finite-$q$ correction factor in
	\eqref{eq:three-level-exact-coefficient} satisfies
	\begin{equation*}
		0
		\leq
		\frac{(q-j-k+1)(q-k+2)}{(q+1)(q+2)}
		\leq
		1.
	\end{equation*}
	Indeed,
	\[
	q-j-k+1\leq q+1,
	\qquad
	q-k+2\leq q+2.
	\]
	Therefore,
	\begin{equation}
		0
		\leq
		\mathcal W_q(j,k)
		\leq
		\frac{(j+1)^2}{(k!)^2(j+k+1)!^2}.
		\label{eq:three-level-uniform-majorant}
	\end{equation}

	\begin{theorem}[Three-level truncation limit]
		\label{thm:three-level-truncation-limit}

		For every $\mu>0$,
		\begin{equation}
			\mathcal R_q^{[2]}(\mu)
			\longrightarrow
			\sum_{j=0}^{\infty}
			\sum_{k=0}^{\infty}
			\mu^{j+2k}
			\frac{(j+1)^2}{(k!)^2(j+k+1)!^2},
			\qquad
			q\to\infty.
			\label{eq:three-level-limit-theorem}
		\end{equation}

	\end{theorem}

	\begin{proof}
		By \eqref{eq:three-level-coefficient-sum} and
		\eqref{eq:three-level-exact-coefficient},
		\begin{align*}
			\mathcal R_q^{[2]}(\mu)
			&=
			\sum_{\substack{j,k\geq0\\j+k\leq q}}
			\mu^{j+2k}
			\frac{(j+1)^2}{(k!)^2(j+k+1)!^2}
			\left(
			\frac{(q-j-k+1)(q-k+2)}{(q+1)(q+2)}
			\right)^2.
		\end{align*}
		For every fixed pair $j,k$,
		\[
		\left(
		\frac{(q-j-k+1)(q-k+2)}{(q+1)(q+2)}
		\right)^2
		\longrightarrow1,
		\qquad
		q\to\infty.
		\]
		Moreover, by \eqref{eq:three-level-uniform-majorant},
		\[
		0
		\leq
		\mu^{j+2k}
		\frac{(j+1)^2}{(k!)^2(j+k+1)!^2}
		\left(
		\frac{(q-j-k+1)(q-k+2)}{(q+1)(q+2)}
		\right)^2
		\leq
		\mu^{j+2k}
		\frac{(j+1)^2}{(k!)^2(j+k+1)!^2}.
		\]
		The right-hand side is summable by
		\eqref{eq:three-level-majorant}. The dominated convergence theorem
		therefore yields \eqref{eq:three-level-limit-theorem}.
	\end{proof}

	The limiting expression in \eqref{eq:three-level-limit-theorem} can
	be written as
	\begin{equation*}
		\mathcal R_\infty^{[2]}(\mu)
		=
		\sum_{j=0}^{\infty}
		\sum_{k=0}^{\infty}
		\mu^{j+2k}
		\frac{(j+1)^2}{(k!)^2(j+k+1)!^2}.
	\end{equation*}
	The restriction to $k=0$ recovers the two-level limit:
	\begin{align*}
		\sum_{j=0}^{\infty}
		\mu^j
		\frac{(j+1)^2}{(j+1)!^2}
		&=
		\sum_{j=0}^{\infty}
		\frac{\mu^j}{(j!)^2}
		\\
		&=
		I_0(2\sqrt{\mu}).
	\end{align*}
	Thus, the terms with $k\geq1$ give the genuine new contribution of the
	third excitation level.

\subsection{Consequences of the first truncation levels}

The first three truncation levels provide rigorous lower approximations
to the normalized sector weight $\mathcal R_q(\mu)$. Since every
summand in the normalized sector representation
\eqref{eq:normalized-sector-sum} is nonnegative, the truncations are
monotone in the truncation level:
\begin{equation*}
	\mathcal R_q^{[0]}(\mu)
	\leq
	\mathcal R_q^{[1]}(\mu)
	\leq
	\mathcal R_q^{[2]}(\mu)
	\leq
	\mathcal R_q(\mu).
\end{equation*}

The zeroth truncation consists only of the minimal configuration
\[
s_1=\cdots=s_q=0.
\]
Therefore,
\begin{equation*}
	\mathcal R_q^{[0]}(\mu)=1.
\end{equation*}

The two-level analysis gives
\begin{equation*}
	\mathcal R_q^{[1]}(\mu)
	\longrightarrow
	I_0(2\sqrt{\mu}),
	\qquad
	q\to\infty.
\end{equation*}

The three-level analysis gives
\begin{equation*}
	\mathcal R_q^{[2]}(\mu)
	\longrightarrow
	\mathcal R_\infty^{[2]}(\mu),
	\qquad
	q\to\infty,
\end{equation*}
where
\begin{equation}
	\mathcal R_\infty^{[2]}(\mu)
	=
	\sum_{j=0}^{\infty}
	\sum_{k=0}^{\infty}
	\mu^{j+2k}
	\frac{(j+1)^2}{(k!)^2(j+k+1)!^2}.
	\label{eq:second-truncation-limit-recalled-definition}
\end{equation}

Since the subfamily $k=0$ in
\eqref{eq:second-truncation-limit-recalled-definition} reproduces the
two-level limit, one has
\begin{align*}
	\mathcal R_\infty^{[2]}(\mu)
	&\geq
	\sum_{j=0}^{\infty}
	\mu^j
	\frac{(j+1)^2}{(j+1)!^2}
	\\
	&=
	\sum_{j=0}^{\infty}
	\frac{\mu^j}{(j!)^2}
	\\
	&=
	I_0(2\sqrt{\mu}).
\end{align*}

Consequently,
\begin{equation}
	\liminf_{q\to\infty}
	\mathcal R_q(\mu)
	\geq
	\mathcal R_\infty^{[2]}(\mu)
	\geq
	I_0(2\sqrt{\mu})
	\geq
	1.
	\label{eq:normalized-sector-liminf-lower-bounds}
\end{equation}

The inequalities in
\eqref{eq:normalized-sector-liminf-lower-bounds} show that the minimal
sector contribution alone does not capture the large-sector
normalization. Even the first nontrivial excitation level produces the
Bessel factor $I_0(2\sqrt{\mu})$. The level-$2$ configurations give an
additional positive correction.

More generally, for every fixed truncation level $M$, define
\begin{equation*}
	\mathcal R_\infty^{[M]}(\mu)
	=
	\lim_{q\to\infty}
	\mathcal R_q^{[M]}(\mu),
\end{equation*}
whenever the limit exists. The positivity of the sector weights implies
the formal monotonicity relation
\begin{equation*}
	\mathcal R_\infty^{[M]}(\mu)
	\leq
	\mathcal R_\infty^{[M+1]}(\mu).
\end{equation*}

The main remaining question is whether the increasing sequence of
truncated limits converges to the infinite reduced tau function:
\begin{equation*}
	\lim_{M\to\infty}
	\mathcal R_\infty^{[M]}(\mu)
	=
	\mathcal T_\infty(\mu).
\end{equation*}

The hook-length analysis in the main text proves precisely this
identification (see Sections~\ref{sec:hook-length-tail-control} and
\ref{sec:fixed-level-convergence}):
\[
\mathcal T_\infty(\mu)
=
e^\mu,
\qquad
\mathcal R_q(\mu)
\longrightarrow
e^\mu.
\]
\begingroup\color{black}
The connection constant is already identified in
Theorem~\ref{thm:connection-constant}; no second connection
argument is needed here.
\endgroup

The calculations for $M=1$ and $M=2$ therefore provide explicit,
finite-dimensional illustrations of the general Schur-function and
hook-length mechanism used in the proof.

	\section{Partition identities and standard Young tableaux}
	\label{app:partition-identities}

	This appendix records the combinatorial identities used in the
	partition expansions of the reduced tau functions and normalized sector
	weights.

	\subsection{Partitions and hook lengths}

	A partition is a finite nonincreasing sequence
	\[
	\nu
	=
	(\nu_1,\nu_2,\ldots),
	\qquad
	\nu_1\geq\nu_2\geq\cdots\geq0.
	\]
	Its size is
	\[
	|\nu|
	=
	\sum_{i\geq1}\nu_i,
	\]
	and its length is
	\[
	\ell(\nu)
	=
	\#\{i:\nu_i>0\}.
	\]

	The Young diagram of $\nu$ consists of the boxes
	\[
	(i,j),
	\qquad
	1\leq i\leq\ell(\nu),
	\qquad
	1\leq j\leq\nu_i.
	\]
	For a box $(i,j)$, let
	\[
	h_{i,j}
	\]
	denote its hook length.

	A standard Young tableau of shape $\nu$ is a filling of the Young
	diagram of $\nu$ by the integers
	\[
	1,2,\ldots,|\nu|
	\]
	which is strictly increasing along every row and every column. Let
	\[
	f^\nu
	\]
	denote the number of standard Young tableaux of shape $\nu$.

	\subsection{Hook-length formula}

	The hook-length formula states that
	\begin{equation}
		f^\nu
		=
		\frac{|\nu|!}{\displaystyle\prod_{(i,j)\in\nu}h_{i,j}}.
		\label{eq:appendix-hook-length-formula}
	\end{equation}
	Equivalently,
	\begin{equation*}
		\prod_{(i,j)\in\nu}
		\frac{1}{h_{i,j}}
		=
		\frac{f^\nu}{|\nu|!}.
	\end{equation*}

	This identity was used to convert the hook-product expressions in the
	tau-function expansions into sums involving the standard-tableau
	numbers $f^\nu$.

	\subsection{Hook-content formula}

	Let $s_\nu$ denote the Schur function associated with the partition
	$\nu$. For an integer $q\geq\ell(\nu)$, the specialization at
	\[
	1^q=(1,\ldots,1)
	\]
	is given by the hook-content formula
	\begin{equation*}
		s_\nu(1^q)
		=
		\prod_{(i,j)\in\nu}
		\frac{q+j-i}{h_{i,j}}.
	\end{equation*}

	Equivalently, one has the Weyl dimension formula
	\begin{equation}
		s_\nu(1^q)
		=
		\prod_{1\leq i<j\leq q}
		\frac{\nu_i-\nu_j+j-i}{j-i},
		\label{eq:appendix-Weyl-dimension-formula}
	\end{equation}
	where zero parts are appended to $\nu$ when necessary.

	The identity
	\eqref{eq:appendix-Weyl-dimension-formula} is the partition-theoretic
	form of the Vandermonde quotient that appears in the normalized sector
	sum.

	\subsection{The tableau square-sum identity}

	For every integer $m\geq0$,
	\begin{equation}
		\sum_{\nu\vdash m}
		(f^\nu)^2
		=
		m!.
		\label{eq:appendix-tableau-square-sum}
	\end{equation}

	A combinatorial explanation follows from the Robinson--Schensted
	correspondence. Every permutation
	\[
	\sigma\in S_m
	\]
	is placed in bijection with a pair
	\[
	(P,Q)
	\]
	of standard Young tableaux having the same shape $\nu\vdash m$.
	For a fixed shape $\nu$, there are
	\[
	(f^\nu)^2
	\]
	such pairs. Summing over all partitions of $m$ therefore counts all
	permutations in the symmetric group $S_m$:
	\[
	|S_m|
	=
	m!
	=
	\sum_{\nu\vdash m}
	(f^\nu)^2.
	\]

	\subsection{Exponential partition sum}

	Combining \eqref{eq:appendix-hook-length-formula} and
	\eqref{eq:appendix-tableau-square-sum} gives
	\begin{align}
		\sum_{\nu}
		\mu^{|\nu|}
		\left(
		\frac{f^\nu}{|\nu|!}
		\right)^2
		&=
		\sum_{m=0}^{\infty}
		\mu^m
		\sum_{\nu\vdash m}
		\left(
		\frac{f^\nu}{m!}
		\right)^2
		\nonumber\\
		&=
		\sum_{m=0}^{\infty}
		\frac{\mu^m}{m!}
		\nonumber\\
		&=
		e^\mu.
		\label{eq:appendix-exponential-partition-sum}
	\end{align}

	Formula \eqref{eq:appendix-exponential-partition-sum} is the
	combinatorial identity underlying both limits
	\[
	\mathcal T_\infty(\mu)
	=
	e^\mu
	\]
	and
	\[
	\lim_{q\to\infty}
	\mathcal R_q(\mu)
	=
	e^\mu.
	\]

	\end{document}